\documentclass[conference]{IEEEtran}

\usepackage[T1]{fontenc}
\usepackage[utf8]{inputenc}
\usepackage{cancel} 
\usepackage{yfonts} 
\usepackage{amssymb}
\usepackage{amsfonts}
\usepackage{amsthm}
\usepackage{mathtools}
\usepackage{calc}
\usepackage{hyperref}
\usepackage{graphicx}
\usepackage[all]{xy}
\usepackage{multicol}
\usepackage{dsfont} 
\usepackage{amsmath}
\usepackage{stmaryrd}
\usepackage{verbatim}
\usepackage{blkarray}
\usepackage{proof}
\usepackage{scalerel}
\DeclarePairedDelimiter\floor{\lfloor}{\rfloor}

\usepackage{subcaption}
\usepackage{tikz}
\usetikzlibrary{cd}

\usetikzlibrary{calc,shapes,arrows,positioning,automata}

\usepackage{etoolbox}

\newtoggle{extern}
\togglefalse{extern}

\DeclareUnicodeCharacter{27E9}{$\rangle$}

\usepackage{wrapfig}

\usepackage{cleveref}

\newcommand{\one}{\ensuremath{\mathds{1}}}

\newcommand{\wireSetA}{\ensuremath{\mathcal{X}}}
\newcommand{\wireSetB}{\ensuremath{\mathcal{Y}}}
\newcommand{\wireSetC}{\ensuremath{\mathcal{Z}}}
\newcommand{\wireSetD}{\ensuremath{\mathcal{W}}}

\newsavebox\MBox

\usetikzlibrary{decorations.markings,arrows,decorations.pathmorphing, decorations.pathreplacing}
\usetikzlibrary{shapes.misc,matrix,backgrounds,folding,positioning}
\usetikzlibrary{shapes.geometric}
\usetikzlibrary{shapes.multipart}
\pgfdeclarelayer{edgelayer}
\pgfdeclarelayer{nodelayer}
\pgfsetlayers{background,edgelayer,nodelayer,main}
\tikzset{every path/.style={draw=black!80, line width=0.6pt}}
\tikzstyle{every picture}=[baseline=-0.25em]
\tikzstyle{none}=[inner sep=0mm]
\tikzstyle{box}=[fill=white, draw=black, shape=rectangle]
\tikzstyle{zxnode}=[shape=circle, minimum width=.25cm, inner sep=0.5pt, font=\footnotesize, draw=black,thick]
\tikzstyle{gn}=[zxnode ,fill=green, draw=green!10!black]
\tikzstyle{rn}=[zxnode ,fill=red, draw=red!10!black]
\tikzstyle{H box}=[rectangle,fill=yellow, draw=yellow!10!black,thick,xscale=1,yscale=1,font=\footnotesize,inner sep=1.2pt,minimum width=0.15cm,minimum height=0.15cm]
\tikzstyle{ug}=[regular polygon, regular polygon sides=3, fill=red,draw=black,inner sep = 0pt,minimum width=0.8em]
\tikzstyle{black dot}=[inner sep=0.4mm,minimum width=0pt,minimum height=0pt,fill=black,draw=black,shape=circle]
\tikzstyle{dot}=[black dot]
\tikzstyle{white dot}=[dot,fill=white, inner sep=-1pt, font=\footnotesize]
\tikzstyle{arrow}=[decoration={markings,mark=at position 1 with
    {\arrow[scale=1.2,>=stealth]{>}}},postaction={decorate}]

\tikzstyle{glabel}=[rounded corners=0.2em,fill=green!30,inner sep=0.1em,font=\scriptsize, anchor=west, xshift=-0.3em, yshift=0,opacity=1]
\tikzstyle{rlabel}=[rounded corners=0.2em,fill=red!30,inner sep=0.1em,font=\scriptsize, anchor=west, xshift=-0.3em, yshift=0,opacity=1]
\tikzstyle{box}=[rectangle, draw=black, fill=white, inner sep=1pt, font=\scriptsize]
\tikzstyle{box-no-outline}=[rectangle, draw=white, fill=white, inner sep=2pt, line width=0.1pt]
\tikzstyle{circle-no-outline}=[circle, draw=white, fill=white, inner sep=0pt, line width=0.1pt]
\tikzstyle{squigglearrow}=[->, line join=round, decorate, decoration={zigzag, segment length=4, amplitude=0.8, post=lineto, post length=2pt}]
\tikzstyle{divide}=[regular polygon, regular polygon sides=3, draw=black, fill=gray!50, inner sep=1.6pt, rounded corners=0.8mm]
\tikzstyle{very thick}=[-, line width=1pt]
\tikzstyle{boxedge}=[draw=gray!50]
\tikzstyle{pbs}=[diamond, draw=black, inner sep=-0.5pt, fill=white]
\tikzstyle{ribbon}=[thick, rounded corners=0.4pt,fill={rgb,255: red,157; green,246; blue,255}, fill opacity=0.7]
\tikzset{tensor/.style={inner sep=2.5pt, draw=.!70!white, fill=.!70!white, circle, path picture={
			\draw[white]
			(path picture bounding box.south east) -- (path picture bounding box.north west) (path picture bounding box.south west) -- (path picture bounding box.north east);
}}}
\tikzset{plus/.style={inner sep=2.5pt, draw, circle, path picture={
  \draw[black]
(path picture bounding box.east) -- (path picture bounding box.west) (path picture bounding box.south) -- (path picture bounding box.north);
}}}
\tikzset{triangle/.style = {fill=white, regular polygon, regular polygon sides=3 }}
\tikzset{border rotated/.style = {shape border rotate=180}}
\tikzset{border leftrotated/.style = {shape border rotate=90}}
\tikzset{border rightrotated/.style = {shape border rotate=270}}
\tikzstyle{maclane}=[inner sep=0.4mm,minimum width=0pt,minimum height=0pt,fill=white!50!blue,draw=white!50!blue,shape=circle]
\tikzstyle{spider}=[inner sep=0.8mm,minimum width=0pt,minimum height=0pt,fill=white!50!blue,draw=black, line width=0.6mm,shape=circle]
\tikzstyle{contraction}=[circle,draw,font={\scriptsize \color{gray} c}, inner sep= 1pt]
\tikzstyle{up-contraction}=[triangle,draw,inner sep= 1.5pt]
\tikzstyle{down-contraction}=[triangle, border rotated,draw, inner sep= 1.5pt]
\tikzstyle{left-contraction}=[triangle, border leftrotated,draw, inner sep= 1.5pt]
\tikzstyle{right-contraction}=[triangle, border rightrotated,draw, inner sep= 1.5pt]
\tikzstyle{small-contraction}=[circle,draw,font={\tiny  \color{gray} c}, inner sep= 0.2pt]
\tikzstyle{vacuum}=[rounded rectangle, draw, fill=gray!50, rounded rectangle west arc=none, rotate=180]

\pgfdeclarelayer{edgelayer}
\pgfdeclarelayer{nodelayer}
\pgfsetlayers{background,edgelayer,nodelayer,main}
\tikzstyle{every loop}=[]

\usetikzlibrary{circuits.ee.IEC}

\colorlet{tapeBg}{red!30}
\colorlet{tapeBorder}{red!60}
\tikzstyle{tapeFill} = [fill=tapeBg]
\tikzstyle{tapeNoFill} = [draw=tapeBorder, line width=0.5pt]
\tikzstyle{tape} = [fill=tapeBg, draw=tapeBorder, line width=0.5pt]

\newcommand{\tikzfig}[1]{
\InputIfFileExists{./figures/#1.tikz}{}{{\color{red}\colorbox{pink}{missing file : #1}}}
}

\iftoggle{extern}{
	\usetikzlibrary{external}
	\tikzexternalize[prefix=tikzfigs/]

	\renewcommand{\tikzfig}[1]{
		\tikzsetnextfilename{#1}
		
\InputIfFileExists{./figures/#1.tikz}{}{{\color{red}\colorbox{pink}{missing file : #1}}}
}

}{
	
}

\newcommand{\interp}[1] {\left\llbracket #1 \right\rrbracket}

\newcommand{\tone}{\ensuremath{\mathds{1}}}
\newcommand{\oone}{\ensuremath{\mathds{1}}}
\DeclareMathAlphabet{\mymathbb}{U}{BOONDOX-ds}{m}{n}
\newcommand{\tzero}{\ensuremath{\mymathbb{0}}}
\newcommand{\ozero}{\ensuremath{\mymathbb{0}}}

\usetikzlibrary{cd}

\newcommand{\bfup}[1]{\textup{\textbf{#1}}}
\newcommand{\Langage}{Tensor-Plus Calculus}
\newcommand{\Cat}[1][R]{\textbf{\textup{TP}}^{#1}}
\newcommand{\FCat}[1][R]{\textbf{\textup{FTP}}^{#1}}
\newcommand{\sCat}[1][R]{\textbf{\textup{singleTP}}^{#1}}
\newcommand{\sFCat}[1][R]{\textbf{\textup{singleFTP}}^{#1}}
\newcommand{\isoML}{\approx_{\lambda\rho\alpha}}
\newcommand{\mirror}[1]{#1^{\mathbb{T}}}

\newcommand{\N}{\mathcal{N}}
\newcommand{\Morphism}[1]{\floor{#1}}
\newcommand{\IsoPlusUn}[1]{\bfup{Fun}(#1)}
\newcommand{\injection}[2]{\scalebox{1}{$%
	\begin{matrix}%
		\textcolor{gray}{%
		\begin{matrix}%
			1\\\vdots\\#1\\\vdots\\#2\\
		\end{matrix}}&%
		\hspace{-0.4cm}%
		\begin{pmatrix}%
			\vdots\\\zero\\\id\\\zero\\\vdots\\
		\end{pmatrix}
\end{matrix}$}}
\newcommand{\projection}[2]{\scalebox{1}{$%
		\begin{matrix}%
		\textcolor{gray}{%
			\begin{matrix}%
			1&&\dots&#1&\dots&&#2\\
			\end{matrix}}\\%
		\begin{pmatrix}%
		\dots&\zero&\id&\zero&\dots\\
		\end{pmatrix}\\\\
		\end{matrix}$}}

\newcommand{\id}{\bfup{id}}
\newcommand{\zero}{\bfup{zero}}
\newcommand{\dist}{\bfup{dist}}

\newcommand{\emptysquare}{\ensuremath{\varnothing}}

\DeclareMathOperator{\tensor}{
\text{\scalerel*{%
		\iftoggle{extern}{\tikzsetnextfilename{tensor}}{}
		\begin{tikzpicture}
		\node[inner sep=2pt, draw=.!70!white, circle, fill=.!70!white, path picture={
			\draw[white]
			(path picture bounding box.south east) -- (path picture bounding box.north west) (path picture bounding box.south west) -- (path picture bounding box.north east);
		}] (c) at (0,0){};
		\end{tikzpicture}}{\oplus}}}

\DeclareMathOperator{\blackpar}{
	\text{\scalerel*{%
			\iftoggle{extern}{\tikzsetnextfilename{blackpar}}{}
			\begin{tikzpicture}
			\node[inner sep=0pt, draw=.!70!white, circle, fill=.!70!white] (c) at (0,0){\textcolor{white}{\rotatebox{180}{\textbf{\&}}}};
			\end{tikzpicture}}{\oplus}}}

\DeclareMathOperator{\whitewith}{
	\text{\scalerel*{%
			\iftoggle{extern}{\tikzsetnextfilename{whitewith}}{}
			\begin{tikzpicture}
			\node[inner sep=0pt, draw=., circle] (c) at (0,0){\&};
			\end{tikzpicture}}{\oplus}}}

\DeclareMathOperator*{\bigtensor}{
	\text{\scalerel*{%
			\iftoggle{extern}{\tikzsetnextfilename{bigtensor}}{}
			\begin{tikzpicture}
			\node[inner sep=2pt, draw=.!70!white, circle, fill=.!70!white, path picture={
				\draw[white]
				(path picture bounding box.south east) -- (path picture bounding box.north west) (path picture bounding box.south west) -- (path picture bounding box.north east);
			}] (c) at (0,0){};
			\end{tikzpicture}}{\bigoplus}}}

\DeclareMathOperator{\contraction}{
\scalerel*{%
		\iftoggle{extern}{\tikzsetnextfilename{contraction}}{}
		\begin{tikzpicture}
		\node[down-contraction] (c) at (0,0){};
		\end{tikzpicture}}{\oplus}}

\DeclareMathOperator{\maclane}{
	\scalerel*{%
		\iftoggle{extern}{\tikzsetnextfilename{maclane}}{}
		\begin{tikzpicture}
		\node[maclane] (c) at (0,0){};
		\end{tikzpicture}}{\oplus}}

\DeclareMathOperator{\spider}{
		\scalerel*{%
			\iftoggle{extern}{\tikzsetnextfilename{spider}}{}
			\begin{tikzpicture}
				\node[spider] (c) at (0,0){};
		\end{tikzpicture}}{\oplus}}

\newcommand{\iso}{\operatorname{iso}}
\newcommand{\Iso}{\operatorname{Iso}}

\newcommand{\Objet}[2]{\raisebox{-5pt}{\scriptsize$#1$}\underline{#2}\raisebox{-5pt}{\scriptsize$#1$}}

\newcommand{\TextEquiv}[1]{\stackrel{\text{#1}}{\equiv}}

\newcommand{\eqTM}{\textup{($\tensor\maclane$)}}
\newcommand{\eqPM}{\textup{($\oplus\maclane$)}}
\newcommand{\eqCM}{\textup{($\contraction\!\maclane$)}}\newcommand{\eqMC}{\eqCM}
\newcommand{\eqNM}{\textup{($0\maclane$)}}\newcommand{\eqMN}{\eqNM}
\newcommand{\eqMM}{\textup{($\maclane\maclane$)}}
\newcommand{\eqM}{\textup{($\maclane$)}}
\newcommand{\eqSM}{\textup{($R\!\maclane$)}}

\newcommand{\eqAlphaT}{\textup{($\alpha\tensor$)}}
\newcommand{\eqAlphaP}{\textup{($\alpha\oplus$)}}
\newcommand{\eqAlphaC}{\textup{($\alpha\!\contraction$)}}
\newcommand{\eqSigmaC}{\textup{($\sigma\!\contraction$)}}
\newcommand{\eqT}{\textup{($\tensor$)}}
\newcommand{\eqP}{\textup{($\oplus$)}}
\newcommand{\eqN}{\textup{($0$)}}

\newcommand{\eqLambdaRhoT}{\textup{($\lambda\rho\tensor$)}}
\newcommand{\eqLambdaP}{\textup{($\lambda\oplus$)}}
\newcommand{\eqRhoP}{\textup{($\rho\oplus$)}}
\newcommand{\eqLambdaC}{\textup{[$\lambda\!\contraction$]}}
\newcommand{\eqRhoC}{\textup{($\rho\!\contraction$)}}
\newcommand{\eqTN}{\textup{($\tensor 0$)}}\newcommand{\eqNT}{\eqTN}
\newcommand{\eqPN}{\textup{($\oplus 0$)}}\newcommand{\eqNP}{\eqPN}
\newcommand{\eqCN}{\textup{($\contraction\! 0$)}}\newcommand{\eqNC}{\eqCN}
\newcommand{\eqNN}{\textup{($00$)}}

\newcommand{\eqTC}{\textup{($\tensor\!\contraction$)}}
\newcommand{\eqPC}{\textup{($\oplus\!\contraction$)}}
\newcommand{\eqCC}{\textup{($\contraction\!\contraction$)}}
\newcommand{\eqTCleft}{\textup{($\contraction\!\tensor$)}}
\newcommand{\eqPPCleft}{\textup{($X\!\oplus\!\contraction$)}}
\newcommand{\eqPPNleft}{\textup{[$X\!\oplus 0$]}}\newcommand{\eqNPPleft}{\eqPPNleft}
\newcommand{\eqPPNright}{\eqPPNleft}
\newcommand{\eqTNleft}{\textup{[$0\tensor$]}}

\newcommand{\eqTS}{\textup{($\tensor\! R$)}}\newcommand{\eqST}{\eqTS}
\newcommand{\eqPS}{\textup{[$\oplus\! R$]}}\newcommand{\eqSP}{\eqPS}
\newcommand{\eqCS}{\textup{($\contraction\! R$)}}\newcommand{\eqSC}{\eqCS}
\newcommand{\eqNS}{\textup{($0 R$)}}\newcommand{\eqSN}{\eqNS}
\newcommand{\eqSS}{\textup{[$R_\times$]}}
\newcommand{\eqS}{\textup{($R_1$)}}
\newcommand{\eqSzero}{\textup{($R_0$)}}
\newcommand{\eqSsum}{\textup{($R_+$)}}

\newcommand{\eqTT}{\textup{($N\!\tensor$)}}
\newcommand{\eqPP}{\textup{($X\!\oplus$)}}
\newcommand{\eqPtoC}{\textup{($\oplus\!\to\!\!\contraction$)}}
\newcommand{\eqmix}{\textup{(mix)}}
\newcommand{\eqbot}{\textup{($\bot$)}}
\newcommand{\eqbotC}{\textup{[$\bot\!\contraction$]}}

\newcommand{\eqLambdaT}{\textup{[$\lambda\tensor$]}}
\newcommand{\eqRhoT}{\textup{($\rho\tensor$)}}

\newcommand{\eqUN}{\textup{($10$)}}
\newcommand{\eqUS}{\textup{(Can)}}
\newcommand{\eqUsum}{\textup{[$1_+$]}}
\newcommand{\eqUpar}{\textup{[$1\parallel$]}}

\newcommand{\eqTTsimple}{\textup{[$X\!\tensor$]}}
\newcommand{\eqSigmaT}{\textup{[$\sigma\tensor$]}}
\newcommand{\eqTCright}{\eqTCleft}
\newcommand{\eqCtoP}{\textup{[$\contraction\!\!\to\!\oplus$]}}
\newcommand{\eqPPtoCC}{\textup{[$X\!\oplus\!\to\!X\!\contraction$]}}
\newcommand{\eqPPtoC}{\textup{[$X\!\oplus\!\to\!\contraction$]}}
\newcommand{\eqCCtoC}{\textup{[$X\!\contraction\!\to\!\!\contraction$]}}
\newcommand{\eqCCtoT}{\textup{[$\contraction\!\contraction\!\to\!\tensor$]}}
\newcommand{\eqPCtoT}{\textup{[$\oplus\!\contraction\!\to\!\tensor$]}}
\newcommand{\eqSigmaP}{\textup{[$\sigma\oplus$]}}
\newcommand{\eqPPCright}{\eqPPCleft}
\newcommand{\eqPPPleft}{\textup{[$X\!\oplus\oplus$]}}

\newcommand{\eqPPSleft}{\textup{[$R\!\oplus\oplus$]}}

\newcommand{\eqPCreverse}{\eqPC}

\newcommand{\eqSpider}{\textup{[$\spider\spider$]}}
\newcommand{\eqSpiderC}{\textup{[$\spider\!\contraction$]}}
\newcommand{\eqSpiderN}{\textup{[$\spider 0$]}}
\newcommand{\eqLambdaRhoSpider}{\textup{[$\lambda\rho\spider$]}}
\newcommand{\eqSpiderS}{\textup{[$\spider R$]}}



\usepackage{regexpatch}

\allowdisplaybreaks

\makeatletter
\DeclareFontEncoding{LS2}{}{\noaccents@}
\makeatother
\DeclareFontSubstitution{LS2}{stix}{m}{n}
\DeclareSymbolFont{largesymbolstix}{LS2}{stixex}{m} {n}
\DeclareMathDelimiter{\lParen}{\mathopen}{largesymbolstix}{"DE}{largesymbolstix}{"02}
\DeclareMathDelimiter{\rParen}{\mathclose}{largesymbolstix}{"DF}{largesymbolstix}{"03}

\newtheorem{theorem}{Theorem}[section]
\newtheorem{definition}[theorem]{Definition}
\newtheorem{lemma}[theorem]{Lemma}

\newtheorem{proposition}[theorem]{Proposition}

\newtheorem{example}[theorem]{Example}
\newtheorem{corollary}[theorem]{Corollary}

\def\tikzfig#1{\ensuremath{\vcenter{\hbox{\includegraphics{tikzfigs/#1.pdf}}}}}

\begin{document}

\title{The \Langage{}}

\author{
 \IEEEauthorblockN{
    Kostia Chardonnet\IEEEauthorrefmark{2},
    Marc de Visme\IEEEauthorrefmark{1},
    Benoît Valiron\IEEEauthorrefmark{1}\IEEEauthorrefmark{3} and
    Renaud Vilmart\IEEEauthorrefmark{1}
  }
  \IEEEauthorblockA{\IEEEauthorrefmark{1}
    Université Paris-Saclay, CNRS, ENS Paris-Saclay, Inria, Laboratoire Méthodes Formelles, 91190, Gif-sur-Yvette,
    France 
  }
  \IEEEauthorblockA{\IEEEauthorrefmark{2}
    Université de Lorraine, CNRS, Inria, LORIA, F-54000 Nancy, France
  }
  \IEEEauthorblockA{\IEEEauthorrefmark{3} CentraleSupélec, France
  }
}
\maketitle

\begin{abstract}
  We propose a graphical language that accommodates two monoidal
  structures: a multiplicative one for pairing and an additional one
  for branching. In this colored PROP, whether wires in parallel are
  linked through the multiplicative structure or the additive
  structure is implicit and determined contextually rather than
  explicitly through tapes, world annotations, or other techniques, as
  is usually the case in the literature. The diagrams are used as
  parameter elements of a commutative semiring, whose choice is
  determined by the kind of computation we want to model, such as
  non-deterministic, probabilistic, or quantum.

  Given such a semiring, we provide a categorical semantics of
  diagrams and show the language as universal for it. We also provide
  an equational theory to identify diagrams that share the same
  semantics and show that the theory is sound and complete and
  captures semantical equivalence.

  In categorical terms, we design an internal language for
  semiadditive categories (C,+,0) with a symmetric monoidal structure
  (C,x,1) distributive over it, and such that the homset C(1,1) is
  isomorphic to a given commutative semiring, e.g., the semiring of
  non-negative real numbers for the probabilistic case.
\end{abstract}

\section{Introduction}

From a computational perspective, the low-level control flow of a
program is inherently linked to the information it acts upon. Besides
composition, two program manipulations can be considered native:
juxtaposition and branching.

The juxtaposition of programs can happen when they operate on distinct
areas of the memory: a function $f$ acting on $A$ can be executed in
parallel to a function $g$ acting on $B$ as long as they do not
interact. Such operations can run asynchronously---they can be
juxtaposed. The categorical interpretation of such an operation is a
monoidal structure: the joint system of $A$ and $B$ is written
$A\tensor B$,%
\footnote{The combinator $\tensor$ is usually written
    $\otimes$. We chose to modify its representation to make the
    distinction with $\oplus$ clearer, which is especially beneficial
    for the upcoming generators of the graphical language.}
and the combined action of $f$ and $g$ is $f\tensor g$. The operation is
\emph{multiplicative}: $A\tensor B$ represents pairs of elements of $A$ and $B$.

These multiplicative monoidal structures can be equipped with a
natural graphical interpretation (a PROP): types are represented as
wires and actions as boxes. The system's action $f\tensor g$ on the
type $A\tensor B$ can be represented graphically as shown in
\Cref{fig:typ-graph-fig:mult}.  For this multiplicative construction,
in the interpretation discussed above, the wire $A\tensor B$ is
intuitively understood as a bundle of two wires of type $A$ and $B$:
the two boxes $\tensor$ represent the split and merge of these two
bundled wires.

\begin{figure}
  \begin{subfigure}{.3\columnwidth}
  \centering$\xymatrix@C=1.3ex@R=2ex{
    &\ar[d]^{A\tensor B}&\\
    & \tensor\ar@{-}@/_2ex/[dl]_{A}\ar@{-}@/^2ex/[dr]^B &
    \\
    *+[F]{f}\ar@/_2ex/[dr]_{A'} & & *+[F]{g}\ar@/^2ex/[dl]^{B'}
    \\
    & \tensor\ar[d]^{A'\tensor B'}\\
    &&
  }$
  \caption{Multiplicative}\label{fig:typ-graph-fig:mult}
\end{subfigure}
\hfill
  \begin{subfigure}{.3\columnwidth}
    \centering$  \xymatrix@C=1.3ex@R=2ex{
    &\ar[d]^{A\oplus B}&\\
    & \oplus\ar@{-}@/_2ex/[dl]_{A}\ar@{-}@/^2ex/[dr]^B &
    \\
    *+[F]{f}\ar@/_2ex/[dr]_{A'} & & *+[F]{g}\ar@/^2ex/[dl]^{B'}
    \\
    & \oplus\ar[d]^{A'\oplus B'}\\
    &&
  }$
  \caption{Additive}\label{fig:typ-graph-fig:add}
  \end{subfigure}
\hfill
  \begin{subfigure}{.3\columnwidth}
  \centering$\xymatrix@C=1.3ex@R=2.3ex{
    &\ar[d]^{A\tensor B}&\\
    & \tensor\ar@{-}@/_1.2ex/[dl]_{A}\ar@{-}@/^1.2ex/[dr]^B &
    \\
    *{}\ar@/_1.2ex/[dr]_{A} & & *{}\ar@/^1.2ex/[dl]^{B}
    \\
    & \oplus \ar[d]^{A\oplus B}\\
    &&
  }$
    \caption{Incompatibility}\label{fig:typ-graph-fig:incompat}
  \end{subfigure}
  \caption{typical graphical structure}
\end{figure}
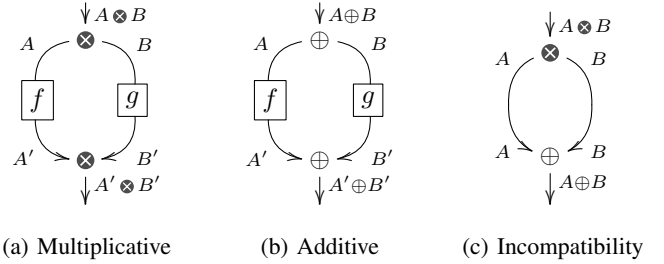

The other arguably native operation is branching: the possibility to
choose a course of action based on the state of the memory. This
choice operation becomes the if-then-else for Boolean values and the
case distinction for more general data structures. From a categorical
standpoint, the choice is typically modeled using a coproduct or a
biproduct $\oplus$. A Boolean value lives in $1\oplus 1$: it is either
the left injection (say, False) or the right injection (say, True). In
general, the type $A\oplus B$ stands for either something of type $A$
or something of type $B$. The additive monoidal structure is
graphically similar to the multiplicative one, shown in
\Cref{fig:typ-graph-fig:add}. The main difference lies in the
interpretation of the wire of type $A\oplus B$: something of type $A$
or type $B$ will get in, and this piece of data will be routed to the
correct wire $A$ or $B$ upon reaching the first split-node
$\oplus$. The merge-node $\oplus$ will then secure the result again in
the wire $A'\oplus B'$.

These multiplicative and additive structures are pervasive in all
models of computation. They serve as a baseline for the
multiplicative, additive fragment of linear logic~\cite{linearlogic,
  proofnets}. Dozens of categories feature these structures:
relational models~\cite{lamarche1992quantitative, laird2013weighted,
pagani2014applying, finiteness},
coherent spaces~\cite{ehrhard2014probabilistic}, vectorial models such
as modules or vector spaces~\cite{Asada-ll, diazcaroMalherbe}, etc.

Despite this ubiquity, it is notoriously difficult to define a
graphical language that handles both a multiplicative, monoidal
structure and an additive structure, whether coproduct or
biproduct. The main difficulty lies in keeping track of how wires are
related: if the language is purely multiplicative, all wires are
tensored; if the language is purely additive, all wires are in
sum. However, if the language features both tensor and sum, three
wires of type $A$, $B$, and $C$ are ambiguous: they could, for
instance, be seen as $A\tensor(B\oplus C)$, but also as the
(incompatible) $(A\tensor B)\oplus C$ or even
$A \oplus (B \tensor C)$. In the literature, this has been approached
with the addition of external information on the graph representing
the computation: worlds~\cite{manyworlds},
sheets~\cite{comfort2020sheet} or
tapes~\cite{Bonchi2022deconstructing}. Considering our three wires of
types $A$, $B$, and $C$, additional information is added to know how
to relate them: for instance, whether $A$ and $B$ should first be
tensored before summing $C$.

Multiplicative and additive structures canonically support algebraic
effects. Relational models can be weighted with
non-deterministic~\cite{relNonDet} or probabilistic
effects~\cite{laird2013weighted}, and finite-dimensional vector spaces
is a model for (pure) quantum
computation~\cite{pagani2014applying}. In this realm, the ``choice''
operation becomes effectful, and wires carry a weight: a probability,
a complex number representing a quantum coefficient, etc. Such
effectful computations take an additional toll on the design of a
graphical language: one has to handle the multiplicative and additive
aspects and the actions of the algebraic effects.

This paper is devoted to studying such a setting: We propose a
graphical language unifying multiplicative and additive actions,
supporting probabilistic, non-deterministic, and more exotic effects
such as purely quantum effects. Contrary to the above-mentioned
worlds, sheets, or tapes, the handling of wires does not require any
additional structure.

\subsection{Strategy followed in the paper.}
The main difficulty consists in giving a meaning to the juxtaposition of
two wires: it cannot simply be a tensor or a sum; it needs to be able
to be both and in a homogeneous manner.

Our proposal, therefore, defines a special monoidal structure
capturing these two possibilities at once. Formally, we set ourselves
in the context of a symmetric monoidal category (representing the
multiplicative tensor) with an additive structure: the additive
structure gives a biproduct, and when enriched this framework is
expressive enough to represent the algebraic effects we care
about. Wire juxtaposition is then represented with a sum of either a
tensor or a sum: $(A\tensor B)\oplus(A\oplus B)$. If this is defined
formally in \Cref{sec:cat_sem}, let us simply mention here that this
binary operation gives a well-defined monoidal structure and a PROP:
we provide in the paper a graphical language for which this category
is the target model.

Moreover, because these two wires might be used inconsistently, for
instance, the diagram shown in \Cref{fig:typ-graph-fig:incompat}: we
need to be able to represent a notion of ``error'' state. We
capitalize on the additive structure (and the enrichment) and
represent it with the ``zero'' map inherited from the additive
structure of the category.

\subsection{Contributions and plan of the paper.}
Formally, the contributions of the paper are as follows.
\begin{itemize}
\item A graphical language: the \Langage{}, unifying multiplicative
  and additive structures and capturing algebraic effects:
  non-determinism, probability, vectorial (and quantum). The language
  is described using the notion of PROP, and comes in two variants:
  $\Cat$ and $\FCat$, with $R$ a given commutative semiring.
\item An interpretation based on symmetric, monoidal, semiadditive
  categories, with a universality result for both
  (\Cref{thm:universality});
\item An equational theory allowing us to rewrite diagrams, proven
  sound and complete (\Cref{prop:soundness_fun} and
  \Cref{thm:complete_fun}). We write $\Cat_{\equiv}$ and
  $\FCat_{\equiv}$ for the categories $\Cat$ and $\FCat$ quotiented by
  the corresponding equivalence (rewriting)
  relation. From the completeness result, we show that
  the \Langage{} can be regarded as an internal language for the
  corresponding category (\Cref{thm:internal_language}).
\end{itemize}
The plan of the paper is as follows. We present the language and
several examples in \Cref{sec:many-worlds}. We then define the
categorical semantics of our language and show its universality in
\Cref{sec:cat_sem}. \Cref{sec:equations} is devoted to presenting a
sound and complete equational theory. We then discuss extensions of
these results.

\section{The \Langage{}}
\label{sec:many-worlds}

The aim of the paper is to define a graphical language that allows
both for pairing and branching of data. In the latter case, the two
branches may not be used at the same time (we use either one branch or
the other, or none), while in the pairing case, the two pieces of data
are either used together, or are both not used. We want all these
possible interactions between wires to remain implicit, and not resort
to explicit additional cues, like sheets, tapes, annotations, etc.
Moreover, we want to be able to account for algebraic effects such as
non-deterministic, probabilistic, and quantum effects (or, more
generally, vectorial effects).

The language we propose, called the \Langage{}
is parameterized by a commutative semiring $(R, +, 0,
\times, 1)$ to account for the kind of algebraic effects we target.
It can be instantiated by the complex numbers
$(\mathbb{C}, +, 0, \times, 1)$ to represent pure quantum
computations, the non-negative real numbers $(\mathbb{R}_{\geq 0}, +,
0,\times, 1)$ for probabilistic computations, or the Boolean values
$(\{0,1\}, \vee, 0,
\wedge, 1)$ for non-deterministic computations.
We come back to these semirings in \Cref{sec:examples}.

We define the \Langage{} within the paradigm of colored
PROP \cite{Carette21,HackneyR15}: A morphism is a diagram composed of
nodes, called \emph{generators}, linked to each other through
\emph{colored wires}---where each \emph{color} corresponds to a
datatype such as ``bit''---that are allowed to cross each other. This
graphical representation allows the rewriting equations provided in
\Cref{fig:PROP}, where $f$ and $g$ stand for any diagrams with any
numbers of inputs and outputs, and the wires could be of any colors.

The \Langage{} is equipped with a denotational semantics
  (\Cref{sec:cat_sem}) and an equational theory
  (\Cref{sec:equations}). In this section, we only focus on the syntax and
  the intuition through a series of examples.

\begin{figure*}
	\centering
	\tikzfig{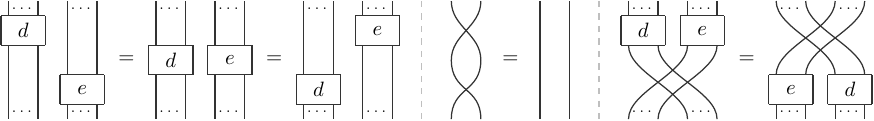}
	\caption{Equations of a colored PROP. These are supposed true for any colors on the wires.}
	\label{fig:PROP}
\end{figure*}

\subsection{The Wires: the Objects of our Category}
\label{sec:objects_category_tensor_plus}

In a colored PROP, the set of objects is freely generated by a set of
\emph{colors} and a strict\footnote{So $\parallel$ is strictly
associative and $\varnothing$ is strictly its neutral element.}
monoidal product $\parallel$ with $\varnothing$ as neutral element. In other
words, an object corresponds to a collection of wires in parallel,
written $A_1 \parallel \dots \parallel A_n$, each wire being ``typed''
by a color $A_i$. We use calligraphic letters $\wireSetA,
\wireSetB, \dots$ to denote the objects of our colored PROP, and
non-calligraphic letters $A, B, \dots$ for colors. The goal is for the colors to represent datatypes such as ``bit'', so the
colors will themselves be generated by the syntax
\[
  A, B\quad ::=\quad \tzero \mid \tone \mid (A\oplus B) \mid (A\tensor B),
\]
where $\oplus$ is used to build sum types such as ``bit =
$\tone \oplus \tone$'', and $\tensor$ is used to build pairs. The symbol $\tzero$ represents the unit of the sum $\oplus$ while $\tone$ is the unit of the tensor $\tensor$.
For example,
\[
  ((A \oplus B) \oplus \tzero) \parallel \tone \parallel (A \tensor \tzero)
\]
is an
object composed of three colors: $(A \oplus B) \oplus \tzero$, $\tone$ and
$A \tensor \tzero$. We can then represent them as three wires in parallel:

\begin{center}
	\tikzfig{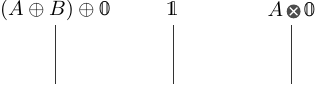}
\end{center}

\noindent
Note that while the monoidal product
$\parallel$ is strictly associative, we do not
consider $\oplus$ and $\tensor$ to be associative as they are purely
syntactical. We instead have an equivalence relation $\isoML$ on
colors, generated by the rules:
\begin{align*}
  A \oplus \tzero &{}\isoML A,
  &
    A \tensor \tone &{}\isoML A,
  \\
  \tzero \oplus A &{}\isoML A,
  &
    \tone \tensor A &{}\isoML A,
\end{align*}
for the units,
\begin{align*}
  (A \oplus B) \oplus C &{}\isoML A \oplus (B \oplus C),
  \\
  (A \tensor B) \tensor C &{}\isoML A \tensor{} (B \tensor C),
\end{align*}
for the associativity of the binary operators, and
\begin{align*}
  A\oplus C &{}\isoML B \oplus C,
  &
    A \tensor C &{}\isoML B \tensor C,
  \\
  C \oplus A &{}\isoML C \oplus B,
  &
    C \tensor A &{}\isoML C \tensor B,
\end{align*}
whenever $A \isoML B$, for the congruence rules.

The choice of the notation $\parallel$ for wires in parallel is
uncommon: the notation $\otimes$ is often preferred in the literature.
We use it to put an emphasis on the fact that contrary to languages
such as proof-nets~\cite{proofnets}, Boolean (or quantum)
circuits, or the ZX-calculus~\cite{zxorigin}, wires that are in
parallel are not necessarily ``in tensor with one another''. In fact,
$A \parallel B$ can be understood semantically as ``either $A \tensor
B$ or $A \oplus B$''. This intepretation is formalized in the
categorical semantics in~\Cref{sec:cat-framework}, and it will
be the underlying meaning of the equation \eqmix{} of
\Cref{fig:eq_tensor_plus}.

\subsection{The Diagrams: the Morphisms of our Category}

The generators of our language are described in \Cref{fig:generator}.
On the top line, one can respectively find:
the Identity $\id_A : A \to A$,
the Swap
$\sigma_{A,B} : A \parallel B \to B \parallel A$,
the Tensor
$\tensor_{A,B} : A \parallel B \to (A \tensor B)$,
the Plus
$\oplus_{A,B} : A \parallel B \to (A \oplus B)$,
the Contraction
$\contraction_A: A \parallel A \to A$,
the Null $\zero_A : \varnothing
\to A$ (represented with a triangle pointing down),
the Unit $\tensor_{\tone} : \varnothing \to \tone$.
On the bottom line, one can read
the Adapter
$\maclane_{A,A'} : A \to A'$ (whenever $A \isoML A'$),
the Scalar
$[s]_A : A \to A$ for $s$ ranging over the commutative semiring $R$,
and the up-down mirrored version of all of the generators\footnote{The Identity, the Swap, the
  Adapter and the Scalar are their own mirrored version.}.
We
write $\mirror{f}$ for the up-down mirrored version of $f$. Diagrams are read
top-to-bottom: the top-most wires are the \textit{input} wires and the
bottom-most wires are the \textit{output} wires.

\begin{figure}[!ht]
	\begin{tabular}{ccccccc}
		&&&&&& Not Functional\\
		$\tikzfig{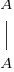}$ & $\tikzfig{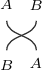} $&
		$\tikzfig{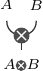}$  & \quad$ \tikzfig{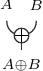}$ &
		$\tikzfig{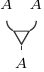}$ & $\tikzfig{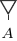} $& $\tikzfig{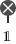}$\\[1cm]
		$\tikzfig{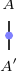}$& $\tikzfig{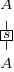}$ &
		$\tikzfig{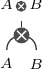}$ &\quad$ \tikzfig{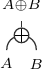}$ &
		$\tikzfig{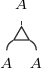}$ & $\tikzfig{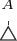} $& $\tikzfig{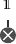}$
	\end{tabular}
	\caption{Generators of our Language ($A \isoML A'$, $s \in R$)}
	\label{fig:generator}
\end{figure}

These generators encompass at once:
  \begin{itemize}
  \item the PROP structure of the \Langage{} (with the Identity and the Swap);
  \item the multiplicative, tensor structure (with the Tensor, the Unit and their duals);
  \item the additive, biproduct structure (with the Plus, the Contraction and their duals);
  \item the scalars (with the Scalar generator).
  \end{itemize}
  The associativity and unit of the multiplicative and additive structures is given by the Adapter
  (representing the equivalence $\isoML$).

Equationally, the language as it stands does not contain more than
  the equational theory of PROP.  In particular, with respect to the Sum,
  the Contraction corresponds to the codiagonal and its dual to the
  diagonal morphism: this is not yet enforced. In \Cref{sec:cat_sem}
  we present a categorical semantics formally describing this
  intuition, and \Cref{sec:equations} equips the \Langage{} with an
  equivalence relation $\equiv$ capturing the structure of the
  denotational semantics.

Diagrams are obtained by composing the generators (\Cref{fig:generator})
in parallel (written $\parallel$), or
sequentially (written $\circ$), as follows:
\[ e\circ d:=\tikzfig{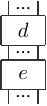}\qquad\qquad d\parallel
e:=\tikzfig{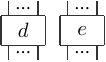} \]
Sequential composition requires the
color (and number) of wires to match.  We write $\Cat$~for the
PROP category of such diagrams (with the equality given by the equational theory of PROP), and
$\Cat(\wireSetA,\wireSetB)$ for the set of diagrams that are
morphisms from $\wireSetA$ to $\wireSetB$.

The goal we announced in the introduction was to create an internal language for semiadditive categories with some additional properties. As such, one would expect $\varnothing$ to be a terminal and initial object of our category, and might question the existence of the Unit as a non-trivial morphism from $\varnothing$ to $\tone$.
The Unit and its mirrored version are considered ``non-functional''.
Computationally, they allow the process to start generating non-zero
outputs without having received any input. More generally, while
useful for practical examples and when representing ``values'' (true, false) instead of ``functions'' (negation, etc), they come with some significant
technical complexities, and they muddy the categorical properties of our language. As such, we define the sub-language of
\textbf{Functional} \Langage{} as the same but without the Unit and
its mirrored version. We write $\FCat$ for the corresponding category.

\subsection{Examples}
\label{sec:examples}

To illustrate the expressive power of our diagrams, we give one example for each of the commutative semirings that were announced at the beginning of the section.
\begin{example}\label{ex:OR}\rm
We consider $R$ to be the boolean ring, and look at the type $\tone \oplus \tone$. There are four different values over this type: True, False, $\bot$ (or Failure), $\top$ (or the Nondeterministic superposition of True and False). One can represent them as
\[\tikzfig{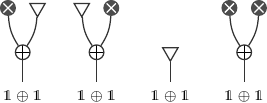}\]

We can define several versions of the logical OR: $(\tone\oplus\tone)\parallel(\tone\oplus\tone)\to\tone\oplus\tone$, respectively the strict OR, the lazy OR, and the parallel OR:
\[\tikzfig{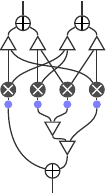}\qquad\quad
\tikzfig{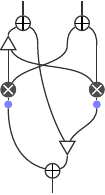}\qquad\quad
\tikzfig{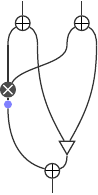}\]
All these versions of the OR produce the same output when they receive well-defined Boolean inputs. Let us give some intuition to better understand those pictures:
\begin{itemize}
	\item The left branch of a Plus corresponds to False, while its right branch corresponds to True.
	\item The Contractions should be seen as rail-switches where data must go left or right and cannot do both.
	\item The Tensor should be thought of as grouping data together.
	\item The Adapter can be ignored, its effect is purely administrative.
\end{itemize}
In all three diagrams, the right branches of the input Pluses are connected in many different ways to the right branch of the output Plus, and never to the left branch. As such, as soon as one input is True, the output cannot be False. The other way around, the left branch of the output Plus is only connected to the left branches of the input Pluses, so for the output to be False, both inputs must be False.

These three versions of the OR differ in the way they deal with failing inputs -- represented by the Null $\zero_{\tone \oplus \tone}$.
The strict OR expects a well-defined Boolean value for both its inputs and will return Null if any of its two input is Null, while the lazy OR can ``short-circuit'' the computation if the first input turns out to be True -- in that case, we don't care about the second input, and directly output True even if that second output would be Null. The parallel OR symmetrizes the lazy OR, and produces True whenever one of the two inputs is True.

These behaviors are completely captured by the upcoming equational theory, and examples of how it can be used to verify said behaviors are provided in \Cref{sec:equations}.
\end{example}

\begin{example}\label{ex:proba_mat}\rm
	When considering $R = \mathbb{R}_{\geq 0}$, we can encode some basic probabilistic primitives in it and
	show how they operate. The most basic data we can represent is a
	probabilistic bit (pbit), seen as a vector $\begin{pmatrix} p \\ q
	\end{pmatrix}$ where $p$ is the probability of False and $q$
	the probability of True -- we do not enforce $q=1-p$ here, although it is required for this to be an actual pbit.

	In $\Cat$, the Boolean values are
	represented by the type $\one\oplus\one$. The above pbit is then represented by the following diagram:
	\[
	\begin{pmatrix} p \\ q \end{pmatrix} \rightsquigarrow~ \tikzfig{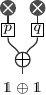} \in \Cat[\mathbb{R}_{\geq 0}](\emptysquare,\tone \oplus \tone) \]
	 A program
	manipulating pbits can be understood as a non-negative-valued (usually stochastic) matrix. For
	example, the matrix $\begin{pmatrix} 1 & 1/2 \\
		0 & 1/2 \end{pmatrix}$ corresponds to the program ``if
	$x$ then coin() else False''. This operation can be represented
	as follows:
	\[
	\begin{pmatrix}
	1 & 1/2 \\0 & 1/2
	\end{pmatrix} \rightsquigarrow~ \tikzfig{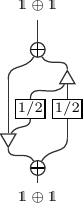} \in \Cat[\mathbb{R}_{\geq 0}](\tone \oplus \tone,\tone \oplus \tone)\]

	In that figure, the
	top-most $\oplus$ allows us to ``open a vector'' (the input) to
	recover its corresponding scalars, the contractions allow us to
	duplicate and sum scalars. Finally, the bottom-most $\oplus$ will
	build a new vector from two scalars. For example, one can look at the composition of the above matrix and vector. Using the upcoming equational theory (\Cref{sec:equations}), we can rewrite the composite diagram into a simpler diagram (done in \Cref{ex:revisiting-proba_mat}):
	\[\tikzfig{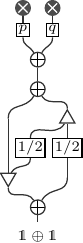} \qquad \equiv \qquad \tikzfig{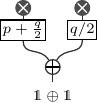}\]
	and this resulting diagram indeed corresponds to the result of the matrix product $\begin{pmatrix} p + \frac{q}{2} \\ q / 2 \end{pmatrix}$.
\end{example}

\begin{example}\label{ex:quantum_mat}\rm
	We consider $R = \mathbb{C}$. The situation here is very similar to the probabilistic case, as our diagrams represent matrices. The main difference is that instead of expecting probabilistic bits to satisfy $p+q = 1$, we expect the quantum bits to satisfy $|p|^2 + |q|^2 = 1$, where $p,q$ are now complex numbers. The presence of negative numbers means that interference is now a possibility: multiple executions of the programs can cancel each other. For example, applying the Hadamard unitary $\begin{pmatrix} 1/\sqrt{2} & 1/\sqrt{2} \\ 1/\sqrt{2} & -1/\sqrt{2} \end{pmatrix}$ to $\begin{pmatrix} 1 / \sqrt{2} \\ 1 / \sqrt{2} \end{pmatrix}$ returns False, because the two executions outputting True cancel each other. This interference will also be visible in the upcoming equational theory, where two opposite scalars will cancel each others with Equation~\eqSsum.
\end{example}

\subsection{The Compact Closure}
\begin{figure*}
	\[\tikzfig{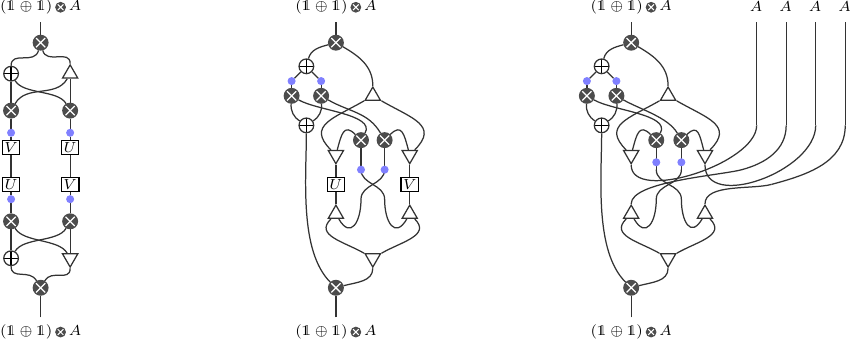}\]
	\caption{The diagrams for $\texttt{switch}_{\texttt{U},\texttt{V}}$ with duplicates, $\texttt{switch}_{\texttt{U},\texttt{V}}$ with single instances, and $\texttt{switch}$ respectively.}
	\label{fig:quantum_switch}
\end{figure*}
	\label{rem:cupcap}
	While we do not have a Cup and a Cap as generators in order to
	``bend'' wires, we can define them inductively as follows:
	\begin{center}
		\begin{tabular}{cc}
			$\tikzfig{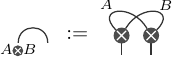} $ & $	\tikzfig{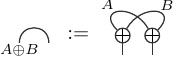}$ \\&\\
			$\tikzfig{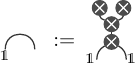}$ & $\tikzfig{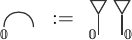}$ \\
		\end{tabular}
	\end{center}
	Since they rely on Unit, they are not part of $\FCat$. As proved
	in \Cref{prop:compact_close}, the equational theory ensures that they satisfy the
	snake equations. In other words, $\Cat_{\equiv}$ is compact
	closed.
	Those bent wires can be used to represent relatively advanced control flows and higher order programs, as shown in the example below.

\begin{example}\label{ex:switch}\rm
	We take inspiration from the literature of quantum computation and study the ``(quantum) switch'' \cite{Chiribella_2013,PhysRevLett.113.250402,
		Taddei_2021}. Note that despite being quantum-inspired, this example works for any commutative semiring $R$. It comes from a physical protocol allowing physicists to apply two operators $U$ and $V$ in an indefinite causal order, that is either $U\circ V$ or $V \circ U$, while relying on only one physical instance of $U$ and one physical instance of $V$ \cite{Abbott_2020}.

	As such, representing it in our language is an interesting case study of both ``higher order'' protocols, as $U$ and $V$ are inputs of the programs, and of advanced control flow, as it is indefinite whether $U$ is executed first or second. We temporarily set aside the ``higher order'' part, and focus on the latter. That is, we start by implementing $\texttt{switch}_{\texttt{U},\texttt{V}}:(\tone\oplus\tone)\tensor
	A \to (\tone\oplus\tone)\tensor A$, assuming two given diagrams for
	$\texttt{U},\texttt{V}:A\to A$. The expected operational behavior is the following:\begin{align*}
	\texttt{switch}_{\texttt{U},\texttt{V}} \texttt{ b a = if b }&\texttt{then (b, VUa)}\\ &\texttt{else (b, UVa)}
	\end{align*}
	In $\Cat{}$, we can easily represent the $\texttt{switch}_{\texttt{U},\texttt{V}}$ using two instances of $U$ and $V$, it is the left side of \Cref{fig:quantum_switch}. The diagram is a simple branch-out depending on the value of $b$, where we either apply $U \circ V$ or $V \circ U$. Finding a graphical representation using only on instance of $U$ and $V$ is harder -- and in fact impossible in some graphical languages \cite{Chiribella_2013} -- but still doable in our language, see the middle of \Cref{fig:quantum_switch}.

	From this diagram with only a single instance of $U$ and $V$, we can easily extract those out of the diagram by bending wires. This yields the diagram at the right of \Cref{fig:quantum_switch}, which represents
	\[ \texttt{switch} : (\tone \oplus \tone) \tensor A \parallel (A \parallel A) \parallel (A \parallel A) \to (\tone \oplus \tone) \tensor A \]
	where the first $(A \parallel A)$ expects $U$ as an input and the second $(A \parallel A)$ expects $V$.
\end{example}

Now that the expressive power of the graphical language is well illustrated, let's provide diagrams with a formal meaning.

\section{Categorical Semantics}
\label{sec:cat_sem}

We provide a semantics based on category theory, although this semantics could also be stated in terms of matrices instead\footnote{We provide such a matrix-based semantics in \Cref{fig:mat_sem_fun,fig:mat_sem}.}.
For example, in the case where the scalars of $R$ are complex numbers, this is simply a semantics into matrices with complex coefficients, so toward $(\bfup{FdHilb},\oplus,\{0\},\otimes,\mathbb{C})$ where
$\bfup{FdHilb}$ is the category of finite dimensional Hilbert spaces,
$\oplus$ is the direct sum (also called Cartesian product), $\{0\}$ is
the trivial Hilbert space, $\otimes$ is the tensor product (also
called Kronecker product), and $\mathbb{C}$ is the field of complex
numbers. The categorical framework used here can be deduced from three structures:
\begin{itemize}
	\item a semiadditive\footnote{Semiadditive categories are also called ``categories with finite biproducts''.} category $(\bfup{H},\oplus,\ozero)$
	\item a symmetric monoidal structure $(\bfup{H},\tensor,\oone)$ that is distributive over the semiadditive structure.
	\item a semiring isomorphism $\bfup{scal} : R \to \bfup{H}(\oone,\oone)$.
\end{itemize}
As those notions are relatively standard, in the core of the paper, we simply present the parts that are of interest for the categorical semantics. We still provide the full definitions in the appendix (\Cref{app:category}), together with proofs of the properties we rely on.

\subsection{The Categorical Framework}\label{sec:cat-framework}

Let $\bfup{H}$ be a category verifying the above three items. $\bfup{H}$ is enriched over $R$-semimodules, in other words given two morphisms $f,g : H \to K$ and two scalars $s,t \in R$ we can build the weighted sum $s \cdot f \bfup{~+~} t \cdot g$ by relying on the isomorphism $\bfup{scal}$. In particular, whenever $R = \mathbb{R}_{\geq 0}$, this can be used to represent probabilistic distribution morphisms.
We write $\zero : H \to K$ for the null morphism, which is the unit of that sum.

Additionally, $\bfup{H}$ comes with three different symmetric monoidal structures, $\tensor$ representing the "pairing" or data, $\oplus$ representing the "superposition" of multiple potential outcomes, and $|$ informally representing  "either $\tensor$ or $\oplus$" and more formally defined as:
\[ H \mid K = (H \tensor K) \oplus (H \oplus K) \]
Those operations are bifunctorial. So for $f : H \to K$ and $f' : H' \to K'$, we have:
	\[ \begin{array}{ccrcl} f \tensor f' &:& H \tensor H' &\to& K \tensor K' \\ f \oplus f' &:& H \oplus H' &\to& K \oplus  K'  \\ f \mid f' &:& H \mid H' &\to& K \mid K' \end{array}  \]
Those operations are, up to isomorphism, associative and respectively have $\oone$, $\ozero$ and $\ozero$ as units. We denote by $m^{\tensor}$ all the isomorphisms obtained from the associators and unitors of $\tensor$ by composing them with $\circ$, $\tensor$, $\oplus$ and $|$. While this notation is ambiguous, Mac Lane's coherence theorem (\Cref{thm:maclane_monoidal} and \cite{Kelly64,Joyal93}) ensures that this ambiguity is never problematic. We define $m^{\oplus}$ and $m^|$ similarly.

Those operations are also symmetric, meaning that we have
\[ \begin{array}{ccrcl} \sigma^{\tensor} &:& H \tensor K &\to& K \tensor H \\ \sigma^{\oplus} &:& H \oplus K &\to& K \oplus H\\ \sigma^{|} &:& H \mid K &\to& K \mid H \end{array}  \]

Lastly, $\oplus$ is a biproduct, meaning that we have injections $\iota$, projections $\pi$, and diagonals $\Delta=\iota_\ell\bfup{+}\iota_r$ and codiagonals $\nabla = \pi_\ell\bfup{+}\pi_r$:
\[\begin{array}{rcrclcccrcl}
	\pi_\ell &:& H \oplus K &\to& H &&
	\iota_\ell &:& H &\to& H \oplus K \\
	\pi_r &:& H \oplus K &\to& K &&
	\iota_r &:& K &\to& H \oplus K \\
	\nabla &:& H \oplus H &\to& H &&
	\Delta &:& H &\to& H \oplus H \\
	\end{array}\]

\subsection{Additional Structures for Non Functional Morphisms}

By virtue of $(\bfup{H},\oplus,\ozero)$ being a semiadditive category, $\ozero$ is an initial object. Said otherwise, for any object $H$, the only morphism $\ozero \to H$ is $\zero$. This actually matches the functional fragment of our language, as the only generator with no input is the Null. However, in the full language, the Unit is a generator with no input that is distinct from the Null. In order to represent that unit, we need to go slightly beyond $\bfup{H}$.

\begin{figure*}
	\[ \interp{\tikzfig{lang/id}} := \id : \interp{A} \to \interp{A} \quad \interp{\tikzfig{lang/swap}} := \sigma^| : \interp{A}\mid\interp{B} \to \interp{B}\mid\interp{A} \quad \interp{\tikzfig{lang/scal}} := s \cdot \id : \interp{A} \to \interp{A} \]
	\[ \interp{\tikzfig{lang/plus}} := \pi_r : \interp{A}\mid\interp{B} \to \interp{A}\oplus\interp{B} \quad  \interp{\tikzfig{lang/plus-inv}} := \iota_r :  \interp{A}\oplus\interp{B} \to \interp{A}\mid\interp{B}\]
	\[ \interp{\tikzfig{lang/tensor-PN}} := \pi_\ell : \interp{A}\mid\interp{B} \to \interp{A}\tensor\interp{B} \quad  \interp{\tikzfig{lang/tensor-PN-inv}} := \iota_\ell :  \interp{A}\tensor\interp{B} \to \interp{A}\mid\interp{B}\]
	\[ \interp{\tikzfig{lang/contraction}} := \nabla \circ \pi_r : \interp{A}\mid\interp{A} \to \interp{A} \quad  \interp{\tikzfig{lang/contraction-inv}} := \iota_r \circ \Delta :  \interp{A} \to \interp{A}\mid\interp{A}\]
	\[ \interp{\tikzfig{lang/null}} := \zero : \ozero \to \interp{A} \quad  \interp{\tikzfig{lang/null-inv}} := \zero :  \interp{A} \to \ozero \quad \interp{\tikzfig{lang/adapt}} := \interp{A \isoML A'} :  \interp{A} \to \interp{A'}\]
	\[ \interp{e \circ d} := \interp{e} \circ \interp{d} \qquad \interp{d \parallel e} := m^| \circ (\interp{d} \mid \interp{e}) \circ m^| \]
	\caption{Categorical Semantics for the Functional \Langage}
	\label{fig:sem_fun}
\end{figure*}

\begin{figure*}
	\[ \interp{d}^{\oplus \oone} := \interp{d} \oplus \id_{\oone} \text{ whenever } d \in \FCat \]
	\[ \interp{\tikzfig{lang/unit}}^{\oplus \oone}  := \Delta \circ m^\oplus \in \bfup{H}^{\oplus \oone}(\ozero,\oone) \quad  \interp{\tikzfig{lang/unit-inv}}^{\oplus \oone}  :=  {m^\oplus}\circ \nabla \in \bfup{H}^{\oplus \oone}(\oone,\ozero)\]
	\[ \interp{e \circ d}^{\oplus \oone}  := \interp{e}^{\oplus \oone}  \circ \interp{d}^{\oplus \oone}  \qquad \interp{d \parallel e}^{\oplus \oone}  := (m^| \oplus \id_{\oone}) \circ \bfup{expand} \circ (\interp{d}^{\oplus \oone}  \tensor \interp{e}^{\oplus \oone}) \circ \bfup{expand}^{-1} \circ (m^| \oplus \id_{\oone}) \]
	\caption{Categorical Semantics for the \Langage}
	\label{fig:sem}
\end{figure*}

\begin{definition}
	We define $\bfup{H}^{\oplus \oone}$ as the category with the same objects as $\bfup{H}$ and for morphisms $f \in \bfup{H}^{\oplus \oone}(H,K)$, the morphisms of $\bfup{H}(H\oplus \oone,K \oplus \oone)$ of the form
	\[ f = g \bfup{~+~} (\iota_r \circ \pi_r)\]
	for some $g \in \bfup{H}(H\oplus \oone,K \oplus \oone)$. In matrix terms, it is a matrix of $\bfup{H}(H\oplus \oone,K \oplus \oone)$ where the bottom-right coefficient is of the form $c \bfup{+} \id$ for $c \in \bfup{H}(\oone,\oone)$. The identity and composition are the same as in $\bfup{H}$.
\end{definition}
The restriction to morphisms of the shape $g \bfup{+} (\iota_r \circ \pi_r)$ is irrelevant if we have access to negative numbers, that is if our semiring $R$ is actually a ring. However, in absence of negative numbers, \Langage{} has no diagram that actually behaves like the $\zero$ morphism\footnote{For example, the empty diagram has for semantics that identity of $\bfup{H}(1,1)$, so to obtain $\zero$ one would need to be "less" than empty.}, meaning that we need to exclude the $\zero$ morphism from $\bfup{H}^{\oplus \oone}$ in order to get universality.

In order to utilize this category, we rely on the following natural isomorphims:
\[\bfup{expand} :  (H\oplus \oone) \tensor{} (K \oplus \oone) \to (H\mid K) \oplus \oone\]
which follow from the distributivity of $\tensor$ over $\oplus$.

\subsection{The Semantics}

We assume that we have $(\bfup{H},\oplus,\ozero,\tensor,\oone,\bfup{scal})$ as
described above. For any color $A$ of $\Cat$ (or $\FCat$), we can give its semantics
$\interp{A}$ in $\bfup{H}$ as follows:
\[ \interp{A\oplus B} := \interp{A} \oplus \interp{B} \quad
\interp{\tzero} := \ozero \]
\[ \interp{A\tensor B} := \interp{A} \tensor \interp{B} \quad
\interp{\tone} := \oone \] We can then extend this semantics to objects as
follows:
\[ \begin{matrix} \interp{\wireSetA \parallel B} :=
\interp{\wireSetA} \mid \interp{B}\qquad\qquad\quad \\=
(\interp{\wireSetA}  \tensor \interp{B}) \oplus
(\interp{\wireSetA}  \oplus \interp{B}) \end{matrix} \qquad\qquad
\interp{\varnothing} := \ozero \] In order to give a semantics to every
diagram, we will start by defining a semantics for the functional
fragment $\interp{\_}: d \in \FCat(\wireSetA,\wireSetB) \mapsto
\interp{d} \in \bfup{H}(\interp{\wireSetA},\interp{\wireSetB})$
and then we generalize it to the whole calculus at the cost of a
slightly different target category $\interp{\_}^{\oplus \oone}: d \in
\Cat(\wireSetA,\wireSetB) \mapsto \interp{d}^{\oplus \oone}
\in\bfup{H}^{\oplus \oone}(\interp{\wireSetA},\interp{\wireSetB})$.

The semantics for the functional fragment is given in
\Cref{fig:sem_fun}, where $m^|$ denote the unique
morphisms given by Mac Lane's coherence theorem\footnote{See \cite{Kelly64,Joyal93}, or \Cref{thm:maclane_monoidal} from the appendices, using the fact that $|$ is a monoidal product as proven in \Cref{prop:parallel_is_monoidal} in the appendices.} on
$(\bfup{H},\mid,\ozero)$ that apply the sequence of associators for $|$
required for the definition to typecheck, and where $\interp{A \isoML
A'}$ is the unique morphism given by our generalization of Mac Lane's
coherence theorem to categories with two monoidal structures (see
\Cref{def:maclane_blue} in the appendices) that applies the sequence of associators for $\tensor$ and $\oplus$, unitors for $\tensor$ and $\oplus$ and their inverses required
to obtain a morphism from $\interp{A}$ to $\interp{A'}$.
In the semantics, remember that $A \mid B$ stands for $(A\tensor B) \oplus (A\oplus B)$, therefore the morphisms $\pi_r, \pi_l, \iota_l, \iota_r$ act on this type and, for example, we can have $\pi_r : A \mid B \to A \oplus B$ and $\pi_\ell : A \mid B \to A \tensor B$. This semantics is universal:
\begin{theorem}[Universality]\label{thm:universality}
	For every object $\wireSetA,\wireSetB$, for all $f \in \bfup{H}(\interp{\wireSetA},\interp{\wireSetB})$, there exists a diagram $d \in \FCat(\wireSetA,\wireSetB)$ such that $\interp{d} = f$.
\end{theorem}

Then, the semantics for the whole calculus is given in \Cref{fig:sem}. For readers preferring matrix notation instead of categorical notations, we provide an equivalent semantics using matrices in the appendices, \Cref{fig:mat_sem_fun,fig:mat_sem}. The universality result also holds in the general case, with $\bfup{H}^{\oplus \oone}$ instead of $\bfup{H}$, although we remind the reader that the definition of $\bfup{H}^{\oplus \oone}$ explicitly excludes morphisms such as $\zero$ when $R$ has no negative element.

\section{The Equational Theory}
\label{sec:equations}
\subsection{The Functional Case}

\begin{figure*}\centering
	\tikzfig{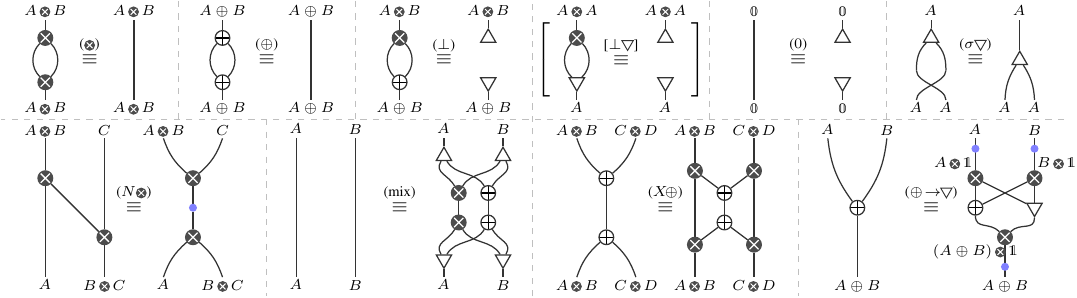}\\[0.1cm]
	and their up-down mirrored versions.
	\caption{Main equations for $\FCat$}
	\label{fig:eq_tensor_plus}
\end{figure*}

\begin{figure*}\centering
	\tikzfig{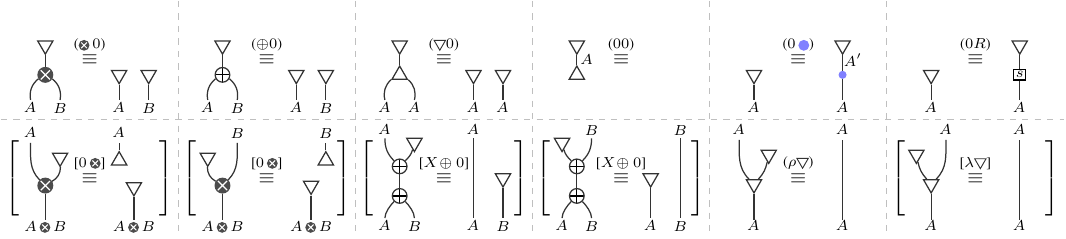}\\[0.1cm]
	and their up-down mirrored versions.
	\caption{Additional equations describing the commutations
	of $\zero$ in $\FCat$}
	\label{fig:eq_null_nat}
\end{figure*}
\begin{figure*}\centering
	\tikzfig{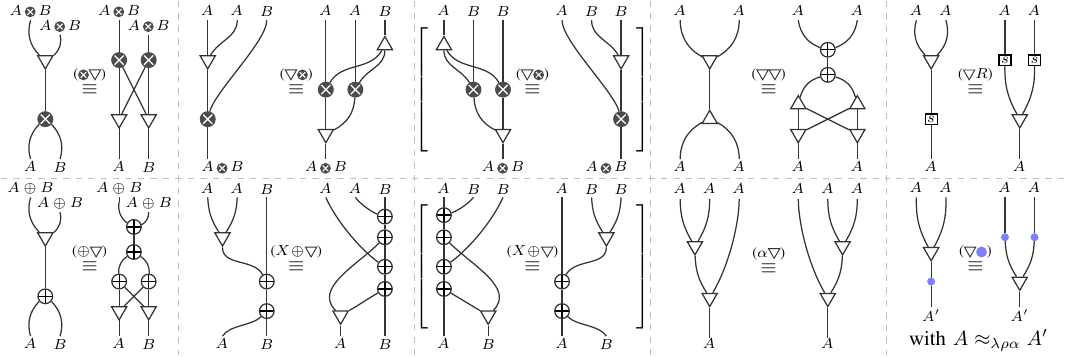}\\[0.1cm]
	and their up-down mirrored versions.
	\caption{Additional equations describing the commutations
	of $\contraction$ in $\FCat$}
	\label{fig:eq_contraction_nat}
\end{figure*}

\begin{figure*}\centering
	\tikzfig{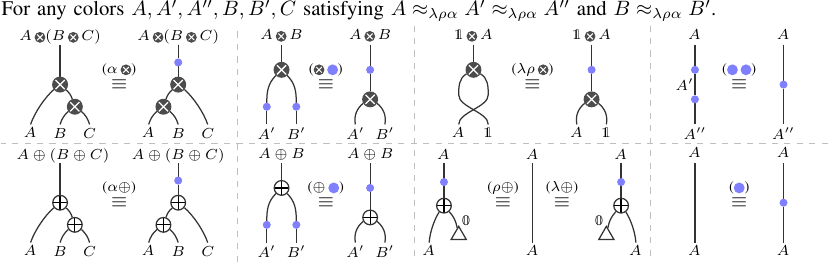}\\[0.1cm]
	and their up-down mirrored versions.
	\caption{Additional equations representing $\isoML$ in $\FCat$}
	\label{fig:eq_maclane}
\end{figure*}

\begin{figure*}\centering
	\tikzfig{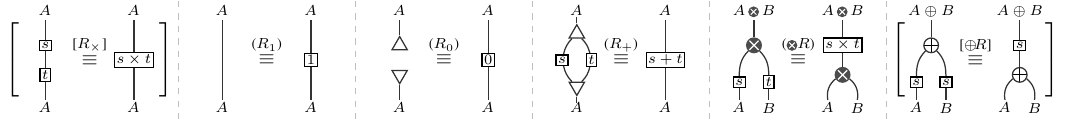}\\[0.1cm]
	and their up-down mirrored versions.
	\caption{Additional equations representing the semiring
	$(R,+,0,\times,1)$ in $\FCat$}
	\label{fig:eq_scalar}
\end{figure*}

\begin{figure*}\centering
	\tikzfig{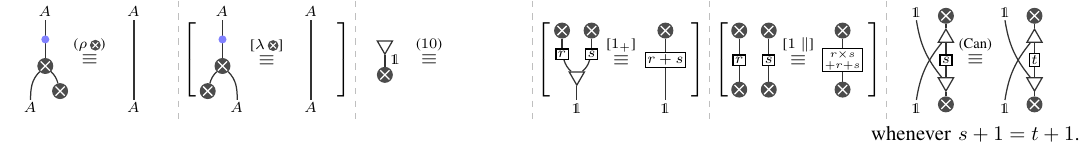}\\[0.1cm]
	and their up-down mirrored versions.
	\caption{Additional Equations for the full \Langage}
	\label{fig:eq_unit}
\end{figure*}

\begin{figure*}\centering
	\[\left[\tikzfig{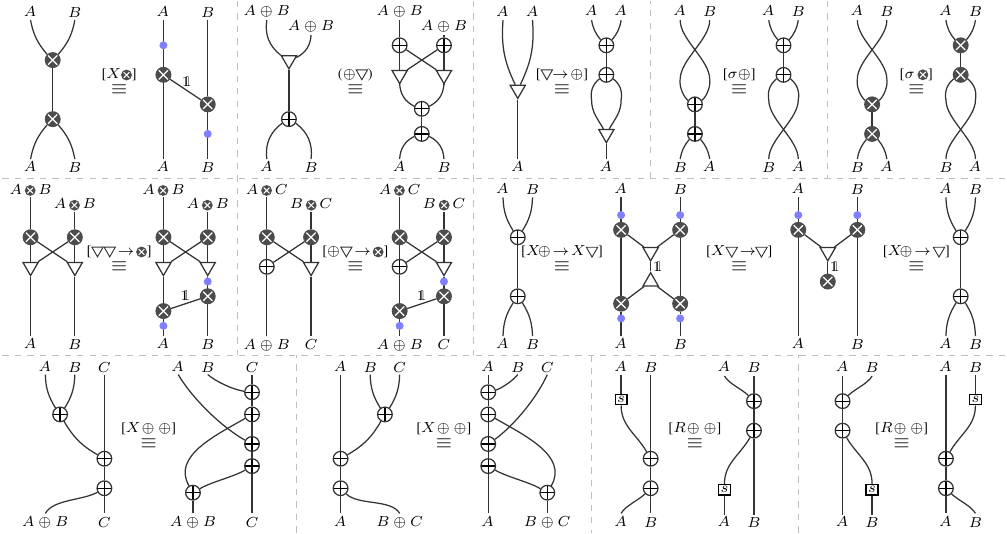}\right]\]
	and their up-down mirrored versions.
	\caption{A Collection of Various Other Deducible Equations}
	\label{fig:eq_induced_in_core}
\end{figure*}

The equivalence relation $\equiv$ is generated by the non-bracketed equations in
\Cref{fig:eq_tensor_plus,fig:eq_maclane,fig:eq_scalar,fig:eq_null_nat,fig:eq_contraction_nat},
which include the up-down mirrored version of each equation. The bracketed equations can be deduced from the non-bracketed ones (see \Cref{app:lemma}) but are left for convenience. The name of those deduced equations also use square brackets, unless they are simply the left-right mirror of an existing axioms in which case we keep the same name. The most important equations are
found in \Cref{fig:eq_tensor_plus}:
\begin{itemize}
	\item \eqT{}, \eqP{} and \eqbot{} show that the Tensor and the
	Plus act as projections/injections, where plugging one to its
	mirrored version gives the identity while plugging one to the other gives
	the Null.
	\item \eqN{} shows that the color $\tzero$ is useless as no data can
	travel through wires of that color.
	\item \eqSigmaC{} shows that the Contraction is a (co-)commutative operation.
	\item \eqTT{} means that Tensors have no computational effect on
	the state, other than enforcing that either both inputs are used at the same time, or they are both not used
	\item \eqPP{} means that two Pluses head-to-head have no
	computational effect on the tokens other than enforcing that at
	most one input is used.
	\item  \eqPtoC{} shows that the Contraction and the Plus have very
	similar behaviors.
	\item \eqmix{} means that $A\parallel B$ is ``either $A \tensor B$
	or $A \oplus B$''.
\end{itemize}
Then:
\begin{itemize}
	\item The equations provided in
	\Cref{fig:eq_contraction_nat,fig:eq_null_nat} describe how both
	the Null and the Contraction can commute with most other
	constructors.
	\item The equations in \Cref{fig:eq_maclane} explain how the
	definition of $\isoML$ is expressed in the equational theory.
	\item  The equations in \Cref{fig:eq_scalar} explain how the
	semiring $(R,+,0,\times,1)$ is expressed in the equational
	theory.
\end{itemize}
Lastly, \Cref{fig:eq_induced_in_core} are a collection of various other equations that can be deduced from the others, included for convenience. Only \eqPPtoC{} and \eqCCtoC{} rely on the Unit, so all the others are in $\FCat{}$.

All these equations were obviously chosen so that their application does not change the semantics of the diagrams:

\begin{proposition}[Soundness]\label{prop:soundness_fun} For $d,e \in
	\FCat(\wireSetA,\wireSetB)$, if $d \equiv e$ then $\interp{d}
	= \interp{e}$.
\end{proposition}
The proof consists in checking that each individual equation is sound with respect to the categorical semantics, which is done in details in \Cref{app:soudness}.

Soundness tells us that the equational theory captures part of the semantics associated to the graphical language. We can go one step beyond and show that it captures \emph{exactly} the semantics, i.e.~the converse of soundness:

\begin{theorem}[Completeness]\label{thm:complete_fun} For $d,e \in
	\FCat(\wireSetA,\wireSetB)$, if $\interp{d} = \interp{e}$ then
	$d \equiv e$.
\end{theorem}

This is proven through a normal form. We refer to \Cref{thm:nf_fun} in
the appendix for the proof that every morphism can be put in normal
form, and to \Cref{cor:nf_unique_fun} for the proof that if two
morphisms share the same semantics, then they share the same normal
form.

\subsection{The General Case}
\label{sec:general-case}

When defining the equational theory for $\Cat$, we simply take the
equational theory for $\FCat$ and add the few equations of
\Cref{fig:eq_unit}, proven sound in \Cref{app:soudness}. The last equation, \eqUS{}, will be irrelevant as soon as $R$ is cancellative, as it follows that $s = t$ whenever $s+1 = t+1$. We write $\Cat_{\equiv}$ for $\Cat$ quotiented by
$\equiv$.

Quite importantly, we can use $\equiv$ for both equational theories
without ambiguity:
\begin{proposition}
	For $d, e \in \FCat(\wireSetA,\wireSetB)$, $d \equiv e$ for
	the equational theory of $\FCat$ if and only if $d \equiv e$ for
	the equational theory of $\Cat$.
\end{proposition}
\begin{proof}
	The direct implication follows from the fact that we only added
	equations. The indirect implication follows from the fact that the
	semantics of $\Cat$ is conservative with respect to
	the semantics of $\FCat$, hence it will not equate morphisms that
	had distinct semantics as morphisms of $\FCat$. Since the
	equational theory of $\FCat$ is complete, it means the additional
	equations do not equate morphisms that are distinct within
	$\FCat_\equiv$.
\end{proof}

Note that outside of those equations, the Unit generator is
ill-behaved, and in particular assuming a non-trivial $R$ we have the following inequations\footnote{More precisely, the first and last inequations require that $R$ satisfies $1+1 \neq 1$, and the other two require the even weaker property $1 \neq 0$.}:
\begin{align*}
	\tikzfig{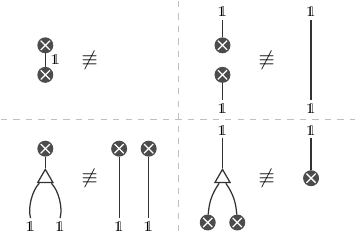}
\end{align*}
This behavior is actually in line with our semantics, as the
completeness result also extends, as proved in
\Cref{thm:nf,cor:nf_unique} using a slight generalization of the
previous normal form.

Therefore, soundness (\Cref{prop:soundness_fun}) and completeness
(\Cref{thm:complete_fun}) extend in this setting.

\begin{figure*}
	\newcommand{\eqreft}[1]{\overset{\substack{#1}}{\equiv}}
	\[
	\tikzfig{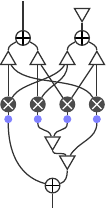}
	\quad\eqreft{\eqPN\\\eqCN}\quad
	\tikzfig{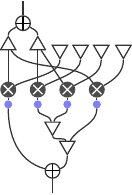}
	\quad\eqreft{\eqTNleft}\quad
	\tikzfig{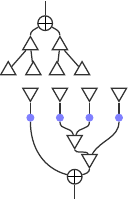}
	\quad\eqreft{\eqRhoC\\\eqNM}\quad
	\tikzfig{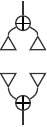}
	\quad\eqreft{\eqPN\\\eqP}\quad
	\tikzfig{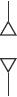}
	\]
	\caption{Rewrite of the strict OR}
	\label{fig:ex:strict-or-rw}
\end{figure*}

\begin{figure*}
	\newcommand{\eqreft}[1]{\overset{\substack{#1}}{\equiv}}
	\[
	\tikzfig{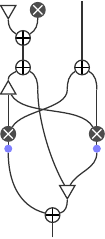}\quad
	\eqreft{\eqPPNleft\\\eqCN}\quad
	\tikzfig{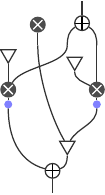}
	\quad
	\eqreft{\eqTNleft}\quad
	\tikzfig{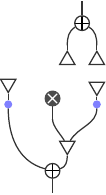}
	\quad
	\eqreft{\eqNM\\\eqRhoC\\\eqPN\\\eqP}\quad
	\tikzfig{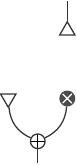}
	\]
	\label{fig:ex:or-rw}
	\caption{Rewrite of the lazy OR}
\end{figure*}

\begin{figure*}\centering
	\tikzfig{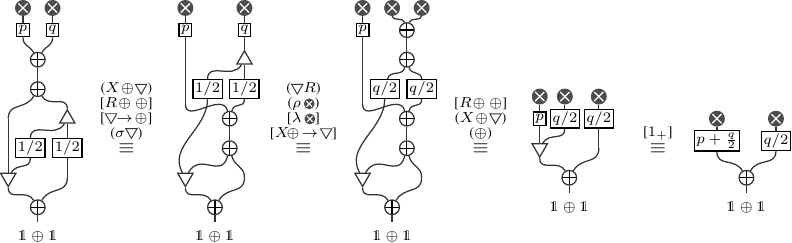}
	\caption{Rewrite of the Application of a Probabilistic Matrix on a pbit (general case).}
	\label{fig:ex_mat_on_bit_rewrite}
\end{figure*}

\begin{figure}
	\[
	\tikzfig{examples/mat_on_bit_1} \hfill \to \hfill \tikzfig{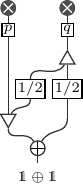} \hfill\to\hfill\tikzfig{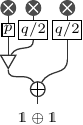}\hfill\to\hfill\tikzfig{examples/mat_on_bit_4}\]
	\caption{Probability matrix applied to a pbit (simple case).}
	\label{fig:ex-mat-in-bit-rewrite-1}
\end{figure}

\subsection{Back to the Examples}

In this section, we revisit earlier examples in order to illustrate
the equational theory.

\begin{example}\rm
  We start by revisiting \Cref{ex:OR}. In fact, the functional part of
  the calculus is enough already to express part of the behavior of
  the ORs, for instance applying a Null to the
  second input of the strict OR. Since the strict OR forces both
  inputs to be evaluated to True or False, and since the Null means
  that the input is not used, the output should not be used either
  here, i.e.~it should be equivalent to throwing away the first input
  and outputting the Null. This is shown in
  \Cref{fig:ex:strict-or-rw}.

  In the context of the general, non-functional \Langage{}, we have
  access to the tensor unit, and we can fully express the behavior of
  the ORs. For instance, the lazy OR allows the second input to not be
  evaluated if the first input is evaluated to True. In that case, it
  should be equivalent to ignoring the second input and outputting
  True, as in \Cref{fig:ex:or-rw}.
\end{example}

\begin{example}\label{ex:revisiting-proba_mat}\rm
	Let us now revisit \Cref{ex:proba_mat} and compute the result of the the application of a matrix to a vector:
	As a first approximation, the two
	Pluses at the top will eliminate each other, connecting the
	left-hand side of each Plus together, and similarly for the
	right-hand side. After that, both the value $q$ and the Unit are
	duplicated through the Contraction. Finally, $p$ and $\frac{q}{2}$
	are summed by the last Contraction. We obtain the final result,
	corresponding to the vector $\begin{pmatrix} p + \frac{q}{2} \\
          \frac{q}{2} \end{pmatrix}$, as shown in \Cref{fig:ex-mat-in-bit-rewrite-1}.

	However, while each of the diagrams of the above figure are equivalent to the others, this ``first approximation'' is only correct because of the presence of the Plus at the very bottom. Without that context, the two head-to-head Pluses would not eliminate each others. In actual practice, instead of eliminating them directly we slide the two head-to-head Pluses downward until they are absorbed by the bottom Plus, as shown in \Cref{fig:ex_mat_on_bit_rewrite}.
\end{example}

\subsection{\Langage{} as an Internal Language}
\label{sec:internal}
In \Cref{sec:cat_sem}, we built a semantics toward a semiadditive category with an
additional symmetric monoidal structure. In this subsection, we claim that $\sFCat$, the restriction of $\FCat$ to morphisms with a single input and a single output, is
an \emph{internal language} for ``semiadditive categories, with a symmetric monoidal structure distributive over it, and such that the homset of automorphisms over the latter's unit are isomorphic to $R$'', which means the following:

\begin{theorem}\label{thm:internal_language_fun}
	Let $\sFCat$ be the full subcategory of $\FCat$ with for objects the colors of $\FCat$.	We can see
	$\sFCat_\equiv$ as a semiadditive category $(\sFCat_\equiv,\oplus,\tzero)$ and a symmetric monoidal category
	$(\sFCat_\equiv,\tensor,\tone)$ such that if we consider $\bfup{H} =
	\sFCat_\equiv$ (with $\bfup{scal}(s) = s$) then the semantics
	$\interp{-} : \sFCat_\equiv \to \bfup{H}$ is the identity.
\end{theorem}

We provide a detailed proof in the appendix \Cref{app:internal}. For the most part, this is a direct consequence of the completeness result of \Cref{thm:complete_fun}. While the restriction to single input/output might look significant, it is actually minimal as:

\begin{proposition}\label{prop:equivalence_single_fun}
	The category $\sFCat_\equiv$ is equivalent to the category $\FCat_\equiv$.
\end{proposition}

\begin{figure}
	\centering
	\tikzfig{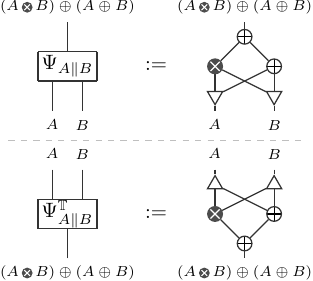}
	\caption{Isomorphism between $A \parallel B$ and $(A\mid B) := (A \tensor B) \oplus (A \oplus B)$.}
	\label{fig:bar_to_parallel}
\end{figure}

This equivalence follows from the existence of a natural isomorphism $\Psi_{A \parallel B}$ defined in \Cref{fig:bar_to_parallel} between the two-colors object $A \parallel B$ and the one-color object $(A \tensor B) \oplus (A \oplus B)$, and is detailed in the appendices \Cref{app:singleIO}. Both results extend to the non-functional case:

\begin{theorem}\label{thm:internal_language}
	Let $\sCat$ be the full subcategory of $\Cat$ with for objects the colors of $\Cat$. We can see $\sCat_\equiv$
	as being equivalent to $(\sFCat_\equiv)^{\oplus \oone}$, such that if
	we consider $\bfup{H} = \sFCat$ (with $\bfup{scal}(s) = s$) then
	the semantics $\interp{-}^{\oplus \oone} : \sCat \to \bfup{H}^{\oplus
	1}$ is the identity.
\end{theorem}

\begin{proposition}\label{prop:equivalence_single}
	The category $\sCat_\equiv$ is equivalent to the category $\Cat_\equiv$.
\end{proposition}

\section{Conclusion}
\label{sec:conclusion}
We introduced a new graphical language unifying additive and
multiplicative structure, gave it categorical semantics and an
equational theory, and proved the main properties one should expect:
universality, soundness, and completeness. We show how this graphical
language is an internal language for semiadditive categories with an
additional monoidal structure that is distributive over it and
parametrized by some algebraic effect.

This language is a first step towards the unification of languages
based on the tensor $\otimes$ and those based on the biproduct
$\oplus$. The language allows us to reason about both systems in
parallel and superpositions of executions, as shown by the encoding of
the Switch in \Cref{ex:switch}.

A natural development of the \Langage{} consists in accommodating
recursive types in the language. We expect significant difficulties
because macros defined by induction on the type -- like the cup/caps
of \Cref{rem:cupcap} and the normal form -- will no longer be
well-defined.

Additional questions also arise when instantiating the language to a
specific semiring. For example, we can represent pure quantum
computation when considering $R = \mathbb{C}$. However, the question
of a complete equational theory for the language on mixed states
remains open (one idea could be to rely on the discard construction
\cite{carette2019completeness}).

Finally, the study of the graphical interaction between additive and
multiplicative connectors has extensive literature within the
community of multiplicative-additive proof nets of linear
logic~\cite{hughes-glabbeek}. As such, a natural development would be
to attempt to split our $\oplus$ and $\tensor$ into two dual
connectors each (respectively Plus $\oplus$ and With $\whitewith$, and
Tensor $\tensor$ and Par $\blackpar$).

%


\appendices
\clearpage
\section{Commutative Semirings and Semimodules}

In this appendix, we recall the basic notions about semirings that we need for our paper.

\subsection{Commutative Semirings}
\label{app:semirings}

\begin{definition}
	A monoid $(G,+,0)$ is a set $G$ together with a distinguised element $0$ and a binary operation $+ : G \times G \to G$ such that:
	\[ \begin{array}{rcl}(a+b)+c & = & a+(b+c) \\
	a+0 & = & a \\
	\end{array}\]	
	It is commutative if additionally
	\[ \begin{array}{rcl}
	a+b & = & b+a \\ \end{array}\]
\end{definition}

\begin{definition} \label{def:semiring}
	A semiring $(R,+,0,\times,1)$ is composed of a commutative monoid  $(R,+,0)$ and a monoid $(R,\times,1)$ with the latter distributive over the former:
	\[ \begin{array}{rcl}
	a \times 0 & = & 0\\ a \times (b+c) & = & a \times b + a \times c \\
	\end{array}\]		
	It is commutative whenever $(R,\times,1)$ is commutative.
\end{definition}
This definition can be instantiated as the complex numbers
$(\mathbb{C}, +, 0, \times, 1)$ to represent pure quantum
computations, the non-negative real numbers $(\mathbb{R}_{\geq 0}, +,
0,\times, 1)$ for probabilistic computations, or the booleans
$(\{\text{False},\text{True}\},$ $ \text{OR}, \text{False},
\text{AND}, \text{True})$ for non-deterministic computations.

The first two examples are \textbf{cancellative}, that is whenever $a+c = b+c$ then $a = b$, while the last example does not since False OR True = True OR True but False $\neq$ True.

\begin{definition}\label{def:semiring_homo}
	A semiring homomorphism $f : (R,+,0,\times,1) \to (R',+',0',\times',1')$ is a function $f : R \to R'$ that preserves all the structures:
	\[ \begin{array}{rclcrcl}
	f(a+b) & = & f(a)+'f(b) & & f(0) & = & 0'\\
	f(a\times b) & = & f(a)\times'f(b) & & f(1) & = & 1'\\
	\end{array} \]
	A semiring homomorphism is an isomorphism whenever it is inversible as a function.
\end{definition}

\subsection{Semimodules}
\label{app:semimodules}

Semimodules are the generalisation of vector spaces where instead of a field they are parametrised by a commutative semiring.

\begin{definition}\label{def:semimodule}
	Let $(R,+,0,\times, 1)$ be a commutative semiring. 
	An $R$-semimodule $(M,\cdot,\bfup{+},\zero)$ is a semigroup $(M,\bfup{+},\zero)$ together with a binary operation $\cdot : R \times M \to M$ satisfying:
	\[ \begin{array}{rclcrcl}
	1 \cdot u & = & M & & (a \times b) \cdot u & = & a \cdot (b \cdot u) \\
	0 \cdot u & = & \zero & & (a + b) \cdot u & = & a \cdot u \bfup{~+~} b \cdot u \\
	a \cdot \zero & = & \zero & & a \cdot (u \bfup{~+~} v) & = & a \cdot u \bfup{~+~} a \cdot v \\
	\end{array}\]		
\end{definition}

In this paper, we do not consider a single semimodule, but instead a category $\bfup{H}$ where every homeset $\bfup{H}(H,K)$ is a semimodule, compatible with each others as follows.

\begin{definition}\label{def:semimodule_enriched}
	A category $\bfup{H}$ is enriched over $R$-semimodules if for every object $H,K$, the homeset $\bfup{H}(H,K)$ is a semimodule $(\bfup{H}(H,K),\cdot,\bfup{+},\zero)$, such that:	
	\[ \begin{array}{rcl}
	(a \times b) \cdot (g \circ f) & = & (a \cdot g) \circ (b \cdot f) \\
	\zero \circ f & = & \zero \\
	g \circ \zero & = & \zero \\
	(g_1 \bfup{~+~} g_2) \circ f & = &  (g_1 \circ f) \bfup{~+~} (g_2 \circ f) \\
	g \circ (f_1 \bfup{~+~} f_2) & = &   g \circ f_1 \bfup{~+~} g \circ f_2 \\
	\end{array}\]		
\end{definition}

In the above definition, the notations are ambiguous as each semimodule has its own $\cdot$, $\bfup{+}$ and $\zero$, so we ought to give them distinct names $\cdot_{H,K}$, $\bfup{+}_{H,K}$ and $\zero_{H,K}$. We consider that this ambiguity is worth the improvement in readability.

\clearpage
\section{Category Theory}
\label{app:category}

In this appendix, we provide detailed definition of the categorical notions used in this paper, including proofs of the various results (all of them either already proved in various textbooks or known in folklore) we use.

\subsection{The Monoidal Structure}

At the core of the categories we use in this paper is the notion of monoidal category. In this subsection, we recall some basic notions about them, including Mac Lane's coherence theorem. In a few words:
\begin{itemize}
	\item A monoidal category is a category where we pair objects together, and where this "pairing" is associative and has a neutral element, but where those associativity and neutrality are not strict but only "up to an isomorphism".
	\item Mac Lane's coherence theorem ensures that even though the associativity and neutrality are "up to an isomorphism", those isomorphisms are fully transparent and could be handwaved\footnote{This "handwaving" can actually be formalised as a strictification, where one builds an equivalent category where the associativity and neutrality are strict.} without causing issues.
\end{itemize} 

\begin{definition}
	A monoidal category $(\bfup{H},\tensor,\oone)$ is a category $\bfup{H}$ together with a bifunctor $\tensor$ and the following natural isomorphisms:
	\begin{itemize}
		\item the left-unitor $\lambda^{\tensor} : \oone \tensor H \to H$
		\item the right-unitor $\rho^{\tensor} : H \tensor \oone \to H$
		\item the associator $\alpha^{\tensor} : (H \tensor K) \tensor L \to H \tensor (K \tensor L)$
	\end{itemize}
	satisfying the following coherence laws:
	\begin{center}
		\tikzfig{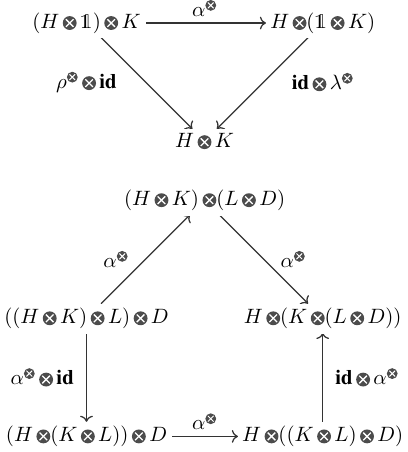}
	\end{center}
	It is a symmetric monoidal category if additionally it has a natural isomorphism:
	\begin{itemize}
	\item The swap $\sigma^{\tensor} : H \tensor K \to K \tensor H$
\end{itemize}
satisfying the following coherence laws:
\begin{center}
	\tikzfig{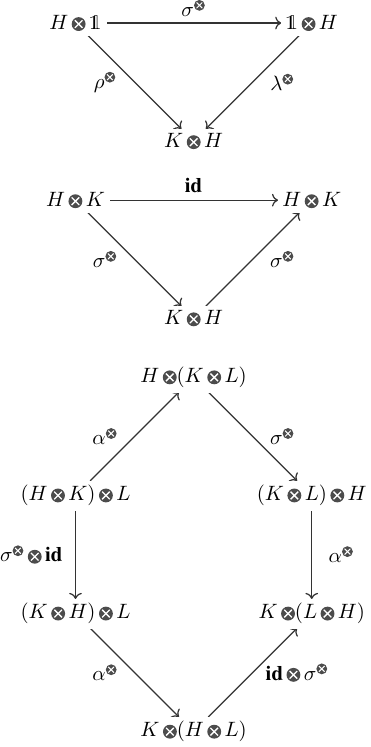}
\end{center}
\end{definition}

We define the following shorthand:
\[ \bigtensor_{i=1}^n H_i := \begin{cases} \oone & \text{whenever $n = 0$} \\ H_1 & \text{whenever $n = 1$} \\ \left(\bigtensor_{i=1}^{n-1} H_i\right) \tensor H_n & \text{whenever $n \geq 2$} \end{cases} \]
We write $\Objet{\tensor}{H_1,\dots,H_n}$ for the set of objects of $\bfup{H}$ that are built with $\tensor$ from the objects $H_1, \dots, H_n$ each used exactly once, and in that specific order, with any number of copies of the object $\oone$ inserted at any position. We note that $\bigtensor_{i=1}^n H_i \in \Objet{\tensor}{H_1,\dots,H_n}$.

For any two $A, B \in \Objet{\tensor}{H_1,\dots,H_n}$, we write $m^{\tensor} : A \to B$ for the natural\footnote{The naturality is with respect to each of the $H_i$.} isomorphism built using $\circ$ and $\tensor$ from the morphisms $\id, \alpha^{\tensor}, \lambda^{\tensor}, \rho^{\tensor}$ and their inverses\footnote{Notably, even in a symmetric monoidal category, we disallow the swap $\sigma^{\tensor}$ from appearing in $m^{\tensor}$.}. In particular, $m^{\tensor} : (\bigtensor_{i=1}^n H_i) \tensor (\bigtensor_{j=1}^m H_{n+j}) \to \bigtensor_{i=1}^m H_i$ is well defined. The uniqueness of $m^{\tensor}$ follows from Mac Lane's coherence theorem, although some subtleties are to be noted:

\begin{example}
	We consider the category of sets and functions between sets. It is a monoidal category for the disjoint union:
	\[ S \uplus T = \{ (0,a) ~|~ a \in S\} \cup \{ (1,b) ~|~ b \in T\}\]
	For a non-empty set $S$, we define $\textup{Left}(S)$ and $\textup{Right}(S)$ as the smallest sets such as:
	\[ \textup{Left}(S) = \textup{Left}(S) \uplus S \qquad \textup{Right}(S) = S \uplus \textup{Righ}(S) \]
	Concretely, an element of $\textup{Left}(S)$ is a sequence of any number of $0$, followed by a $1$, and finally an element of $S$. Similarly, an element of $\textup{Right}(S)$ is a sequence of any number of $1$, followed by a $0$, and finally an element of $S$. We look at the associator $\alpha$:
	\[\begin{matrix}  (\textup{Left}(S) \uplus S) \uplus \textup{Right}(S) &\to &\textup{Left}(S) \uplus (S \uplus \textup{Right}(S)) \\ 
	(0,(0,x)) & \mapsto &(0,x) \\
	(0,(1,x)) & \mapsto & (1,(0,x)) \\
	(1,x) & \mapsto & (1,(1,x)) \\
	\end{matrix} \]
	But using the equations satisfied by $\text{Left}(S)$ and $\text{Right}(S)$, we have:
	\[ \begin{matrix} (\textup{Left}(S) \uplus S) \uplus \textup{Right}(S) &=&\textup{Left}(S) \uplus \textup{Right}(S) \\&=&  \textup{Left}(S) \uplus (S \uplus \textup{Right}(S))\end{matrix} \]
	Said otherwise, $\id$ is a valid morphism between those two objects. And since $\id \neq \alpha$, we found two structural morphisms between the same objects that are different, hence breaking the naive understanding of Mac Lane's coherence.
\end{example}

This counterexample is why the actual coherence theorem is about formal morphisms, as we define here.

\begin{definition}
	A $\tensor$-formal object of $\bfup{H}$ is an element of the following syntax:
	\[ A, B, \dots ::= \oone~|~A\tensor B~|~\underline{H} \text{ (for any $H$ object of $\bfup{H}$)} \]
	We write $(\tensor\!\bfup{-FormH},\tensor,\oone)$ for the monoidal category of formal objects together with only the structural morphisms: $\lambda^{\tensor},\lambda^{\tensor-1},\rho^{\tensor},\rho^{\tensor-1},\alpha^{\tensor},\alpha^{\tensor-1}$, and their composition through $\circ$, and $\tensor$. 
	There is an obvious forgetful functor $U : \tensor\!\bfup{-FormH} \to \bfup{H}$. 
\end{definition}

\begin{theorem}[Mac Lane's Coherence Theorem]\label{thm:maclane_monoidal}
	Given two morphisms $f,g \in \tensor\!\!\bfup{-FormH}(A,B)$, we necessarily have $f=g$. 
\end{theorem}

As such, when we say that $m^{\tensor}$ is unique, it is only correct because we are implicitly working with formal morphisms, hence forbidding any "using the fact that an object can be written in different equivalent ways" like in the previous counter-example.

\subsection{The Semiadditive Structure}

The second structure central to our paper is the one of categories with (finite) biproducts, also called semiadditive categories.

\begin{definition}
	A semiadditive category $(\bfup{H},\oplus,\ozero)$ is a category $\bfup{H}$ with finite biproducts, that is a bifunctor $\oplus$ and the following natural transformations:
	\begin{itemize}
		\item the left-injection $\iota_\ell : H \to H \oplus K$
		\item the right-injection $\iota_r : K \to H \oplus K$
		\item the initial morphism $\zero : \ozero \to K$
		\item the co-diagonal $\nabla : H \oplus H \to H$
		\item the left-projection $\pi_\ell : H \oplus K \to H $
		\item the right-projection $\pi_r : H \oplus K \to K$
		\item the terminal morphism\footnote{We reuse the name $\zero$ at many places, and this is deliberate as the notation is never ambiguous. For example, the initial morphism $\zero : \ozero \to \ozero$ is the same as the terminal morphism $\zero : \ozero \to \ozero$.} $\zero : H \to \ozero$
		\item the diagonal $\Delta : H \to H \oplus H$
	\end{itemize}
	satisfying the following coherence laws:
	\begin{itemize}
		\item For every $f : H \to \ozero$ or $f : \ozero \to H$, necessarily $f = \zero$.
		\item For every $f : L \to H$, $g : L \to K$, $f' : H \to L'$ and $g' : H \to L'$, then
		\[ h = (f \oplus g) \circ \Delta :  L \to H \oplus K \]
		\[ h' = \nabla \circ  (f' \oplus g') : H \oplus K \to L'\]
		 are the unique morphisms such that:
	\end{itemize}
\begin{center}
	\tikzfig{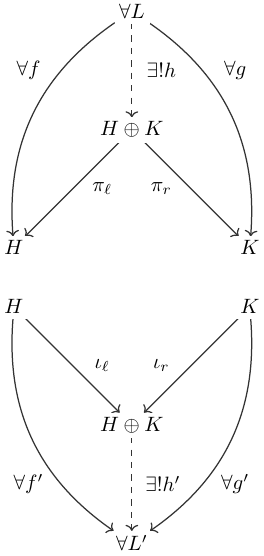}
\end{center}
\end{definition}

The uniqueness of $h$ and $h'$ is a very strong property, which allows us to derive many well-known identities, such as $\pi_\ell \circ \Delta = \id$.

\begin{proposition}
	The semiadditive category $(\bfup{H},\oplus,\ozero)$  is symmetric monoidal, meaning that we can define the natural isomorphisms $\lambda^{\oplus}$, $\rho^{\oplus}$, $\alpha^{\oplus}$, and $\sigma^{\oplus}$. Additionally, it is enriched over commutative monoids $(\bfup{H}(H,K),\bfup{+},\zero)$, where:
	\begin{itemize}
		\item the neutral element $\zero : H \to K$ is the composition of the terminal morphism $\zero : H \to \ozero$ and the initial morphism $\zero : \ozero \to K$.
		\item the sum $f\bfup{+}g$ is simply $\nabla \circ (f \oplus g) \circ \Delta$.
	\end{itemize}
	We also note that $(\bfup{H}(H,H),\bfup{+},\zero,\circ,\id)$ is a semiring.
\end{proposition}
\begin{proof}
	We simply take
	\begin{itemize}
		\item $\lambda^{\oplus} := \iota_\ell$
		\item $\rho^{\oplus} := \iota_r$
		\item $\alpha^{\oplus} := (\pi_\ell \circ \pi_\ell, (\pi_r \circ \pi_\ell,\pi_r))$
		\item $\sigma^{\oplus} := (\pi_r,\pi_\ell)$
	\end{itemize}
	and the coherence laws follow from the properties of finite biproducts. Similarly, those properties directly lead to $(\bfup{H}(H,K),\bfup{+},\zero)$ being a commutative monoid, and the composition $\circ$ being linear with respect to that monoid.
\end{proof}

Since $(\bfup{H},\oplus,\ozero)$  is symmetric monoidal, we can define $\Objet{\oplus}{H_1,\dots,H_n}$ similarly to $\Objet{\tensor}{H_1,\dots,H_n}$.

\begin{definition}\label{def:matrix}
	For $A \in \Objet{\oplus}{H_1,\dots,H_n}$, and $f_i \in \bfup{H}(H_i,K)$ for every $i$, we write \[\begin{pmatrix}
	f_1 & \dots & f_n\\
	\end{pmatrix} := \nabla \circ  ( \nabla \circ (f_1 \oplus f_2) \dots \oplus f_n) \circ m^\oplus \in \bfup{H}(A,K) \]
	Similarly, for $B \in \Objet{\oplus}{K_1,\dots,K_m}$, and $f_j \in \bfup{H}(H,K_j)$ for every $j$, we write \[\begin{pmatrix}
	f_1 \\ \vdots \\ f_m\\
	\end{pmatrix} := m^{\oplus} \circ (((f_1 \oplus f_2) \circ \Delta \dots  \oplus f_n) \circ \Delta \in \bfup{H}(H,B)\]
	And for $A \in \Objet{\oplus}{H_1,\dots,H_n}$, $B \in \Objet{\oplus}{K_1,\dots,K_m}$, and $f_{i,j} \in \bfup{H}(H_i,K_j)$ for every $i,j$ we write

	\[\begin{pmatrix}
	f_{1,1} & \dots & f_{n,1}\\
	\vdots & & \vdots \\
	f_{1,m} & \dots & f_{n,m}\\
	\end{pmatrix} = \left(f_{i,j} \right)_{\substack{1 \leq i \leq n\\1 \leq j \leq m}}\]
	for the morphism of $\bfup{H}(A,B)$ defined as\[
	\begin{pmatrix}
	\begin{pmatrix}	f_{1,1} & \dots & f_{n,1}\end{pmatrix}\\
	\vdots \\
	\begin{pmatrix}f_{1,m} & \dots & f_{n,m}\end{pmatrix}
	\end{pmatrix} =
	\begin{pmatrix}
	\begin{pmatrix} f_{1,1} \\ \vdots \\ f_{1,m}\end{pmatrix}&
	\dots &
	\begin{pmatrix}f_{n,1} \\ \vdots \\ f_{n,m}\end{pmatrix}
	\end{pmatrix}\]
\end{definition}

	The matrix notation introduced can be composed using the usual product of matrices, that is:
	\[ \left(g_{j,k} \right)_{\substack{1 \leq j \leq m\\1 \leq k \leq p}} \circ \left(f_{i,j} \right)_{\substack{1 \leq i \leq n\\1 \leq j \leq m}} = \left(\sum_{j=1}^m g_{j,k} \circ f_{i,j} \right)_{\substack{1 \leq i \leq n\\1 \leq k \leq p}} \]
	We note that the elements of the matrix are uniquely determined, as  for $1 \leq i_0 \leq n$ and $1 \leq j_0 \leq m$:
	\[ f_{i_0,j_0} =\projection{j_0}{m} \circ \left(f_{i,j} \right)_{\substack{1 \leq i \leq n\\1 \leq j \leq m}} \circ \injection{i_0}{n} \]
	The following lemma is a direct consequence of that fact.
	\begin{lemma}\label{lem:coef_by_coef}		
		For $A \in \Objet{\oplus}{H_1,\dots,H_n}$, $B \in \Objet{\oplus}{K_1,\dots,K_m}$, and $f \in \bfup{H}(A,B)$, then $f$ has a unique matrix form given by
		\[ f =  \left(\projection{j}{m} \circ f \circ \injection{i}{n}  \right)_{\substack{1 \leq i \leq n\\1 \leq j \leq m}} \]
		In particular, for $f,g \in \bfup{H}(A,B)$, the morphisms $f$ and $g$ are equal if and only if for every $1 \leq i \leq n$ and $1 \leq j \leq m$ we have
		\[\projection{j}{m} \circ f \circ \injection{i}{n}\]\[=\]\[ \projection{j}{m} \circ g \circ \injection{i}{n}\]
	\end{lemma}

	\begin{definition}\label{def:matrix_R}
		Assuming a semiring $(R,+,0,\times,1)$, a distinguished object $\oone$ and a semiring isomorphism $\bfup{scal} : R \to \bfup{H}(\oone,\oone)$, then for every $n \times m$ matrix $M = (m_{i,j})_{i,j} $ with coefficients in $R$, we can define $\Morphism{M}$ as the unique morphism of $\bfup{H}(\bigoplus_{i=1}^n \oone,\bigoplus_{j=1}^m \oone)$ equal to
		\[ \Morphism{M} = \left(\bfup{scal}(m_{i,j}) \right)_{\substack{1 \leq i \leq n\\1 \leq j \leq m}} \]
	\end{definition}
	
	\begin{proposition}\label{prop:matrix_is_fff}
		The operation $M \mapsto \Morphism{M}$ is a full and faithful functor, that is linear for $\bfup{+}$.
	\end{proposition}
	\begin{proof}
		The functoriality and linearity follows from $\bfup{scal}(-)$ being a semiring homomorphism. The fullness and faitfulness follows from $\bfup{scal}(-)$ being an isomorphism, and from \Cref{lem:coef_by_coef}.		
	\end{proof}

\subsection{Adding Distributivity}

In this paper, we consider a category with two structures, one distributive over the other. We note that the use of "distributive category" in this paper is non-standard, as it is usually used whenever $\tensor$ is a cartesian product (ours is only monoidal) and $\oplus$ is a coproduct (ours is even a biproduct). Distributivity can even be studied in the case where $\oplus$ is only a symmetric monoidal product, as done in \cite{laplaza72}.

We now consider $(\bfup{H},\tensor,\oone,\oplus,\ozero)$ such that $(\bfup{H},\tensor,\oone)$ is a symmetric monoidal category and $(\bfup{H},\oplus,\ozero)$ is a semiadditive category.

\begin{definition}
	The category $\bfup{H}$ is said distributive if the following two natural transformations are isomorphisms:
	\begin{itemize}
		\item the left-distributor $\bfup{dist}_\ell : H \tensor (K \oplus L) \to (H \tensor K) \oplus (H \tensor L)$ defined as $(\id \tensor \pi_\ell \oplus \id \tensor \pi_r) \circ \Delta$, with for inverse $\nabla \circ (\id \tensor \iota_\ell \oplus \id \tensor \iota_r)$.
		\item the left-annihilator $\zero : H \tensor \ozero \to \ozero$ defined as the terminal morphism, with for inverse the initial morphism.
	\end{itemize}
\end{definition}
Since $\tensor$ is symmetric, the following natural transformations are also isomorphisms in a distributive category:
\begin{itemize}
	\item the right-distributor $\bfup{dist}_r : (H \oplus K) \tensor L \to (H \tensor L) \oplus (K \tensor L)$ defined as $(\pi_\ell \tensor \id \oplus \pi_r \tensor \id) \circ \Delta$, with for inverse the co-pairing $\nabla \circ (\iota_\ell \tensor \id \oplus \iota_r \tensor \id)$	
	\item the right-annihilator $\zero : \ozero \tensor H \to \ozero$ which is the terminal morphism, with for inverse the initial morphism.
\end{itemize}

\begin{lemma}\label{lem:dist_diag}
	We have the following properties about the distributors:
	\[ \begin{array}{rclcrcl}
	\bfup{dist}_\ell \circ (\id \tensor \Delta) &=& \Delta&&
	\bfup{dist}_r \circ (\Delta \tensor \id) &=& \Delta\\
	(\id \tensor \nabla) \circ \bfup{dist}_\ell^{-1} &=& \nabla&&
	(\nabla \tensor \id) \circ \bfup{dist}_r^{-1} &=& \nabla\\
	\end{array}
	 \] 
\end{lemma}
\begin{proof}
	For the first equation, we write $\textup{lhs}$ for the left-hand-side, and want to prove that $\textup{lhs} = \Delta$. We start by considering $\pi_\ell \circ \textup{lhs}$. Using the fact that $\pi_\ell \circ (f \oplus g) = f \circ \pi_\ell$, and the fact that $\pi_\ell \circ \Delta = \id$, we have:
	\[ \begin{array}{rcl} \pi_\ell \circ \textup{lhs} &=& 
	\pi_\ell \circ \left((\id \tensor \pi_\ell) \oplus (\id \tensor \pi_r)\right) \circ \Delta \circ (\id \tensor \Delta) \\
	&=&(\id \tensor \pi_\ell) \circ \pi_\ell \circ \Delta \circ (\id \tensor \Delta)\\
	&=& (\id \tensor \pi_\ell) \circ (\id \tensor \Delta)\\
	&=& \id \end{array} \]
	Similarly, we obtain $\pi_r \circ  \textup{lhs} = \id$. Using the universal property of the product $\oplus$, we know that there exists a unique morphism $h$ such that $\pi_r \circ h = \id$ and $\pi_\ell \circ h = \id$, and this morphism is exactly $\Delta$. As such, $\textup{lhs} = \Delta$.
	
	The three other cases are proved similarly. 
\end{proof}

\begin{proposition}\label{prop:semimodule}
	Whenever $\bfup{H}$ is distributive, the semiring  $(\bfup{H}(\oone,\oone),\bfup{+},\zero,\circ,\id)$ is commutative, and $\bfup{H}$ is enriched over semimodules\footnote{Semimodules are "vector spaces" but where instead of a field we have a semiring . See \Cref{def:semimodule}.} $(\bfup{H}(H,K),\cdot,\bfup{+},\zero)$ over that commutative semiring, where $x \cdot f$ is defined as $\lambda^{\tensor} \circ (x \tensor f) \circ (\lambda^{\tensor})^{-1}$ for $x \in \bfup{H}(\oone,\oone)$ and $f \in \bfup{H}(H,K)$.
\end{proposition}
\begin{proof}		
	For $f,g \in \bfup{H}(\oone,\oone)$, we have:
	\[ \begin{matrix} f \circ g &=& f \circ \lambda^{\tensor} \circ (\lambda^{\tensor})^{-1} \circ g \\ &=& \lambda^{\tensor} \circ (f \tensor \id) \circ (\id \tensor g) \circ(\lambda^{\tensor})^{-1} \\ &=& \lambda^{\tensor}  \circ (\id \tensor g)\circ (f \tensor \id) \circ(\lambda^{\tensor})^{-1} \\ &=& g \circ \lambda^{\tensor} \circ (\lambda^{\tensor})^{-1} \circ f \\&=& g \circ f \end{matrix} \]
	Hence the semiring is commutative.
	
	From the distributivity of $\tensor$ over $\oplus$ we can deduce the distributivity of $\tensor$ over $\bfup{+}$. More precisely, we need to show:	
	\[ \begin{array}{rclcrcl}
	h \tensor (f \bfup{+} g) &=& (h \tensor f) \bfup{+} (h \tensor g)&&
	h \tensor \zero &=& \zero \\
	(f \bfup{+} g) \tensor h &=& (f \tensor h) \bfup{+} (g \tensor h)&&
	\zero \tensor h &=& \zero \\
	\end{array}\]
	First, to show that $h \tensor (f \bfup{+} g) = (h \tensor f) \bfup{+} (h \tensor g)$, we simply use \Cref{lem:dist_diag} and the naturality of $\bfup{dist}_\ell$ to make the following diagram commute:	
	\begin{center}
		\tikzfig{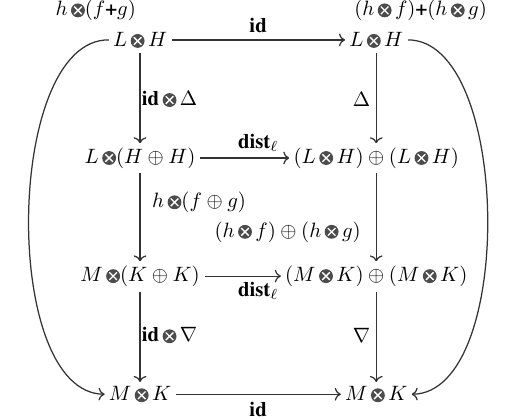}
	\end{center}
	Then, for $h \tensor \zero = \zero$, we start by a special case:
	\begin{itemize}
		\item For $h \tensor \zero : L \tensor H \to L \tensor \ozero$, since $L \tensor \ozero$ is isomorphic to the terminal element $\ozero$, then $L \tensor \ozero$ is also terminal, hence there exists a unique morphism from $L \tensor H$ to $L \tensor \ozero$, hence $h \tensor \zero = \zero$.
		\item Now, for the general case $h \tensor \zero : L \tensor H \to L \tensor K$, we decompose it using $L \tensor \ozero$ as an intermediary object, so $h \tensor \zero = (h \tensor \zero) \circ (\id \tensor \zero)$. Using the previous case, we obtain $h \tensor \zero = \zero \circ (\id \tensor \zero) = \zero$.
	\end{itemize}
	We proceed similarly to show that $(f \bfup{+} g) \tensor h = (f \tensor h) \bfup{+} (g \tensor h)$ and $\zero \tensor h = \zero$.
	From those and the naturality of unitors, it directly follows that $(\bfup{H}(H,K),\cdot,\bfup{+},\zero)$ is a semimodule, that is:
	\[ \begin{array}{rclcrcl}
	\id \cdot f & = & f & & (x \circ y) \cdot f & = & x \cdot (y \cdot f) \\
	\zero \cdot f & = & \zero & & (x \bfup{~+~} y) \cdot f & = & x \cdot f \bfup{~+~} y \cdot f \\
	x \cdot \zero & = & \zero & & x \cdot (f \bfup{~+~} g) & = & x \cdot f \bfup{~+~} x \cdot g \\
	\end{array}\]		
	It remains to show that it is an enrichment, that is:		
	\[ \begin{array}{rcl}
	(x \circ  y) \cdot (g \circ f) & = & (x \cdot g) \circ (y \cdot f) \\
	\zero \circ f & = & \zero \\
	g \circ \zero & = & \zero \\
	(g_1 \bfup{~+~} g_2) \circ f & = &  (g_1 \circ f) \bfup{~+~} (g_2 \circ f) \\
	g \circ (f_1 \bfup{~+~} f_2) & = &   g \circ f_1 \bfup{~+~} g \circ f_2 \\
	\end{array}\]
	The first equation follows from the bifunctoriality of $\tensor$. The second and third ones follows from the fact that one can see $\zero : H \to K$ as the composition of the terminal morphism $H \to \ozero$ and the initial morphism $\ozero \to K$. The fourth and fifth equations are simply the naturality of $\Delta$ and $\nabla$. 
\end{proof}

As we would like to extend Mac Lane's coherence result to distributive categories, we also need to extend the notion of formal morphism.

\begin{definition}
	A formal object of $\bfup{H}$ is an element of the following syntax:
	\[ A, B, \dots ::= \ozero~|~\oone~|~A\oplus B~|~A\tensor B~|~\underline{H} \text{ (for any $H$ object of $\bfup{H}$)} \]
	We write $(\bfup{FormH},\tensor,\oone,\oplus,\ozero)$ for the distributive category of formal objects together with only the structural morphisms: $\lambda^{\tensor},\lambda^{\tensor-1},\rho^{\tensor},\rho^{\tensor-1},\alpha^{\tensor},\alpha^{\tensor-1},\sigma^{\tensor}, \iota_\ell,\iota_r,\pi_\ell,\pi_r,\zero,\nabla,\Delta$, and their composition through $\circ$, $\tensor$, $\oplus$. 
	There is an obvious forgetful functor $U : \bfup{FormH} \to \bfup{H}$. 
\end{definition}
We will rely on those formal morphisms to show equality of some specific morphisms of $\bfup{H}$. Indeed, if we want to show that $f = g$ in $\bfup{H}$, and notice that $f$ and $g$ are built from structural morphisms, then by looking at their corresponding formal morphism $U(f') = f$ and $U(g') = g$ it is enough to prove that $f' = g'$.

We say that a formal object is in normal form if it is of the form $\bigoplus_{i=1}^n\bigtensor_{j=1}^{m_i} \underline{H_{i,j}}$ for some objects $H_{i,j}$ of $\bfup{H}$. In particular, $\ozero$ and $\oone$ are in normal form\footnote{With $n=0$ for $\ozero$, and $n=1, m_1=0$ for $\oone$.}. For every formal object $A$, we define inductively in \Cref{fig:formal_object_normal_form} its normal form $\mathcal{N}(A)$, and the normalization natural isomorphism $n_A : A \to \N(A)$. In this definition, we prioritize the right-distributor over the left-distributor, meaning that we read pairs using the lexicographic order.
This normalization extends to a full and faithful functor with for $f : A \to B$, $\N(f) : \N(A) \to \N(B)$ defined as $n_B \circ f \circ n_A^{-1}$. Relying on \Cref{def:matrix}, we choose to see $\N(f)$ as a matrix. 

\begin{figure*}
	\[
	\begin{array}{lll}
		A & \N(A) & n_A\\\hline
		\\
		\ozero & \ozero & \id \\\\
		\oone & \oone & \id \\\\
		\underline{H} & \underline{H} & \id \\\\
		B \oplus C & \N(C) & \lambda^{\oplus} \circ (n_B \oplus n_C) \\
		\quad\text{with } \N(B) = \ozero && \\\\
		B \oplus C & \N(B) & \rho^{\oplus} \circ (n_B \oplus n_C) \\
		\quad\text{with } \N(B) \neq \ozero &&\\\quad\text{and } \N(C) = \ozero && \\\\	
		B \oplus C & \N(B) \oplus \N(C) & n_B \oplus n_C \\	
		\quad\text{with } \N(B) \neq \ozero &&\\\quad\text{and } \N(C) = \bigtensor_{k=1}^p \underline{L_k} && \\\\	
		B \oplus C & \N(B \oplus C') \oplus \bigtensor_{k=1}^p \underline{L_k} & (n_{B \oplus C'} \oplus \id) \circ \alpha^{\oplus-1} \circ (\id \oplus n_C) \\
		\quad\text{with } \N(B) \neq \ozero &&\\\quad\text{and }  \N(C) = C'\oplus \bigtensor_{k=1}^p \underline{L_k}  && \\\\
		B \tensor C & \ozero & \zero \circ (n_B \tensor n_C) \\
		\quad\text{with } \N(B) = \ozero && \\\\	
		B \tensor C & \N(C) & \lambda^{\tensor} \circ (n_B \tensor n_C) \\	
		\quad\text{with } \N(B) = \oone && \\\\	
		B \tensor C & \ozero & \zero \circ (n_B \tensor n_C) \\
		\quad\text{with } \N(B) \notin \{\ozero,\oone\} &&\\\quad\text{and } \N(C) = \ozero && \\\\	
		B \tensor C & \N(B) & \rho^{\tensor} \circ (n_B \tensor n_C) \\	
		\quad\text{with } \N(B)  \notin \{\ozero,\oone\} &&\\\quad\text{and } \N(C) = \oone && \\\\
		B \tensor C & \N(B' \tensor C) \oplus \N(\bigtensor_{j=1}^m \underline{K_j} \tensor C) & m^\oplus \circ  (n_{B' \tensor C} \oplus n_{\bigtensor_{j=1}^m \underline{K_j} \tensor C}) \circ \dist_r \circ (n_B \tensor \id) \\
		\quad\text{with } \N(B) = B'\oplus \bigtensor_{j=1}^m \underline{K_j} &&\\
		\quad\text{and }  \N(C) \notin \{\ozero,\oone\}  && \\\\
		B \tensor C & \N(B \tensor C') \oplus \N(B \tensor \bigtensor_{k=1}^p \underline{L_k}) &  (n_{B \tensor C'} \oplus n_{B \tensor \bigtensor_{k=1}^p \underline{L_k}}) \circ \dist_\ell \circ (\id \tensor n_C) \\
		\quad\text{with } \N(B)  = \bigtensor_{j=1}^m \underline{K_j}&&\\
		\quad\text{and }  \N(C) = C'\oplus \bigtensor_{k=1}^p \underline{L_k}  && \\\\
		B \tensor C & \bigtensor_{i=1}^{m+p} \underline{H_i} & m^{\tensor} \\
		\quad\text{with } \N(B) = \bigtensor_{j=1}^m \underline{K_i} &\quad\text{with } H_j = K_j&\\\quad\text{and }  \N(C) = \bigtensor_{k=1}^p \underline{L_k}  &  \quad\text{and } H_{m+k} = L_k& \\		
	\end{array}
	\]
	\caption{Normalization of formal objects.}
	\label{fig:formal_object_normal_form}
\end{figure*}

\begin{lemma}\label{lem:proj_inj_matrix}
	We consider $A = \bigoplus_{i=1}^n \bigtensor_{i'=1}^{n_i'} \underline{H_{i,i'}}$ and $B = \bigoplus_{j=1}^m \bigtensor_{j'=1}^{m_j'} \underline{K_{j,j'}}$. We write
	\[ \bfup{proj}_i^n = \projection{i}{n}\]
	\[ \bfup{inj}_i^n = \injection{i}{n}\]	
	Those morphisms satsify
	\[  \bfup{proj}_{i}^{n+m} = \N(\bfup{proj}_{i}^{n} \oplus \zero) \quad \bfup{inj}_{i}^{n+m} = \N(\bfup{inj}_{i}^{n} \oplus \zero) \]	
	with $\zero : B \to \ozero$ and $\zero : \ozero \to B$ respectively.	
	\[  \bfup{proj}_{n+j}^{n+m} = \N(\zero \oplus \bfup{proj}_{j}^{m})\quad   \bfup{inj}_{n+j}^{n+m} = \N(\zero \oplus \bfup{inj}_{j}^{m}) \]	
	with $\zero : A \to \ozero$ and $\zero : \ozero \to A$ respectively.	
	\[  \bfup{proj}_{(i-1)\times n + j}^{n\times m} = \N(\bfup{proj}_{i}^{n} \tensor \bfup{proj}_{j}^{m}) \]	
	\[  \bfup{inj}_{(i-1)\times n + j}^{n\times m} = \N(\bfup{inj}_{i}^{n} \tensor \bfup{inj}_{j}^{m}) \]
\end{lemma}
\begin{proof}
	We will focus on the equations satisfied by $\bfup{proj}$, as the equations satisfied by $\bfup{inj}$ can be derived following exactly the same logic. We assume that $A \neq \ozero$ ($n > 0$) and $B \neq \ozero$ ($m > 0$) as otherwise the equations are trivially true.
	
	We start by unfolding the explicit definition of $\bfup{proj}_i^n$:
	\[\bfup{proj}_i^n = \begin{cases}  \pi_\ell^n &\text{if } i=0 \\ \pi_r \circ \pi_\ell^{(n-i)} & \text{if } i > 1 \end{cases} : A \to \bigtensor_{i'=1}^{n_i'}\underline{H_{i,i'}} \]
	Then, since $\N(A) = A$ and $\N(B) = B$, we can rewrite more explicitly the equations that we want to prove:
	\[  \bfup{proj}_{i}^{n+m} \stackrel{?}{=} \rho^\oplus \circ (\bfup{proj}_{i}^{n} \oplus \zero) \circ m^{\oplus} \]
	\[  \bfup{proj}_{n+j}^{n+m} \stackrel{?}{=} \lambda^\oplus \circ (\zero \oplus \bfup{proj}_{j}^{m}) \circ  m^\oplus \]		
	\[  \bfup{proj}_{(i-1)\times n + j}^{n\times m} \stackrel{?}{=} m^{\tensor} \circ (\bfup{proj}_{i}^{n} \tensor \bfup{proj}_{j}^{m}) \circ n_{A \tensor B}^{-1}  \]
	We recall that $\rho^\oplus$ is simply $\pi_\ell$, and that $\pi_\ell$ satisfies $\pi_\ell \circ (f \oplus g) = f \circ \pi_\ell$. As such, the first equation can be rewritten as:	
	\[  \bfup{proj}_{i}^{n+m} \stackrel{?}{=} \bfup{proj}_{i}^{n} \circ \pi_\ell \circ m^{\oplus} \]
	Then, we look at $m^{\oplus} : \N(A \oplus B) \to A \oplus B$ from that equation:
	\begin{itemize}
		\item If $B = \bigtensor_{j'=1}^{m_1'} \underline{K_{1,j'}}$ then $m=1$ and $m^{\oplus}$ is the identity, hence		
		\[ \begin{array}{lll} \bfup{proj}_{i}^{n} \circ \pi_\ell \circ m^{\oplus} &= & \bfup{proj}_{i}^{n} \circ \pi_\ell \\ &=& \bfup{proj}_{i}^{n+1} \\ &=&  \bfup{proj}_{i}^{n+m} \end{array} \]
		\item Otherwise $m > 1$, then $m^{\oplus}$ is simply a combination of $m-1$ associators, so by iterating $\pi_\ell \circ \alpha = \pi_\ell \circ \pi_\ell$ we obtain $\pi_\ell \circ m^{\oplus} = \pi_\ell^m$, hence		
		\[ \begin{array}{lll} \bfup{proj}_{i}^{n} \circ \pi_\ell \circ m^{\oplus} &= & \bfup{proj}_{i}^{n} \circ \pi_\ell^m \\ &=&  \bfup{proj}_{i}^{n+m} \end{array} \]
	\end{itemize}
Similarly to $\rho^\oplus$, we recall that $\lambda^\oplus$ is simply $\pi_r$, and that $\pi_r$ satisfies $\pi_r \circ (f \oplus g) = g \circ \pi_r$. As such, the first equation can be rewritten as:	
\[  \bfup{proj}_{i}^{n+m} \stackrel{?}{=} \bfup{proj}_{j}^{m} \circ \pi_r \circ m^{\oplus} \]
Then, we look at $m^{\oplus} : \N(A \oplus B) \to A \oplus B$ from that equation it is simply a combination of $m$ inverses of associators, so by iterating $\pi_r \circ \alpha^{\oplus-1} = \pi_r \oplus \id$ we obtain $\pi_r \circ m^{\oplus} = ((\pi_r \oplus \id) \dots \oplus \id)$. Hence
\[\bfup{proj}_{j}^{m} \circ \pi_r \circ m^{\oplus} = \bfup{proj}_{i}^{m} \circ ((\pi_r \oplus \id) \dots \oplus \id) \]\[ =  \begin{cases}  \pi_\ell^m \circ ((\pi_r \oplus \id) \dots \oplus \id) &\text{if } j=0 \\ \pi_r \circ \pi_\ell^{(m-j)} \circ  ((\pi_r \oplus \id) \dots \oplus \id) & \text{if } j > 1 \end{cases} \] \[ \begin{array}{ccc} &=&   \begin{cases}   \pi_r \circ \pi_\ell^m &\text{if } j=0 \\ \pi_r \circ \pi_\ell^{(m-j)} & \text{if } j > 1 \end{cases} \\ &=&  \pi_r \circ \pi_\ell^{(m-j)} \\ &=&  \pi_r \circ \pi_\ell^{(n+m-n-j)}  \\ &=&  \bfup{proj}_{n+j}^{n+m}  \end{array} \]
Now, for the last equation, we take a closer look at $n^{-1}_{A \tensor B}$ and relying on $n_A = \id$ we start to unfold its definition as:

\[ \begin{array}{rl} & n^{-1}_{A \tensor B} \\
= & \dist_r^{-1} \\ \circ& \left(\id \oplus n^{-1}_{\bigtensor_{i'=n}^{n_n'} \underline{H_{n,i'}} \tensor B} \right) \\
\circ & \dots \\ 
\circ & 
\left(
\left(
\dist_r^{-1} \oplus \id
\right) 
\dots \oplus \id
\right) 
\\\circ &
\left(
\left(
\left(
n^{-1}_{\bigtensor_{i'=1}^{n_1'} \underline{H_{1,i'}}\tensor B} \oplus n^{-1}_{\bigtensor_{i'=2}^{n_2'} \underline{H_{2,i'}} \tensor B}
\right) 
\oplus \id
\right) 
\dots \oplus \id
\right)
\\\circ& m^\oplus \end{array}  \]
Then, using $(\pi_\ell \tensor \id) \circ \dist_r^{-1} = \pi_\ell$ and $(\pi_r \tensor \id) \circ \dist_r^{-1} = \pi_r$, we have
\[ (\bfup{proj}_{i}^{n} \tensor \id) \circ n_{A \tensor B}^{-1}= n^{-1}_{A \tensor \bigtensor_{j'=j}^{m_j'} \underline{K_{j,j'}}} \circ \bfup{proj}_{i}^{n} \circ m^{\oplus}\]
We take a closer look at $n^{-1}_{A \tensor \bigtensor_{j'=j}^{m_j'} \underline{K_{j,j'}}}$ and relying on $n_B = \id$ we start to unfold its definition as:
\[ \begin{array}{rl} & n^{-1}_{A \tensor \bigtensor_{j'=j}^{m_j'}} \\
= & \dist_\ell^{-1} \\ \circ& \left(\id \oplus m^{\tensor} \right) \\
\circ & \dots \\ 
\circ & 
\left(
\left(
\dist_\ell^{-1} \oplus \id
\right) 
\dots \oplus \id
\right) 
\\\circ &
\left(
\left(
\left(m^{\tensor} \oplus m^{\tensor} 
\right) 
\oplus \id
\right) 
\dots \oplus \id
\right) \end{array}  \]
Then, using $(\id \tensor \pi_\ell) \circ \dist_\ell^{-1} = \pi_\ell$ and $(\id \tensor \pi_r) \circ \dist_\ell^{-1} = \pi_r$, we have
\[(\id \tensor \bfup{proj}_{j}^{m}) \circ n^{-1}_{\bigtensor_{i'=i}^{n_i'} \underline{H_{i,i'}} \tensor B} = m^{\tensor} \circ \bfup{proj}_{j}^{m} \]
Hence 
\[ m^{\tensor} \circ (\bfup{proj}_{i}^{n} \tensor \bfup{proj}_{j}^{m}) \circ n_{A \tensor B}^{-1} =m^{\tensor} \circ m^{\tensor} \circ \bfup{proj}_{j}^{m} \circ \bfup{proj}_{i}^{n} \circ m^{\oplus} \]
We note that in this case, the two successive $m^{\tensor}$ are inverse of one another. Additionally, $m^{\oplus}$ from that equation it is simply a combination of associators, so by iterating $\pi_\ell \circ \alpha = \pi_\ell \circ \pi_\ell$ we obtain
\[  \bfup{proj}_{j}^{m} \circ \bfup{proj}_{i}^{n} \circ m^{\oplus} = \bfup{proj}_{(i-1) \times m + j}^{n \times m}\]
Hence the expected result:
\[ m^{\tensor} \circ (\bfup{proj}_{i}^{n} \tensor \bfup{proj}_{j}^{m}) \circ n_{A \tensor B}^{-1} = \bfup{proj}_{(i-1) \times m + j}^{n \times m}\]
\end{proof}

\begin{lemma}\label{lem:formal_normalization_monoidal}
	We consider $f : A \to B$, $g : B \to C$ and $h : C \to W$, and write $\N(f)$, $\N(g)$ and $\N(h)$ as:
	\[ \left(f_{i,j} \right)_{\substack{1 \leq i \leq n\\1 \leq j \leq m}} \quad \left(g_{j,k} \right)_{\substack{1 \leq j \leq m\\1 \leq k \leq p}} \quad  \left(h_{k,\ell} \right)_{\substack{1 \leq k \leq p\\1 \leq \ell \leq q}} \]
	The composition corresponds to the usual product of matrices:
	\[ \N(g \circ f) = \left(\sum_{j=1}^m g_{j,k} \circ f_{i,j} \right)_{\substack{1 \leq i \leq n\\1 \leq k \leq p}} \]	
	The biproduct corresponds ot the usual direct sum of matrices:
	\[ \N(f \oplus h) = \begin{pmatrix}
	\left(f_{i,j} \right)_{\substack{1 \leq i \leq n\\1 \leq j \leq m}}
	& \left(\zero \right)_{\substack{1 \leq k \leq p\\1 \leq j \leq m}} \\
	\left(\zero \right)_{\substack{1 \leq i \leq n\\1 \leq \ell \leq q}} &
	\left(h_{k,\ell} \right)_{\substack{1 \leq k \leq p\\1 \leq \ell \leq q}}
	\end{pmatrix}\]
	The tensor corresponds to the usual Kronecker product of matrices: $\N(f \tensor h) =$
	\[ \begin{pmatrix}
	\left(\N(f_{1,1} \tensor h_{k,\ell}) \right)_{\substack{1 \leq k \leq p\\1 \leq \ell \leq q}} & \dots 
	& \left(\N(f_{n,1} \tensor h_{k,\ell}) \right)_{\substack{1 \leq k \leq p\\1 \leq \ell \leq q}}\\
	\vdots & & \vdots \\
	\left(\N(f_{1,m} \tensor h_{k,\ell}) \right)_{\substack{1 \leq k \leq p\\1 \leq \ell \leq q}}& \dots &
	\left(\N(f_{n,m} \tensor h_{k,\ell}) \right)_{\substack{1 \leq k \leq p\\1 \leq \ell \leq q}}
	\end{pmatrix}\]
\end{lemma}
\begin{proof}
	For the composition:
	\[ \begin{matrix} \N(g) \circ \N(f) &=& n_C \circ g \circ  n_B^{-1} \circ n_B \circ f \circ n_A^{-1}\\ & =&n_C \circ g  \circ f \circ n_A^{-1} \\&= &\N(g \circ f) \end{matrix}  \]
	For the biproduct and the tensor, we rely on \Cref{lem:proj_inj_matrix} and \Cref{lem:coef_by_coef} to make a coefficient-by-coefficient analysis. 
\end{proof}

\begin{proposition}\label{prop:maclane_bimonoidal}
	Let $f$ be a formal morphism obtained from $\lambda^{\tensor}$, $\lambda^{\tensor-1}$, $\alpha^{\tensor}$, $\alpha^{\tensor-1}$, $\lambda^{\oplus}$, $\lambda^{\oplus-1}$, $\alpha^{\oplus}$, $\alpha^{\oplus-1}$, through composition, $\tensor$, and $\oplus$. Then $\N(f) = \id$. It follows that any two such formal morphisms $f,g : A \to B$ are necessarily equal.
\end{proposition}
\begin{proof}
	We proceed by induction on $f$. For the base cases, we rely on \Cref{lem:proj_inj_matrix} to check that every coefficient of $\N(f)$ is $\id$ on the diagonal and $\zero$ otherwise. For the inductive case, we directly use \Cref{lem:formal_normalization_monoidal}.
\end{proof}

As a direct application of this proposition, we can define $\interp{A \isoML B}$ in a unique way:

\begin{definition}\label{def:maclane_blue}
	For $A$, a color of our graphical language, we define $\interp{A}_{\bfup{Form}}$ as the following formal object:
	\[ \begin{array}{cc} \interp{\tzero}_{\bfup{Form}} = \ozero & \interp{A \oplus B}_{\bfup{Form}} =  \interp{A}_{\bfup{Form}}  \oplus  \interp{B}_{\bfup{Form}} \\  \interp{\tone}_{\bfup{Form}} = \oone & \interp{A \tensor B}_{\bfup{Form}} =  \interp{A}_{\bfup{Form}}  \tensor  \interp{B}_{\bfup{Form}}   \end{array} \]
	We have $\interp{A} = U(\interp{A}_{\bfup{Form}})$. 
	If we assume $A  \isoML B$ as defined in \Cref{sec:objects_category_tensor_plus}, then there exists a unique formal morphism $\interp{A \isoML B}_{\bfup{Form}} : \interp{A} \to \interp{B}$ such that $\N(\interp{A \isoML B}_{\bfup{Form}}) = \id$. We then simply take $\interp{A \isoML B} = U(\interp{A \isoML B}_{\bfup{Form}})$.
\end{definition}

\begin{proposition}\label{prop:maclane_distributive}
	Let $f$ be a formal morphism obtained through composition, $\tensor$, and $\oplus$ from:
	\begin{itemize}
		\item $\lambda^{\tensor}$, $\lambda^{\tensor-1}$, $\alpha^{\tensor}$, $\alpha^{\tensor-1}$, $\lambda^{\oplus}$, $\lambda^{\oplus-1}$, $\alpha^{\oplus}$, $\alpha^{\oplus-1}$,
		\item $\sigma^{\oplus}$, $\dist_\ell$, $\dist_\ell^{-1}$, $\dist_r$, $\dist_r^{-1}$, \item $\zero$ whenever it is $A \tensor  \ozero \to \ozero$, $\ozero \tensor A \to \ozero$, $\ozero \to \ozero \tensor A$ or $\ozero \to A \tensor \ozero$.
	\end{itemize}
	Then $\N(f)$ is a permutation matrix, that is exactly one coefficient in every row and every column is $\id$ and the others are $\zero$. It follows that any two such formal morphisms $f,g : A \to B$ are equal if and only if they correspond to the same permutation.
\end{proposition}
\begin{proof}
	We proceed by induction on $f$. For the base cases, we rely on \Cref{lem:proj_inj_matrix} to check that every coefficient of $\N(f)$ is indeed $\id$ or $\zero$ with exactly one $\id$ per column and row. For the inductive case, we directly use \Cref{lem:formal_normalization_monoidal}.
\end{proof}

\subsection{The Parallel Monoidal Structure}

In the core of the paper, we claim that we can derive a third monoidal product that is the superposition of the other two. We prove that claim in this subsection. 

We now consider $(\bfup{H},\tensor,\oone,\oplus,\ozero)$ to be a distributive category as defined above, and additionally assume that $\tensor$ is symmetric.

\begin{proposition}\label{prop:parallel_is_monoidal}
	We define $A|B = (A \tensor B) \oplus (A \oplus B)$, and similarly $f|g = (f \tensor g) \oplus (f \oplus g)$. It is a bifunctor, and $(\bfup{H},|,\ozero)$ is a symmetric monoidal category.
\end{proposition}
\begin{proof}
	We start by proving that is a monoidal category, postponing the symmetric part. We do the proof in $\bfup{FormH}$, as any equality proved between formal morphisms leads to an equality between the corresponding morphisms of $\bfup{H}$. The unitors are defined as follows:
	\[ \begin{array}{ll} \lambda^| : & \ozero | A \\
	& = \\
	& (\ozero \tensor A) \oplus (\ozero \oplus A) \\
	& \downarrow \zero \oplus \lambda^\oplus \\
	& \ozero \oplus A \\
	& \downarrow \lambda^{\oplus} \\
	& A \\
	\end{array} \qquad \begin{array}{ll} \rho^| : & A|\ozero \\
	& = \\
	& (A \tensor \ozero) \oplus (A \oplus \ozero) \\
	& \downarrow \zero \oplus \rho^\oplus \\
	& \ozero \oplus A \\
	& \downarrow \lambda^{\oplus} \\
	& A \\
	\end{array} \]
	Defining the associator properly is more complex. While we provide an explicit definition below (omitting the $\tensor$ for readability), we will instead work with its normalization.
	\[ \begin{array}{ll} \alpha^| : & (A | B)|C \\
	& = \\
	& \left(AB \oplus (A \oplus B)\right)C \oplus \left(\left(AB \oplus (A \oplus B)\right) \oplus C\right) \\
	& \downarrow \dist_r \oplus \id\\
	& \left((AB)C \oplus (A \oplus B)C\right) \oplus \left(\left(AB \oplus (A \oplus B)\right) \oplus C\right) \\
	&\downarrow (\alpha^{\tensor} \oplus \dist_r) \oplus \id \\	
	& \left(A(BC) \oplus (AC \oplus BC)\right) \oplus \left(\left(AB \oplus (A \oplus B)\right) \oplus C\right) \\
	& \downarrow m^{\oplus} \\
	& A(BC) \oplus \left(AC \oplus \left(\left(BC \oplus AB\right) \oplus \left(A \oplus \left(B \oplus C\right)\right)\right)\right) \\
	& \downarrow \id \oplus (\id \oplus (\sigma^\oplus \oplus \id)) \\
	& A(BC) \oplus \left(AC \oplus \left(\left(AB \oplus BC\right) \oplus \left(A \oplus \left(B \oplus C\right)\right)\right)\right) \\
	& \downarrow m^{\oplus} \\
	& \left(A(BC) \oplus (AC \oplus AB)\right) \oplus \left(\left(BC \oplus A\right) \oplus (B \oplus C)\right) \\
	& \downarrow (\id \oplus \sigma^{\oplus}) \oplus (\sigma^\oplus \oplus \id)\\
	& \left(A(BC) \oplus (AB \oplus AC)\right) \oplus \left(\left(A \oplus BC\right) \oplus (B \oplus C)\right) \\
	& \downarrow m^{\oplus} \\
	& \left(A(BC) \oplus (AB \oplus AC)\right) \oplus \left(A \oplus \left(B C \oplus (B \oplus C)\right)\right) \\
	& \downarrow (\id \oplus \dist_\ell^{-1}) \oplus \id \\
	& \left(A(BC) \oplus A(B \oplus C)\right) \oplus \left(A \oplus \left(B C \oplus (B \oplus C)\right)\right) \\
	& \downarrow \dist_\ell^{-1} \oplus \id \\
	& A\left(B C \oplus (B \oplus C)\right) \oplus \left(A \oplus \left(B C \oplus (B \oplus C)\right)\right) \\
	&=\\
	& A| (B|C) \\
	\end{array} \]
	And here is the matrix of $\N(\alpha^|)$, with annotation indicating what each row and column corresponds to, and leaving the cell empty when the morphism would be $\zero$:
	\[
	\begin{blockarray}{cccccccc}
			& \text{\scalebox{0.7}{$(A\tensor B) \tensor C$}} & \text{\scalebox{0.7}{$A\tensor C$}} & \text{\scalebox{0.7}{$B\tensor C$}} & \text{\scalebox{0.7}{$A \tensor B$}} & \text{\scalebox{0.7}{$A$}} & \text{\scalebox{0.7}{$B$}} & \text{\scalebox{0.7}{$C$}} \\
			\begin{block}{c(ccccccc)}
				\text{\scalebox{0.7}{$(A\tensor B)\tensor C$}}  & \id &&&&&&\\
				\text{\scalebox{0.7}{$A\tensor B$}}  &&&& \id &&& \\
				\text{\scalebox{0.7}{$A\tensor C$}}  && \id &&&&& \\
				\text{\scalebox{0.7}{$A$}}  &&&&& \id && \\
				\text{\scalebox{0.7}{$B\tensor C$}}  &&& \id &&&& \\
				\text{\scalebox{0.7}{$B$}}  &&&&&& \id & \\
				\text{\scalebox{0.7}{$C$}}  &&&&&&& \id  \\
			\end{block}
	\end{blockarray} 
	\]
	We note that it is a permutation matrix, which we already knew from \Cref{prop:maclane_distributive}. In fact, we can rely on \Cref{prop:maclane_distributive} to prove all the coherence laws of monoidal categories just by checking that both sides of the equality correspond to the same permutation. And said permutations can easily be computed using \Cref{lem:formal_normalization_monoidal}.	
	We note that all the unitors and associator are composed of natural isomorphisms, and hence are natural isomorphisms.

	We postponed the case of the symmetry, so we come back to it. We simply define 
	\[ \sigma^| = \sigma^{\tensor} \oplus \sigma^{\oplus} \]
	It is again a natural isomorphism as it is composed of natural isomorphisms. The matrix of $\N(\sigma^|)$ is:	
	\[
	\begin{blockarray}{cccc}
	& \text{\scalebox{0.7}{$A\tensor B$}}  & \text{\scalebox{0.7}{$A$}} & \text{\scalebox{0.7}{$B$}}  \\
	\begin{block}{c(ccc)}
	\text{\scalebox{0.7}{$B\tensor A$}}  &\sigma^{\tensor}&& \\
	\text{\scalebox{0.7}{$B$}}  && \id & \\
	\text{\scalebox{0.7}{$A$}}  && \id & \\
	\end{block}
	\end{blockarray} 
	\]
	The coherence laws associated to the symmetry can simply be checked by relying on the matrix notation and \Cref{lem:formal_normalization_monoidal} to compose those matrices.	
\end{proof}

\begin{figure*}
	\[ \interp{\tikzfig{lang/id}} :=
	\begin{blockarray}{cc}
	& \text{\scalebox{0.7}{$\interp{A}$}} \\
	\begin{block}{c(c)}
	\text{\scalebox{0.7}{$\interp{A}$}}  & \id \\
	\end{block}
	\end{blockarray} 
	\qquad \interp{\tikzfig{lang/swap}} :=
	\begin{blockarray}{cccc}
	& \text{\scalebox{0.7}{$\interp{A}\tensor \interp{B}$}} & \text{\scalebox{0.7}{$\interp{A}$}} & \text{\scalebox{0.7}{$\interp{B}$}} \\
	\begin{block}{c(ccc)}
	\text{\scalebox{0.7}{$\interp{B} \tensor \interp{A} $}}  & \sigma^{\tensor} &&\\
	\text{\scalebox{0.7}{$\interp{B}$}}  &&&\id\\
	\text{\scalebox{0.7}{$\interp{A}$}}  &&\id&\\
	\end{block}
	\end{blockarray} 
	\qquad \interp{\tikzfig{lang/scal}} :=
	\begin{blockarray}{cc}
	& \text{\scalebox{0.7}{$\interp{A}$}} \\
	\begin{block}{c(c)}
	\text{\scalebox{0.7}{$\interp{A}$}}  & s \cdot \id \\
	\end{block}
	\end{blockarray} 
	\]\[ \interp{\tikzfig{lang/plus}} :=
	\begin{blockarray}{cccc}
	& \text{\scalebox{0.7}{$\interp{A}\tensor \interp{B}$}} & \text{\scalebox{0.7}{$\interp{A}$}} & \text{\scalebox{0.7}{$\interp{B}$}} \\
	\begin{block}{c(ccc)}
	\text{\scalebox{0.7}{$\interp{A}$}}  &&\id&\\
	\text{\scalebox{0.7}{$\interp{B}$}}  &&&\id\\
	\end{block}
	\end{blockarray} 
	\qquad \interp{\tikzfig{lang/plus-inv}} :=
	\begin{blockarray}{ccc}
	& \text{\scalebox{0.7}{$\interp{A}$}} & \text{\scalebox{0.7}{$\interp{B}$}} \\
	\begin{block}{c(cc)}
	\text{\scalebox{0.7}{$\interp{A} \tensor \interp{B} $}}  &&\\
	\text{\scalebox{0.7}{$\interp{A}$}}  &\id&\\
	\text{\scalebox{0.7}{$\interp{B}$}}  &&\id\\
	\end{block}
	\end{blockarray} 
	\]\[ \interp{\tikzfig{lang/tensor-PN}} := 
	\begin{blockarray}{cccc}
	& \text{\scalebox{0.7}{$\interp{A}\tensor \interp{B}$}} & \text{\scalebox{0.7}{$\interp{A}$}} & \text{\scalebox{0.7}{$\interp{B}$}} \\
	\begin{block}{c(ccc)}
	\text{\scalebox{0.7}{$\interp{A} \tensor \interp{B} $}}  & \id &&\\
	\end{block}
	\end{blockarray} 
	\qquad \interp{\tikzfig{lang/tensor-PN-inv}} := 
	\begin{blockarray}{cc}
	& \text{\scalebox{0.7}{$\interp{A}\tensor \interp{B}$}}  \\
	\begin{block}{c(c)}
	\text{\scalebox{0.7}{$\interp{A} \tensor \interp{B} $}}  & \id\\
	\text{\scalebox{0.7}{$\interp{A}$}}  &\\
	\text{\scalebox{0.7}{$\interp{B}$}}  &\\
	\end{block}
	\end{blockarray} 
	\]\[\interp{\tikzfig{lang/contraction}} := 
	\begin{blockarray}{cccc}
	& \text{\scalebox{0.7}{$\interp{A}\tensor \interp{A}$}} & \text{\scalebox{0.7}{$\interp{A}$}} & \text{\scalebox{0.7}{$\interp{A}$}} \\
	\begin{block}{c(ccc)}
	\text{\scalebox{0.7}{$\interp{A}$}}  &&\id&\id\\
	\end{block}
	\end{blockarray} 
	\qquad \interp{\tikzfig{lang/contraction-inv}} := 
	\begin{blockarray}{cc}
	&  \text{\scalebox{0.7}{$\interp{A}$}} \\
	\begin{block}{c(c)}
	\text{\scalebox{0.7}{$\interp{A} \tensor \interp{A} $}}  &\\
	\text{\scalebox{0.7}{$\interp{A}$}}  &\id\\
	\text{\scalebox{0.7}{$\interp{A}$}}  &\id\\
	\end{block}
	\end{blockarray} 
	\qquad \interp{\tikzfig{lang/null}} := 
	\begin{blockarray}{cc}
	&  \\
	\begin{block}{c(c)}
	\text{\scalebox{0.7}{$\interp{A}$}}  & \\
	\end{block}
	\end{blockarray} 
	\qquad \interp{\tikzfig{lang/null-inv}} := 
	\begin{blockarray}{cc}
	& \text{\scalebox{0.7}{$\interp{A}$}} \\
	\begin{block}{c(c)}
	& \\
	\end{block}
	\end{blockarray} 
	\]\[ \interp{\tikzfig{lang/adapt}} := 
	\begin{blockarray}{cc}
	& \text{\scalebox{0.7}{$\interp{A}$}} \\
	\begin{block}{c(c)}
	\text{\scalebox{0.7}{$\interp{A'}$}}  & \interp{A \isoML A'} \\
	\end{block}
	\end{blockarray}
	\qquad \interp{e \circ d} := 
	\interp{e} \circ \interp{d} 
	\qquad \interp{d \parallel e} := 
	m^| \circ (\interp{d} | \interp{e}) \circ {m^|} 
	\]
	
	where $\interp{A \isoML A'}$,  $m^{\tensor}$, $m^{\oplus}$, and $m^|$ always correspond to some permutation matrices.
	\caption{Matrix Semantics for the Functional \Langage}
	\label{fig:mat_sem_fun}
\end{figure*}

\begin{figure*}
	\[ \interp{f}^{\oplus \oone} := 
	\begin{blockarray}{ccc}
	& \text{\scalebox{0.7}{$\interp{\wireSetA}$}} & \text{\scalebox{0.7}{$\oone$}}  \\
	\begin{block}{c(cc)}
	\text{\scalebox{0.7}{$\interp{\wireSetB}$}}  & \interp{f} &\\
	\text{\scalebox{0.7}{$\oone$}}  &&\id\\
	\end{block}
	\end{blockarray}
	\text{ whenever } f \in \FCat(\wireSetA,\wireSetB) 
	\qquad \interp{\tikzfig{lang/unit}}^{\oplus \oone}  := 
	\begin{blockarray}{cc}
	& \text{\scalebox{0.7}{$\oone$}}  \\
	\begin{block}{c(c)}
	\text{\scalebox{0.7}{$\interp{\tone}$}}  & \id \\
	\text{\scalebox{0.7}{$\oone$}}  &\id\\
	\end{block}
	\end{blockarray} 
	\qquad  \interp{\tikzfig{lang/unit-inv}}^{\oplus \oone}  := 
	\begin{blockarray}{ccc}
	& \text{\scalebox{0.7}{$\interp{\tone}$}} & \text{\scalebox{0.7}{$\oone$}}  \\
	\begin{block}{c(cc)}
	\text{\scalebox{0.7}{$\oone$}}  & \id &\id\\
	\end{block}
	\end{blockarray}
	\]\[ \interp{e \circ d}^{\oplus \oone}  := 
	\interp{e}^{\oplus \oone}  \circ \interp{d}^{\oplus \oone}  
	\qquad \interp{d \parallel e}^{\oplus \oone}  := 
	(m^| \oplus \id_{\oone}) \circ \bfup{expand} \circ (\interp{d}^{\oplus \oone}  \tensor \interp{e}^{\oplus \oone}) \circ \bfup{expand}^{-1} \circ ({m^|}' \oplus \id_{\oone}) \]
	\caption{Matrix Semantics for the \Langage}
	\label{fig:mat_sem}
\end{figure*}

\subsection{Adding the Unit: the "$\oplus \oone$" Category}

In order to model the Unit of our graphical language, we need to add $\oplus \oone$ to all of our objects.

\begin{definition}
	We define $\bfup{H}^{\oplus \oone}$ as the category with the same objects as $\bfup{H}$ and for morphisms $f \in \bfup{H}^{\oplus \oone}(H,K)$ the morphisms of $\bfup{H}(H\oplus \oone,K \oplus \oone)$ of the form 	
	\[ f = g \bfup{+} (\iota_r \circ \pi_r)\]
	 for some $g \in \bfup{H}(H\oplus \oone,K \oplus \oone)$. In matrix terms, it is a matrix of $\bfup{H}(H\oplus \oone,K \oplus \oone)$ where the bottom-right coefficient is of the form $c + \id$ for $c \in \bfup{H}(\oone,\oone)$. The identity and composition are the same as in $\bfup{H}$.
\end{definition}
In particular, any morphism of $\bfup{H}(H\oplus \oone,K \oplus \oone)$ of the form $h \oplus \id$ is also of the form $(h \oplus \zero) \bfup{+} (\iota_r \circ \pi_r)$, so is in $\bfup{H}^{\oplus \oone}$. A morphism of $\bfup{H}^{\oplus \oone}$ of the form $h \oplus \id$ is said \textbf{functional}\footnote{This name reflects that we will interpret the $\oplus \oone$ on the domain as "the morphism spontaneously computes something even when no input is given" and the $\oplus \oone$ on the codomain as "the morphism does not output anything". As such, functional morphisms are morphisms that output something if and only if they had an input.}.

	We define the natural isomorphism $\bfup{expand} : (H \oplus \oone) \tensor (K \oplus \oone) \to (H|K) \oplus \oone$ as the following composition:
	\[ \begin{matrix} (H \oplus \oone) \tensor (K \oplus \oone)\hspace{-0.4cm} &\xrightarrow{\bfup{dist}_r}& H \tensor (K \oplus \oone) \oplus \oone \tensor (K \oplus \oone) \\ &\xrightarrow{\bfup{dist}_\ell \oplus \lambda^{\tensor}}& (H \tensor K \oplus H \tensor \oone) \oplus (K \oplus \oone) \\ &\xrightarrow{(\id \oplus \rho^{\tensor}) \oplus \id}& (H \tensor K \oplus H) \oplus (K \oplus \oone) \\ &\xrightarrow{m^{\oplus}}& ((H \tensor K) \oplus (H \oplus K)) \oplus \oone \\  \end{matrix}  \]
	
\begin{proposition}
	The category $\bfup{H}^{\oplus \oone}$ has a symmetric monoidal structure, defined as $H|K$ on objects and $\bfup{expand} \circ (f \tensor g) \circ \bfup{expand}^{-1}$. All its structural morphisms -- left-unitor, right-unitor, associator, swap -- are functional morphisms.
\end{proposition}
\begin{proof}
	The symmetric monoidal structure follows from the symmetric monoidal structure of $(\bfup{H},|,\ozero)$, by simply taking $\lambda^|\oplus\id, \rho^|\oplus\id, \alpha^|\oplus\id, \lambda^|\oplus\id, \sigma^|\oplus \id$ as structural morphisms, and by remarking that
	\[ \bfup{expand} \circ ((f \oplus \id) \tensor (g \oplus \id) ) \circ \bfup{expand}^{-1} = (f|g) \oplus \id\]
\end{proof}

	\subsection{Adding the Unit: the Double Kleisli Category}
	
This section explore an alternative way the Unit could have been interpreted semantically, relying on a Kleisli construction rather than the more ad hoc "$\oplus \oone$" construction. 

While Kleisli categories have an extensive literature \cite{Maranda66}, we will focus in this section on a specific category, hence will skip the background needed to define the general case. The object of our attention is the Double Kleisli \cite{doublekleisli93} -- which is a co-Kleisli category of a Klesil category -- of $\bfup{H}$ with respect to $\_ \mapsto (\_ \oplus \oone)$, which is both a monad and co-monad.

\begin{definition}
	The category $\bfup{KcK-H}$ is the category with the same objects as $\bfup{H}$ and with for morphisms $\bfup{KcK-H}(H,K)$ the morphisms of $\bfup{H}(H\oplus \oone,K \oplus \oone)$, but with a significantly different composition $\circ_{\textup{KcK}}$. For $f \in \bfup{KcK-H}(H,K)$ and $g \in \bfup{KcK-H}(K,L)$:
	\[  \begin{array}{rrl} g \circ_{\textup{KcK}} f : H \oplus \oone
	&\xrightarrow{\id \oplus \Delta} & H \oplus (\oone\oplus\oone)  \\
	& \xrightarrow{\alpha^{\oplus-1}} & (H\oplus\oone) \oplus\oone \\
	& \xrightarrow{f \oplus \id} & (K \oplus \oone ) \oplus\oone \\
	& \xrightarrow{\alpha^{\oplus}} & K \oplus (\oone \oplus \oone) \\
	& \xrightarrow{\id \oplus \sigma^{\oplus}} & K \oplus (\oone \oplus \oone) \\
	& \xrightarrow{\alpha^{\oplus-1}} & (K \oplus \oone ) \oplus\oone \\
	& \xrightarrow{g \oplus \id}	& (L \oplus \oone ) \oplus\oone \\
	& \xrightarrow{\alpha^{\oplus}}	& L \oplus (\oone \oplus \oone) \\
	& \xrightarrow{\id \oplus \nabla}	& L \oplus \oone
	\end{array}  \]
	Or, as a picture, with $\alpha^\oplus$ and $\alpha^{\oplus-1}$ implicit:
	\begin{center}
		\tikzfig{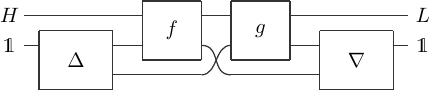}
	\end{center}
	The identity for this composition is $\id\oplus\zero : H\oplus\oone \to H\oplus\oone$, and more generally $\_ \mapsto \_\oplus\zero$ is a faithful functor from $\bfup{H}$ to $\bfup{KcK-H}$.
\end{definition}
\begin{proof}
	The main property to prove is the associativity of the composition.Using the properties of the biproduct we obtain:
	\[ \nabla \circ (\nabla \oplus \id) = \nabla \circ (\id \oplus \nabla) \circ \alpha^\oplus  \]
	\[  (\Delta \oplus \id) \circ \Delta = \alpha^{\oplus-1} \circ (\id \oplus \Delta) \circ \Delta  \]
	which directly entails the associativity of $\circ_{\textup{KcK}}$.
\end{proof}

\begin{proposition}
	If $\bfup{H}$ is cancellative, that is
	\[ f \bfup{+} h = g \bfup{+} h \implies f = g \]
	then the categories $\bfup{H}^{\oplus \oone}$ and $\bfup{KcK-H}$ are equivalent through the following identity-on-objects equivalence. The morphism $f \in \bfup{KcK-H}(H,K)$ corresponds to $f \bfup{+} (\iota_r \circ \pi_r) \in \bfup{H}^{\oplus \oone}(H,K)$.
\end{proposition}
\begin{proof}
	We easily check that:
	\[(\id\oplus \zero) \bfup{+}  (\iota_r \circ \pi_r) = \id\]
	\[ (g  \bfup{+}  (\iota_r \circ \pi_r)) \circ (f \bfup{+}  (\iota_r \circ \pi_r)) = (g \circ_{\textup{KcK}} f)  \bfup{+}  (\iota_r \circ \pi_r)  \]
\end{proof}

\subsection{Matrix Semantics}

Using the matrix notations previously defined, and keeping the convention of annotating rows and column, and leaving an empty cell whenever the morphism of that cell would be $\zero$, we can rewrite \Cref{fig:sem_fun,fig:sem} into \Cref{fig:mat_sem_fun,fig:mat_sem}. We note that the isomorphism $\bfup{scal}: R \to \bfup{H}(\oone,\oone)$ is kept implicit.

\clearpage
\section{Soundness of the Equational Theory}
\label{app:soudness}

In this appendix, we provide a proof of soundness of the equational theory.
And intuitive proof of soundness can be done relatively easily, by looking at how \emph{data} would flow across the diagram before and after the application of an equation. However, to formalize such an intuitive proof we would need to define a token-based semantics and to prove that this token-based semantics is equivalent to the categorical semantics. Instead, we provide a more direct but maybe less intuitive proof.
\begin{enumerate}
	\item We start by all the equations easy to prove sound in a direct way, including a couple of equations that are technically superfluous (the bracketed ones) but are useful as stepping stones, that is
	\begin{itemize}
		\item The equations on the left half of   \Cref{fig:eq_tensor_plus}, that is \eqT, \eqP, \eqbot, \eqbotC, \eqN{} and  \eqSigmaC.
		\item All the equations from  \Cref{fig:eq_null_nat}, that is \eqTN, \eqPN, \eqCN, \eqNN, \eqNM, \eqNS, \eqTNleft, \eqNPPleft, \eqRhoC{} and \eqLambdaC.
		\item All the equations of \Cref{fig:eq_maclane} except the first column, that is \eqTM, \eqPM, \eqLambdaRhoT, \eqRhoP, \eqLambdaP, \eqMM{} and \eqM.
		\item All the equations from  \Cref{fig:eq_scalar}, that is \eqSS, \eqS, \eqSzero, \eqSsum, \eqST{} and \eqSP.
	\end{itemize}
	\item We continue with a set of relatively difficult equations to check directly, and we brute-force them by computing the matrix semantics (see \Cref{fig:mat_sem_fun}) of both sides of the equations and remarking that they are identical. More precisely
	\begin{itemize}
		\item The equations \eqTT{} and \eqmix{} from   \Cref{fig:eq_tensor_plus}.
		\item The equations \eqAlphaT{} and \eqAlphaP{} from  \Cref{fig:eq_maclane}.
	\end{itemize}
	\item We then prove a couple of powerful lemmas, which allows us to derive the soundness of the remaining equations for the functional language, that is
	\begin{itemize}
		\item The equations \eqPP{} and \eqPtoC{} from  \Cref{fig:eq_tensor_plus}.
		\item All the equations from \Cref{fig:eq_contraction_nat}, that is \eqTC, \eqTCleft, \eqPC, \eqPPCleft, \eqCC, \eqAlphaC, \eqSC{} and \eqMC.
	\end{itemize}
	\item We finish by handling the additional equations of the equational theory for the full language, that is \eqRhoT, \eqUN{} and \eqUS{} from \Cref{fig:eq_unit}. Since \eqLambdaT, \eqUsum{} and \eqUpar{} can be deduced from the others, we do not need to prove their soundness.
\end{enumerate}
Implicitly in all those steps is that the mirrored up-down version of the equations can be proven sound in similar ways, as the all the properties we use of $\bfup{H}$ are also true in its dual category.

\subsection{The Easy Equations}
Now, let us start with the following equations.
\[\tikzfig{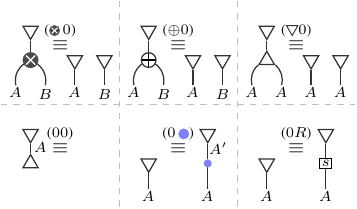}\]

Those are sound because in a semiadditive category, $\zero$ satisfies for any $f$: $\zero \circ f = \zero$, $f \circ \zero = \zero$, $\zero \tensor f = \zero$, and $f \tensor \zero = \zero$. In the remaining of the proof, those properties of $\zero$ will be used implicitly.
\[\tikzfig{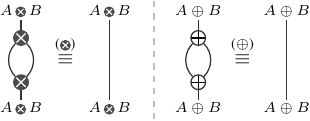}\]

Those are sound because in a semiadditive category:\[\pi_\ell \circ \iota_\ell = \id = \pi_r \circ \iota_r\]

\[\tikzfig{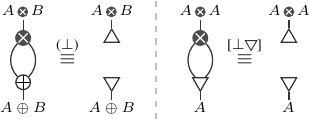}\]
Those are sound because in a semiadditive category:\[\pi_r \circ \iota_\ell = \zero\]

\[\tikzfig{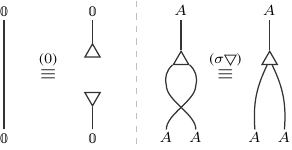}\]
Those are sound because  $\sigma^| = \sigma^{\tensor} \oplus \sigma^{\oplus}$, and in a semiadditive category all the morphisms on $\ozero \to \ozero$ are equal, $(f \oplus g) \circ \iota_r = \iota_r \circ g$ and $\sigma^\oplus \circ \Delta = \sigma^\oplus$.

\[\tikzfig{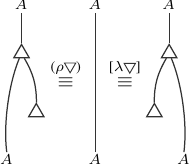}\]
The left one is sound because $\id | \zero = \zero \oplus (\id \oplus \zero)$ and in a semiadditive category  $(f \oplus g) \circ \iota_r = \iota_r \circ g$ and $\Delta \circ (\id  \oplus \zero) = \id$. The right one is sound for similar reasons.

\[\tikzfig{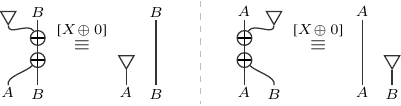}\]
The left one is sound because $\zero | \id = \zero \oplus (\zero \oplus \id)$ and in a semiadditive category $\pi_r \circ \iota_r = \zero \oplus \id$ and $\oplus$ is a bifunctor. . The right one is sound for similar reasons.
\[\tikzfig{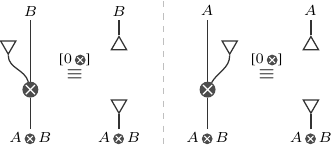}\]
The left one is sound because  $\zero | \id = \zero \oplus (\zero \oplus \id)$ and in a semiadditive category $(f \oplus g) \circ \iota_\ell = \iota_\ell \circ f$.
\[\tikzfig{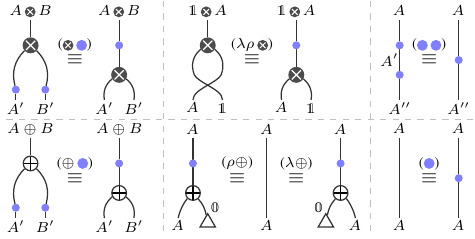}\]
Those equations all rely on the generalisation of Mac Lane's coherence theorem stated in  \Cref{prop:maclane_bimonoidal}, together with the following facts:
\begin{itemize}
	\item In a semiadditive category, $(f \oplus g) \circ \iota_\ell = \iota_\ell \circ f$ and $(f \oplus g) \circ \iota_r = \iota_r \circ g$.
	\item $\sigma^| = \sigma^{\tensor} \oplus \sigma^{\oplus}$.
	\item  $\id | \zero = \zero \oplus (\id \oplus \zero)$ and $\zero | \id = \zero \oplus (\zero \oplus \id)$.
\end{itemize}
\[\tikzfig{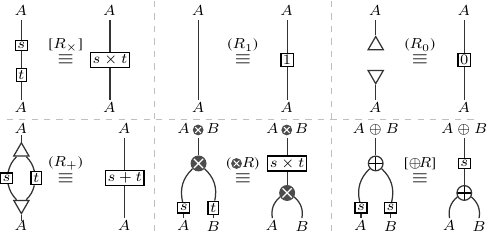}\]
Those equations all rely on \bfup{scal} being an isomorphism of semirings, together with the fact that in a semiadditive category, $(f \oplus g) \circ \iota_\ell = \iota_\ell \circ f$,  $(f \oplus g) \circ \iota_r = \iota_r \circ g$ and $\nabla \circ (f \oplus g) \circ \Delta = f \bfup{+} g$, and additionally in \bfup{H} we have distributivity hence $s \cdot (f \oplus g) = (s \cdot f \oplus s \cdot g)$.

\subsection{The Hard Equations}
\begin{figure*}
$
\begin{blockarray}{c ccccc ccccc ccccc}
	&\text{\scalebox{0.7}{((A$\alpha$)B)$\beta$}}
	&\text{\scalebox{0.7}{(AB)$\beta$}}
	&\text{\scalebox{0.7}{($\alpha$B)$\beta$}}
	&\text{\scalebox{0.7}{(A$\alpha$)$\beta$}}
	&\text{\scalebox{0.7}{A$\beta$}}
	&\text{\scalebox{0.7}{$\alpha\beta$}}
	&\text{\scalebox{0.7}{B$\beta$}}
	&\text{\scalebox{0.7}{(A$\alpha$)B}}
	&\text{\scalebox{0.7}{AB}}
	&\text{\scalebox{0.7}{$\alpha$B}}
	&\text{\scalebox{0.7}{A$\alpha$}}
	&\text{\scalebox{0.7}{A}}
	&\text{\scalebox{0.7}{$\alpha$}}
	&\text{\scalebox{0.7}{B}}
	&\text{\scalebox{0.7}{$\beta$}}\\
	\begin{block}{c(ccccc ccccc ccccc)}
		\text{\scalebox{0.7}{AB}}  &&&&&\id &\id&&\id&\id& &&&&&\\
		\text{\scalebox{0.7}{A}}  &&&&& &&&&& &&\id&\id&&\\
		\text{\scalebox{0.7}{B}}  &&&&& &&&&& &&&&\id&\id\\
	\end{block}
\end{blockarray} $\\
$\circ~
\begin{blockarray}{c ccccc ccccc ccccc}
	&\text{\scalebox{0.7}{((AB)$\alpha$)$\beta$}}
	&\text{\scalebox{0.7}{(A$\alpha$)$\beta$}}
	&\text{\scalebox{0.7}{(B$\alpha$)$\beta$}}
	&\text{\scalebox{0.7}{(AB)$\beta$}}
	&\text{\scalebox{0.7}{A$\beta$}}
	&\text{\scalebox{0.7}{B$\beta$}}
	&\text{\scalebox{0.7}{$\alpha\beta$}}
	&\text{\scalebox{0.7}{(AB)$\alpha$}}
	&\text{\scalebox{0.7}{A$\alpha$}}
	&\text{\scalebox{0.7}{B$\alpha$}}
	&\text{\scalebox{0.7}{AB}}
	&\text{\scalebox{0.7}{A}}
	&\text{\scalebox{0.7}{B}}
	&\text{\scalebox{0.7}{$\alpha$}}
	\text{\scalebox{0.7}{$\beta$}}\\
	\begin{block}{c(ccccc ccccc ccccc)}
		\text{\scalebox{0.7}{((A$\alpha$)B)$\beta$}}  &f&&&& &&&&& &&&&&\\
		\text{\scalebox{0.7}{(AB)$\beta$}} &&&&\id& &&&&& &&&&&\\
		\text{\scalebox{0.7}{($\alpha$B)$\beta$}} &&&g&& &&&&& &&&&&\\
		\text{\scalebox{0.7}{(A$\alpha$)$\beta$}}  &&\id&&& &&&&& &&&&&\\
		\text{\scalebox{0.7}{A$\beta$}} &&&&&\id &&&&& &&&&&\\
		\text{\scalebox{0.7}{$\alpha\beta$}} &&&&& &&\id&&& &&&&&\\
		\text{\scalebox{0.7}{B$\beta$}} &&&&& &\id&&&& &&&&&\\
		\text{\scalebox{0.7}{(A$\alpha$)B}} &&&&& &&&h&& &&&&&\\
		\text{\scalebox{0.7}{AB}} &&&&& &&&&& &\id&&&&\\
		\text{\scalebox{0.7}{$\alpha$B}}&&&&& \sigma^{\tensor} &&&&& &&&&&\\
		\text{\scalebox{0.7}{A$\alpha$}} &&&&\id& &&&&& &&&&&\\
		\text{\scalebox{0.7}{A}}&&&&& &&&&& &&\id&&&\\
		\text{\scalebox{0.7}{$\alpha$}} &&&&& &&&&& &&&&\id&\\
		\text{\scalebox{0.7}{B}}&&&&& &&&&& &&&\id&&\\
		\text{\scalebox{0.7}{$\beta$}} &&&&& &&&&& &&&&&\id\\
	\end{block}
\end{blockarray}
~\circ~
\begin{blockarray}{c ccccc}
	&\text{\scalebox{0.7}{(AB)$\alpha$}}
	&\text{\scalebox{0.7}{(AB)$\beta$}}
	&\text{\scalebox{0.7}{AB}}
	&\text{\scalebox{0.7}{$\alpha$}}
	&\text{\scalebox{0.7}{$\beta$}}\\
	\begin{block}{c(ccccc)}
		\text{\scalebox{0.7}{((AB)$\alpha$)$\beta$}}  &&&&&\\
		\text{\scalebox{0.7}{(A$\alpha$)$\beta$}}    &&&&&\\
		\text{\scalebox{0.7}{(B$\alpha$)$\beta$}}  &&&&&\\
		\text{\scalebox{0.7}{(AB)$\beta$}}  &&\id&&&\\
		\text{\scalebox{0.7}{A$\beta$}}  &&&&&\\
		\text{\scalebox{0.7}{B$\beta$}}  &&&&&\\
		\text{\scalebox{0.7}{$\alpha\beta$}}  &&&&&\\
		\text{\scalebox{0.7}{(AB)$\alpha$}}  &\id&&&&\\
		\text{\scalebox{0.7}{A$\alpha$}}  &&&&&\\
		\text{\scalebox{0.7}{B$\alpha$}} &&&&&\\
		\text{\scalebox{0.7}{AB}} &&&\id&&\\
		\text{\scalebox{0.7}{A}}  &&&&&\\
		\text{\scalebox{0.7}{B}} &&&&&\\
		\text{\scalebox{0.7}{$\alpha$}}  &&&&\id&\\
		\text{\scalebox{0.7}{$\beta$}}   &&&&&\id\\
	\end{block}
\end{blockarray}$\\
$\circ~
\begin{blockarray}{c ccccc ccccc ccccc}
	&\text{\scalebox{0.7}{((AB)$\alpha$)$\beta$}}
	&\text{\scalebox{0.7}{(A$\alpha$)$\beta$}}
	&\text{\scalebox{0.7}{(B$\alpha$)$\beta$}}
	&\text{\scalebox{0.7}{(AB)$\beta$}}
	&\text{\scalebox{0.7}{A$\beta$}}
	&\text{\scalebox{0.7}{B$\beta$}}
	&\text{\scalebox{0.7}{$\alpha\beta$}}
	&\text{\scalebox{0.7}{(AB)$\alpha$}}
	&\text{\scalebox{0.7}{A$\alpha$}}
	&\text{\scalebox{0.7}{B$\alpha$}}
	&\text{\scalebox{0.7}{AB}}
	&\text{\scalebox{0.7}{A}}
	&\text{\scalebox{0.7}{B}}
	&\text{\scalebox{0.7}{$\alpha$}}
	&\text{\scalebox{0.7}{$\beta$}}\\
	\begin{block}{c(ccccc ccccc ccccc)}
		\text{\scalebox{0.7}{(AB)$\alpha$}} &&&&& &&&\id&& &&&&&\\
		\text{\scalebox{0.7}{(AB)$\beta$}} &&&&\id& &&&&& &&&&&\\
		\text{\scalebox{0.7}{AB}} &&&&& &&&&& &\id&&&&\\
		\text{\scalebox{0.7}{$\alpha$}} &&&&& &&&&& &&&&\id&\\
		\text{\scalebox{0.7}{$\beta$}} &&&&& &&&&& &&&&&\id\\
	\end{block}
\end{blockarray}$\\
	$\circ~
	\begin{blockarray}{c ccccc ccccc ccccc}
		&\text{\scalebox{0.7}{((A$\alpha$)B)$\beta$}}
		&\text{\scalebox{0.7}{(AB)$\beta$}}
		&\text{\scalebox{0.7}{($\alpha$B)$\beta$}}
		&\text{\scalebox{0.7}{(A$\alpha$)$\beta$}}
		&\text{\scalebox{0.7}{A$\beta$}}
		&\text{\scalebox{0.7}{$\alpha\beta$}}
		&\text{\scalebox{0.7}{B$\beta$}}
		&\text{\scalebox{0.7}{(A$\alpha$)B}}
		&\text{\scalebox{0.7}{AB}}
		&\text{\scalebox{0.7}{$\alpha$B}}
		&\text{\scalebox{0.7}{A$\alpha$}}
		&\text{\scalebox{0.7}{A}}
		&\text{\scalebox{0.7}{$\alpha$}}
		&\text{\scalebox{0.7}{B}}
		&\text{\scalebox{0.7}{$\beta$}}\\
		\begin{block}{c(ccccc ccccc ccccc)}
			\text{\scalebox{0.7}{((AB)$\alpha$)$\beta$}}  &f&&&& &&&&& &&&&&\\
			\text{\scalebox{0.7}{(A$\alpha$)$\beta$}}   &&&&\id& &&&&& &&&&&\\
			\text{\scalebox{0.7}{(B$\alpha$)$\beta$}}  &&&g&& &&&&& &&&&&\\
			\text{\scalebox{0.7}{(AB)$\beta$}}  &&\id&&& &&&&& &&&&&\\
			\text{\scalebox{0.7}{A$\beta$}}  &&&&&\id &&&&& &&&&&\\
			\text{\scalebox{0.7}{B$\beta$}}   &&&&& &&\id&&& &&&&&\\
			\text{\scalebox{0.7}{$\alpha\beta$}}   &&&&& &\id&&&& &&&&&\\
			\text{\scalebox{0.7}{(AB)$\alpha$}}   &&&&& &&&h&& &&&&&\\
			\text{\scalebox{0.7}{A$\alpha$}}  &&&&& &&&&& &\id&&&&\\
			\text{\scalebox{0.7}{B$\alpha$}}  &&&&& &&&&& \sigma^{\tensor}&&&&&\\
			\text{\scalebox{0.7}{AB}} &&&&& &&&&\id& &&&&&\\
			\text{\scalebox{0.7}{A}} &&&&& &&&&& &&\id&&&\\
			\text{\scalebox{0.7}{B}}  &&&&& &&&&& &&&&\id&\\
			\text{\scalebox{0.7}{$\alpha$}}  &&&&& &&&&& &&&\id&&\\
			\text{\scalebox{0.7}{$\beta$}}  &&&&& &&&&& &&&&&\id\\
		\end{block}
	\end{blockarray}
	~\circ~
	\begin{blockarray}{cccc}
		& \text{\scalebox{0.7}{AB}} & \text{\scalebox{0.7}{A}} & \text{\scalebox{0.7}{B}} \\
		\begin{block}{c(ccc)}
			\text{\scalebox{0.7}{((A$\alpha$)B)$\beta$}}  &&&\\
			\text{\scalebox{0.7}{(AB)$\beta$}}  &&&\\
			\text{\scalebox{0.7}{($\alpha$B)$\beta$}}  &&&\\
			\text{\scalebox{0.7}{(A$\alpha$)$\beta$}}  &&&\\
			\text{\scalebox{0.7}{A$\beta$}}  &\id&&\\
			\text{\scalebox{0.7}{$\alpha\beta$}}  &\id&&\\
			\text{\scalebox{0.7}{B$\beta$}}  &&&\\
			\text{\scalebox{0.7}{(A$\alpha$)B}}  &&&\\
			\text{\scalebox{0.7}{AB}}  &\id&&\\
			\text{\scalebox{0.7}{$\alpha$B}}  &\id&&\\
			\text{\scalebox{0.7}{A$\alpha$}}  &&&\\
			\text{\scalebox{0.7}{A}}  &&\id&\\
			\text{\scalebox{0.7}{$\alpha$}}  &&\id&\\
			\text{\scalebox{0.7}{B}}  &&&\id\\
			\text{\scalebox{0.7}{$\beta$}}  &&&\id\\
		\end{block}
	\end{blockarray}
	$\\
	where $f = m^{\tensor} \circ (\id \tensor \sigma^{\tensor} \tensor \id) \circ m^{\tensor}$, $g = \sigma^{\tensor} \tensor \id$ and $h = m^{\tensor} \circ (\id \tensor \sigma^{\tensor}) \circ m^{\tensor}$.

	\caption{Semantics of the right-hand-side of \eqmix.}
	\label{fig:proof_soundness_mix}
\end{figure*}

The equations proven in this subsection rely on the associativity of $|$ in \bfup{H}, which itself relies on the distributivity of $\tensor$ over $\oplus$. Leveraging those properties directly is not easy, hence we instead rely on the matrix representation of the semantics, given in \Cref{fig:mat_sem_fun}. When using the matrix representation, we usually annotate the rows and columns for readability, and while we will continue to do so, we will simplify those annotations by omitting the $\interp{-}$ and the $\tensor$, hence ``AB'' stands for ``$\interp{A} \tensor \interp{B}$'' and ``A(BC)'' stands for ``$\interp{A} \tensor (\interp{B} \tensor \interp{C})$''.
\[\tikzfig{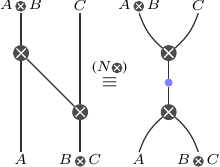}\]
On the left-hand-side, we have

\[
\begin{blockarray}{cccccccc}
	& \text{\scalebox{0.7}{(AB)C}} & \text{\scalebox{0.7}{AC}} & \text{\scalebox{0.7}{BC}} & \text{\scalebox{0.7}{AB}} & \text{\scalebox{0.7}{A}} & \text{\scalebox{0.7}{B}} & \text{\scalebox{0.7}{C}} \\
	\begin{block}{c(ccccccc)}
		\text{\scalebox{0.7}{A(BC)}}  & \alpha^{\tensor} &&&&&&\\
		\text{\scalebox{0.7}{A}}  &&&&&\id &&\\
		\text{\scalebox{0.7}{BC}}  &&&\id &&&&\\
	\end{block}
\end{blockarray}
\circ
\begin{blockarray}{cccc}
	& \text{\scalebox{0.7}{(AB)C}} & \text{\scalebox{0.7}{AB}} & \text{\scalebox{0.7}{C}} \\
	\begin{block}{c(ccc)}
		\text{\scalebox{0.7}{(AB)C}}  & \id &&\\
		\text{\scalebox{0.7}{AC}}  &&&\\
		\text{\scalebox{0.7}{BC}}  &&&\\
		\text{\scalebox{0.7}{AB}}  && \id&\\
		\text{\scalebox{0.7}{A}}  &&&\\
		\text{\scalebox{0.7}{B}}  &&&\\
		\text{\scalebox{0.7}{C}}  &&& \id\\
	\end{block}
\end{blockarray}
\]
On the right-hand-side we have

\[
\begin{blockarray}{cc}
	& \text{\scalebox{0.7}{A(BC)}} \\
	\begin{block}{c(c)}
		\text{\scalebox{0.7}{A(BC)}}  & \id \\
		\text{\scalebox{0.7}{A}}  &\\
		\text{\scalebox{0.7}{BC}}  &\\
	\end{block}
\end{blockarray}
\circ
\begin{blockarray}{cc}
	& \text{\scalebox{0.7}{(AB)C}}  \\
	\begin{block}{c(c)}
		\text{\scalebox{0.7}{A(BC)}}  & \alpha^{\tensor} \\
	\end{block}
\end{blockarray}
\circ
\begin{blockarray}{cccc}
	& \text{\scalebox{0.7}{(AB)C}} & \text{\scalebox{0.7}{AB}} & \text{\scalebox{0.7}{C}} \\
	\begin{block}{c(ccc)}
		\text{\scalebox{0.7}{(AB)C}}  & \id &&\\
	\end{block}
\end{blockarray}
\]
Both composition yield the same following matrix:
\[
\begin{blockarray}{cccc}
	& \text{\scalebox{0.7}{(AB)C}} & \text{\scalebox{0.7}{AB}} & \text{\scalebox{0.7}{C}} \\
	\begin{block}{c(ccc)}
		\text{\scalebox{0.7}{A(BC)}}  & \alpha^{\tensor} &&\\
		\text{\scalebox{0.7}{A}}  &&\\
		\text{\scalebox{0.7}{BC}}  &&\\
	\end{block}
\end{blockarray}  \]
Hence the equation \eqTT{} is sound.

\[\tikzfig{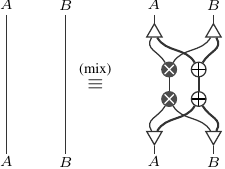}\]

The matrices here are particularly large: up to 15 rows and 15 columns.
For clarity, the thick wires of the above diagram will be represented by Greek letters ($\alpha$ for A and $\beta$ for B) in the row and column annotations. On the right-hand-side we have  the composition described in \Cref{fig:proof_soundness_mix}, where the matrices correspond respectively to the bottom two contractions, the swap, the bottom tensor and plus, the top tensor and plus, and the top two contractions. When computing the composition, we obtain the identity matrix, hence \eqmix{} is sound.

\[\tikzfig{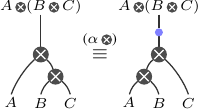}\]

On the left-hand-side, we have

\[
\begin{blockarray}{cccc}
	& \text{\scalebox{0.7}{A(BC)}} & \text{\scalebox{0.7}{A}} & \text{\scalebox{0.7}{BC}} \\
	\begin{block}{c(ccc)}
		\text{\scalebox{0.7}{(AB)C}}  & \alpha^{\tensor-1}&&\\
		\text{\scalebox{0.7}{AC}}  &&&\\
		\text{\scalebox{0.7}{BC}}  &&&\id\\
		\text{\scalebox{0.7}{AB}}  &&&\\
		\text{\scalebox{0.7}{A}}  &&\id&\\
		\text{\scalebox{0.7}{B}}  &&&\\
		\text{\scalebox{0.7}{C}}  &&&\\
	\end{block}
\end{blockarray}
\circ
\begin{blockarray}{cc}
	& \text{\scalebox{0.7}{A(BC)}} \\
	\begin{block}{c(c)}
		\text{\scalebox{0.7}{A(BC)}}  & \id\\
		\text{\scalebox{0.7}{A}}  &\\
		\text{\scalebox{0.7}{BC}}  &\\
	\end{block}
\end{blockarray}
\]
On the right-hand-side we have

\[
\begin{blockarray}{cccc}
	& \text{\scalebox{0.7}{(AB)C}} & \text{\scalebox{0.7}{AB}} & \text{\scalebox{0.7}{C}} \\
	\begin{block}{c(ccc)}
		\text{\scalebox{0.7}{(AB)C}}  & \id&&\\
		\text{\scalebox{0.7}{AC}}  &&&\\
		\text{\scalebox{0.7}{BC}}  &&&\\
		\text{\scalebox{0.7}{AB}}  &&\id&\\
		\text{\scalebox{0.7}{A}}  &&&\\
		\text{\scalebox{0.7}{B}}  &&&\\
		\text{\scalebox{0.7}{C}}  &&&\id\\
	\end{block}
\end{blockarray}
\circ
\begin{blockarray}{cc}
	& \text{\scalebox{0.7}{(AB)C}} \\
	\begin{block}{c(c)}
		\text{\scalebox{0.7}{(AB)C}}  & \id\\
		\text{\scalebox{0.7}{AB}}  &\\
		\text{\scalebox{0.7}{C}}  &\\
	\end{block}
\end{blockarray}
\circ
\begin{blockarray}{cc}
	& \text{\scalebox{0.7}{A(BC)}}  \\
	\begin{block}{c(c)}
		\text{\scalebox{0.7}{(AB)C}}  & \alpha^{\tensor-1} \\
	\end{block}
\end{blockarray}
\]
Both composition yield the same following matrix:
\[
\begin{blockarray}{cc}
	& \text{\scalebox{0.7}{A(BC)}} \\
	\begin{block}{c(c)}
		\text{\scalebox{0.7}{A(BC)}}  & \alpha^{\tensor-1}\\
		\text{\scalebox{0.7}{AC}} &\\
		\text{\scalebox{0.7}{BC}}   &\\
		\text{\scalebox{0.7}{AB}}  &\\
		\text{\scalebox{0.7}{A}}  &\\
		\text{\scalebox{0.7}{B}}  &\\
		\text{\scalebox{0.7}{C}}  &\\
	\end{block}
\end{blockarray}  \]
Hence the equation \eqAlphaT{} is sound.

\[\tikzfig{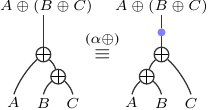}\]

On the left-hand-side, we have

\[
\begin{blockarray}{cccccc}
	& \text{\scalebox{0.7}{AB}} &\text{\scalebox{0.7}{AC}} & \text{\scalebox{0.7}{A}} & \text{\scalebox{0.7}{B}}  & \text{\scalebox{0.7}{C}}  \\
	\begin{block}{c(ccccc)}
		\text{\scalebox{0.7}{(AB)C}}  &&&&&\\
		\text{\scalebox{0.7}{AC}}  &&\id&&&\\
		\text{\scalebox{0.7}{BC}}  &&&&&\\
		\text{\scalebox{0.7}{AB}}  &\id&&&&\\
		\text{\scalebox{0.7}{A}}  &&&\id&&\\
		\text{\scalebox{0.7}{B}}  &&&&\id&\\
		\text{\scalebox{0.7}{C}}  &&&&&\id\\
	\end{block}
\end{blockarray}
\circ
\begin{blockarray}{cccc}
	& \text{\scalebox{0.7}{A}}& \text{\scalebox{0.7}{B}} & \text{\scalebox{0.7}{C}}  \\
	\begin{block}{c(ccc)}
		\text{\scalebox{0.7}{AB}}  &&&\\
		\text{\scalebox{0.7}{AC}}  &&&\\
		\text{\scalebox{0.7}{A}}  &\id&&\\
		\text{\scalebox{0.7}{B}}  &&\id&\\
		\text{\scalebox{0.7}{C}}  &&&\id\\
	\end{block}
\end{blockarray}
\]
On the right-hand-side we have

\[
\begin{blockarray}{cccccc}
	& \text{\scalebox{0.7}{AC}} &\text{\scalebox{0.7}{BC}} & \text{\scalebox{0.7}{A}} & \text{\scalebox{0.7}{B}}  & \text{\scalebox{0.7}{C}}  \\
	\begin{block}{c(ccccc)}
		\text{\scalebox{0.7}{(AB)C}}  &&&&&\\
		\text{\scalebox{0.7}{AC}}  &\id&&&&\\
		\text{\scalebox{0.7}{BC}}  &&\id&&&\\
		\text{\scalebox{0.7}{AB}}  &&&&&\\
		\text{\scalebox{0.7}{A}}  &&&\id&&\\
		\text{\scalebox{0.7}{B}}  &&&&\id&\\
		\text{\scalebox{0.7}{C}}  &&&&&\id\\
	\end{block}
\end{blockarray}
\circ
\begin{blockarray}{cccc}
	& \text{\scalebox{0.7}{A}}& \text{\scalebox{0.7}{B}} & \text{\scalebox{0.7}{C}}  \\
	\begin{block}{c(ccc)}
		\text{\scalebox{0.7}{AC}}  &&&\\
		\text{\scalebox{0.7}{BC}}  &&&\\
		\text{\scalebox{0.7}{A}}  &\id&&\\
		\text{\scalebox{0.7}{B}}  &&\id&\\
		\text{\scalebox{0.7}{C}}  &&&\id\\
	\end{block}
\end{blockarray}
\circ
\begin{blockarray}{cccc}
& \text{\scalebox{0.7}{A}} & \text{\scalebox{0.7}{B}} & \text{\scalebox{0.7}{C}}  \\
\begin{block}{c(ccc)}
	\text{\scalebox{0.7}{A}}  &\id&&\\
	\text{\scalebox{0.7}{B}}  &&\id&\\
	\text{\scalebox{0.7}{C}}  &&&\id\\
\end{block}
\end{blockarray}
\]
Both composition yield the same following matrix:
\[
\begin{blockarray}{cccc}
	& \text{\scalebox{0.7}{A}} & \text{\scalebox{0.7}{B}} & \text{\scalebox{0.7}{C}}  \\
\begin{block}{c(ccc)}
	\text{\scalebox{0.7}{A(BC)}}  &&&\\
	\text{\scalebox{0.7}{AC}} &&&\\
	\text{\scalebox{0.7}{BC}}  &&&\\
	\text{\scalebox{0.7}{AB}}  &&&\\
	\text{\scalebox{0.7}{A}}  &\id&&\\
	\text{\scalebox{0.7}{B}}  &&\id&\\
	\text{\scalebox{0.7}{C}}  &&&\id\\
\end{block}
\end{blockarray}  \]
Hence the equation \eqAlphaP{} is sound.

\subsection{Finishing the Functional Fragment}

We can now prove the soundness of \Cref{lem:spider}, that is:
\begin{lemma}\label{lem:spider_sound}
	Given two diagrams $d,e : \wireSetA \to \wireSetB$  satisfying the following:
	\begin{itemize}
		\item They are connected planar graphs, and in particular do not include any Swap.
		\item It is composed only of Identities $\id$, Tensors $\tensor$, upside-down Tensors, and Adapters $\maclane$.
	\end{itemize}
	Then each can be rewritten into the other using only \eqAlphaT, \eqT, \eqTT, \eqTM, \eqM and \eqMM. It follows that this rewriting is sound.
\end{lemma}
\begin{proof}
	See \Cref{lem:spider} for the existence of the rewriting. All the listed equations have been proven to be sound in the previous subsections, so  this rewriting is sound.
\end{proof}
We will need two lemmas, which allows us to decompose an equation difficult to prove sound into three (or seven) equations easier to prove sound. They both derive from the fact that in a semiadditive category, if $f \circ \iota_\ell = g \circ \iota_\ell$ and $f \circ \iota_r = g \circ \iota_r$ then $f = g$.
\begin{lemma}\label{lem:decomp_para_bin}
	Given two diagrams $d,e : A \parallel B \to \wireSetA$, then $\interp{d} = \interp{e}$ if and only if each of the following hold:
	\[\tikzfig{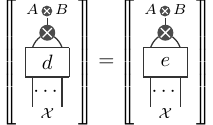}\]
	\[\tikzfig{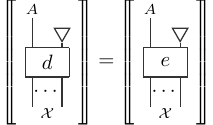}\]
	\[\tikzfig{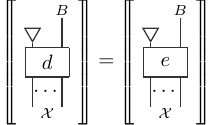}\]
\end{lemma}
\begin{proof}
	We start by noting that $\interp{A \parallel B} = \interp{A} | \interp{B} = (\interp{A} \tensor \interp{B}) \oplus (\interp{A} \oplus \interp{B})$. Since we are in a semiadditive category, $\interp{d} = \interp{e}$ if and only if:
	\begin{itemize}
		\item $\interp{d} \circ \iota_\ell = \interp{e} \circ \iota_\ell$
		\item $\interp{d} \circ \iota_r \circ \iota_\ell = \interp{e} \circ \iota_r \circ \iota_\ell$
		\item $\interp{d} \circ \iota_r \circ \iota_r = \interp{e} \circ \iota_r \circ \iota_r$
\end{itemize}
Each of those three items directly correspond to one of the equations, in order.
\end{proof}

\begin{lemma}\label{lem:decomp_para_ter}
	Given two diagrams $d,e : A \parallel B \parallel C \to \wireSetA$, then $\interp{d} = \interp{e}$ if and only if each of the following hold:
	\[\tikzfig{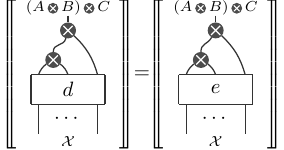}\]
	\[\tikzfig{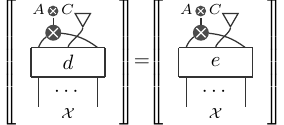}\]
	\[\tikzfig{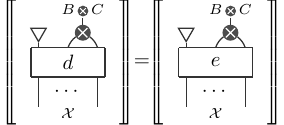}\]
	\[\tikzfig{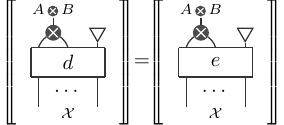}\]
	\[\tikzfig{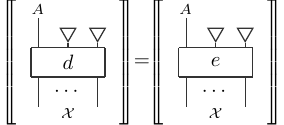}\]
	\[\tikzfig{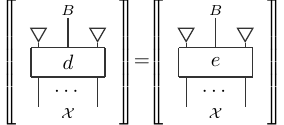}\]
	\[\tikzfig{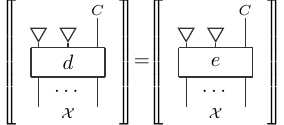}\]
\end{lemma}
\begin{proof}
	We start by noting that $\interp{A \parallel B \parallel C} = (\interp{A} | \interp{B}) | \interp{C}$ which can be unfolded into
	\[\begin{matrix} (((\interp{A} \tensor \interp{B}) \oplus (\interp{A} \oplus \interp{B})) \tensor \interp{C}) \\ \oplus (((\interp{A} \tensor \interp{B}) \oplus (\interp{A} \oplus \interp{B})) \oplus \interp{C}) \end{matrix} \]
	Since we are in a semiadditive category and additionally $\tensor$ is distributive over $\oplus$, $\interp{d} = \interp{e}$ if and only if:
	\begin{itemize}
		\item $\interp{d}
		\circ \iota_\ell
		\circ (\iota_\ell \tensor \id)
		=\interp{e}
		\circ \iota_\ell
		\circ (\iota_\ell \tensor \id)
		$ 
		\item $\interp{d}
		\circ \iota_\ell
		\circ (\iota_r \tensor \id)
		\circ (\iota_\ell \tensor \id)
		= \interp{e}
		\circ \iota_\ell
		\circ (\iota_r \tensor \id)
		\circ (\iota_\ell \tensor \id)
		$ 
		\item $\interp{d}
		\circ \iota_\ell
		\circ (\iota_r \tensor \id)
		\circ (\iota_r \tensor \id)
		= \interp{e}
		\circ \iota_\ell
		\circ (\iota_r \tensor \id)
		\circ (\iota_r \tensor \id)
		$
		\item $\interp{d}
		\circ  \iota_r
		\circ \iota_\ell
		\circ \iota_\ell
		=\interp{e}
		\circ  \iota_r
		\circ \iota_\ell
		\circ \iota_\ell
		$ 
		\item $\interp{d}
		\circ  \iota_r
		\circ \iota_\ell
		\circ \iota_r
		\circ \iota_\ell
		= \interp{e}
		\circ \iota_r
		\circ \iota_\ell
		\circ  \iota_r
		\circ \iota_\ell
		$ 
		\item $\interp{d}
		\circ \iota_r
		\circ \iota_\ell
		\circ  \iota_r
		\circ \iota_r
		= \interp{e}
		\circ \iota_r
		\circ \iota_\ell
		\circ \iota_r
		\circ  \iota_r
		$ 
		\item $\interp{d}
		\circ \iota_r
		\circ \iota_r
		= \interp{e}
		\circ \iota_r
		\circ \iota_r
		$ 
	\end{itemize}
	Each of those seven items directly correspond to one of the equations, in order.
\end{proof}

Now, we can tackle the remaining equations. For that, we note that we can freely rewrite diagrams using equations that we already proved to be sound.
\[\tikzfig{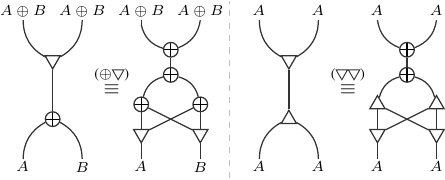}\]
To prove the soundness of those, we use \Cref{lem:decomp_para_bin}, and then rely on the already proven soundness of \eqbot, \eqbotC, \eqNP, \eqNC, \eqLambdaC, \eqRhoC{} and \eqPPCleft.

\[\tikzfig{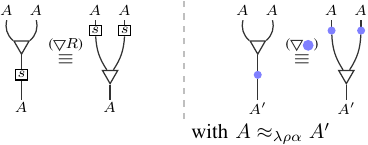}\]
To prove the soundness of those, we use \Cref{lem:decomp_para_bin}, and then rely on the already proven soundness of \eqbotC, \eqLambdaC, \eqRhoC, \eqNS, \eqNM, \eqTS{} and \eqTM.

\[\tikzfig{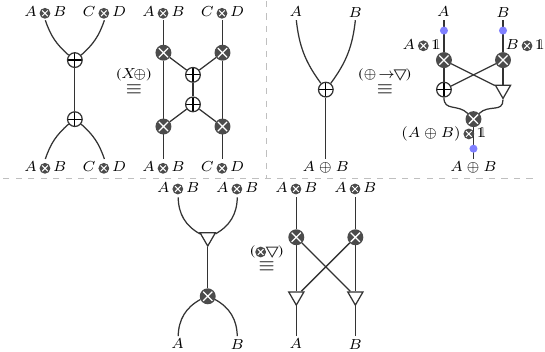}\]
To prove the soundness of those, we use \Cref{lem:decomp_para_bin}, and then rely \Cref{lem:spider_sound} and on the already proven soundness of  \eqbot, \eqbotC, \eqNT, \eqNP, \eqNC, \eqNM, \eqLambdaC, \eqRhoC, \eqTNleft{} and \eqPPCleft.

\[\tikzfig{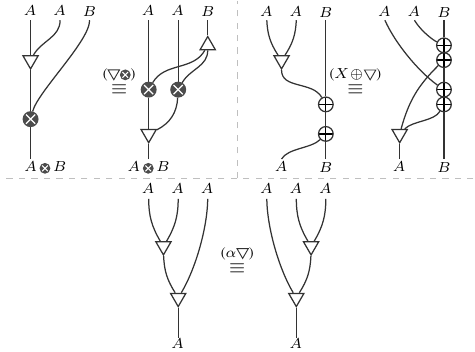}\]
To prove the soudness of those, we use \Cref{lem:decomp_para_ter}, and then rely \Cref{lem:spider_sound} and on the already proven soundness of  \eqbot, \eqbotC, \eqNT, \eqNP, \eqNC, \eqNM, \eqLambdaC, \eqRhoC, \eqTNleft{} and \eqPPCleft.

\subsection{The Full Language}

\begin{figure*}$
	\begin{blockarray}{c cccc}
		& \text{\scalebox{0.7}{$\interp{\tone} \tensor \interp{\tone}$}}   & \text{\scalebox{0.7}{$\interp{\tone}$}} & \text{\scalebox{0.7}{$\interp{\tone}$}}  & \text{\scalebox{0.7}{$\oone$}} \\
		\begin{block}{c(cccc)}
			\text{\scalebox{0.7}{$\interp{\tone}$}}  &\rho^{\tensor}&\id&&\\
			\text{\scalebox{0.7}{$\oone$}}  &&&\id&\id\\
		\end{block}
	\end{blockarray}~\circ~
	\begin{blockarray}{c cccc cccc}
		&\text{\scalebox{0.7}{$\interp{\tone} \tensor \interp{\tone} \tensor \interp{\tone}$}}
		&\text{\scalebox{0.7}{$\interp{\tone} \tensor \interp{\tone}$}}
		&\text{\scalebox{0.7}{$\interp{\tone} \tensor \interp{\tone}$}}
		&\text{\scalebox{0.7}{$\interp{\tone} \tensor \interp{\tone}$}}
		&\text{\scalebox{0.7}{$\interp{\tone}$}}
		&\text{\scalebox{0.7}{$\interp{\tone}$}}
		&\text{\scalebox{0.7}{$\interp{\tone}$}}
		&\text{\scalebox{0.7}{$\oone$}}\\
		\begin{block}{c(cccccccc)}
			\text{\scalebox{0.7}{$\interp{\tone} \tensor \interp{\tone}$}} &&\id&&\id&&&&\\
			\text{\scalebox{0.7}{$\interp{\tone}$}}&&&& &\id&&&\\
			\text{\scalebox{0.7}{$\interp{\tone}$}}&&&& &&\id&\id&\\
			\text{\scalebox{0.7}{$\oone$}} &&&& &&&&\id \\
		\end{block}
	\end{blockarray}$\\$
	\circ~\begin{blockarray}{c cccc cccc}
		&\text{\scalebox{0.7}{$\interp{\tone} \tensor \interp{\tone} \tensor \interp{\tone}$}}
		&\text{\scalebox{0.7}{$\interp{\tone} \tensor \interp{\tone}$}}
		&\text{\scalebox{0.7}{$\interp{\tone} \tensor \interp{\tone}$}}
		&\text{\scalebox{0.7}{$\interp{\tone} \tensor \interp{\tone}$}}
		&\text{\scalebox{0.7}{$\interp{\tone}$}}
		&\text{\scalebox{0.7}{$\interp{\tone}$}}
		&\text{\scalebox{0.7}{$\interp{\tone}$}}
		&\text{\scalebox{0.7}{$\oone$}}\\
		\begin{block}{c(cccccccc)}
			\text{\scalebox{0.7}{$\interp{\tone} \tensor \interp{\tone} \tensor \interp{\tone}$}}  &s \cdot \sigma^{\tensor} \tensor \id &&& &&&&\\
			\text{\scalebox{0.7}{$\interp{\tone} \tensor \interp{\tone}$}}  &&&s \cdot \id& &&&&\\
			\text{\scalebox{0.7}{$\interp{\tone} \tensor \interp{\tone}$}}  &&s \cdot \id&& &&&&\\
			\text{\scalebox{0.7}{$\interp{\tone} \tensor \interp{\tone}$}}  &&&&\sigma^{\tensor} &&&&\\
			\text{\scalebox{0.7}{$\interp{\tone}$}}  &&&& &&\id&&\\
			\text{\scalebox{0.7}{$\interp{\tone}$}}  &&&& &\id&&&\\
			\text{\scalebox{0.7}{$\interp{\tone}$}}  &&&& &&&s \cdot \id&\\
			\text{\scalebox{0.7}{$\oone$}}  &&&& &&&&\id\\
		\end{block}
	\end{blockarray}$\\$
	\circ~
	\begin{blockarray}{c cccc}
		& \text{\scalebox{0.7}{$\interp{\tone} \tensor \interp{\tone}$}}   & \text{\scalebox{0.7}{$\interp{\tone}$}} & \text{\scalebox{0.7}{$\interp{\tone}$}}  & \text{\scalebox{0.7}{$\oone$}} \\
		\begin{block}{c(cccc)}
			\text{\scalebox{0.7}{$\interp{\tone} \tensor \interp{\tone} \tensor \interp{\tone}$}}  &&&&\\
			\text{\scalebox{0.7}{$\interp{\tone} \tensor \interp{\tone}$}}  &\id&&&\\
			\text{\scalebox{0.7}{$\interp{\tone} \tensor \interp{\tone}$}}  &&&&\\
			\text{\scalebox{0.7}{$\interp{\tone} \tensor \interp{\tone}$}}  &\id&&&\\
			\text{\scalebox{0.7}{$\interp{\tone}$}}  &&\id&&\\
			\text{\scalebox{0.7}{$\interp{\tone}$}}  &&&\id&\\
			\text{\scalebox{0.7}{$\interp{\tone}$}}  &&&\id&\\
			\text{\scalebox{0.7}{$\oone$}}  &&&&\id\\
		\end{block}
	\end{blockarray}~
	\circ~
	\begin{blockarray}{ccc}
		& \text{\scalebox{0.7}{$\interp{\tone}$}} & \text{\scalebox{0.7}{$\oone$}}  \\
		\begin{block}{c(cc)}
			\text{\scalebox{0.7}{$\interp{\tone} \tensor \interp{\tone}$}}  &\rho^{\tensor-1} & \\
			\text{\scalebox{0.7}{$\interp{\tone}$}}  &\id& \\
			\text{\scalebox{0.7}{$\interp{\tone}$}}  && \id \\
			\text{\scalebox{0.7}{$\oone$}}  && \id \\
		\end{block}
	\end{blockarray} $
	\caption{Semantics of the left-hand-side of \eqUS.}
	\label{fig:proof_soundness_can}
\end{figure*}

In this section, was will use rely on the matrix representation given in \Cref{fig:mat_sem}.
\[\tikzfig{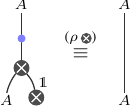}\]
On the left-hand-side we have\\
$
\begin{blockarray}{c cccc}
	& \text{\scalebox{0.7}{$\interp{A} \tensor \interp{\tone}$}} & \text{\scalebox{0.7}{$\interp{A}$}}  & \text{\scalebox{0.7}{$\interp{\tone}$}} & \text{\scalebox{0.7}{$\oone$}} \\
	\begin{block}{c(cccc)}
		\text{\scalebox{0.7}{$\interp{A}$}}  &\rho^{\tensor}&\id&&\\
		\text{\scalebox{0.7}{$\oone$}}  & &&\id & \id\\
	\end{block}
\end{blockarray}$\\$
\circ~
\begin{blockarray}{ccc}
	& \text{\scalebox{0.7}{$\interp{A \tensor \oone}$}} & \text{\scalebox{0.7}{$\oone$}}\\
	\begin{block}{c(cc)}
		\text{\scalebox{0.7}{$\interp{A} \tensor \interp{\tone}$}}  & \id &\\
		\text{\scalebox{0.7}{$\interp{A}$}}  &&\\
		\text{\scalebox{0.7}{$\interp{\tone}$}}  &&\\
		\text{\scalebox{0.7}{$\oone$}}  && \id\\
	\end{block}
\end{blockarray}
\circ
\begin{blockarray}{ccc}
& \text{\scalebox{0.7}{$\interp{A}$}} & \text{\scalebox{0.7}{$\oone$}}  \\
\begin{block}{c(cc)}
	\text{\scalebox{0.7}{$\interp{A \tensor \oone}$}}  & \rho^{\tensor-1}& \\
	\text{\scalebox{0.7}{$\oone$}}  && \id \\
\end{block}
\end{blockarray} $\\
which simplifies to the identity matrix. Hence \eqRhoT{} is sound.
\[\tikzfig{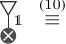}\]
On the left-hand-side we have
\[
\begin{blockarray}{c cc}
	  & \text{\scalebox{0.7}{$\interp{\tone}$}} & \text{\scalebox{0.7}{$\oone$}} \\
	\begin{block}{c(cc)}
		\text{\scalebox{0.7}{$\oone$}}  & \id & \id\\
	\end{block}
\end{blockarray}
\circ
\begin{blockarray}{cc}
	 & \text{\scalebox{0.7}{$\oone$}}\\
	\begin{block}{c(c)}
		\text{\scalebox{0.7}{$\interp{\tone}$}}  &\\
		\text{\scalebox{0.7}{$\oone$}}  & \id \\
	\end{block}
\end{blockarray} \]
which simplifies to the identity 1-by-1 matrix, hence \eqUN{} is sound. Last but not least, we consider:
\[\tikzfig{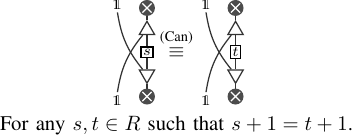}\]
On the left-hand-side, we have the composition described in \Cref{fig:proof_soundness_can}, where the matrices correspond respectively to the bottom unit, the bottom contraction, the swap with a scalar, the top contraction, and the top unit. When computing the composition, we obtain the following (we note that on $\oone \tensor \one \to \oone \tensor \oone$, we actually have $\sigma^{\tensor} = \id$):
\[
\begin{blockarray}{ccc}
	& \text{\scalebox{0.7}{$\interp{\tone}$}} & \text{\scalebox{0.7}{$\oone$}}  \\
	\begin{block}{c(cc)}
		\text{\scalebox{0.7}{$\interp{\tone}$}}  & \id & \id \\
		\text{\scalebox{0.7}{$\oone$}}  &\id &(s+1) \cdot \id\\
	\end{block}
\end{blockarray}\]
Similarly, if we compute the matrix for the right-hand-side, we obtain the following:
\[
\begin{blockarray}{ccc}
	& \text{\scalebox{0.7}{$\interp{\tone}$}} & \text{\scalebox{0.7}{$\oone$}}  \\
	\begin{block}{c(cc)}
		\text{\scalebox{0.7}{$\interp{\tone}$}}  & \id & \id \\
		\text{\scalebox{0.7}{$\oone$}}  &\id &(t+1) \cdot \id\\
	\end{block}
\end{blockarray}\]
Both are equal whenever $s+1 = t+1$, hence \eqUS{} is sound.
\begin{figure*}\centering
	\tikzfig{induced/induced_eq}
	\caption{Additional Induced Equations for the \Langage}
	\label{fig:eq_induced}
\end{figure*}
\clearpage
\begin{figure*}\centering
	\tikzfig{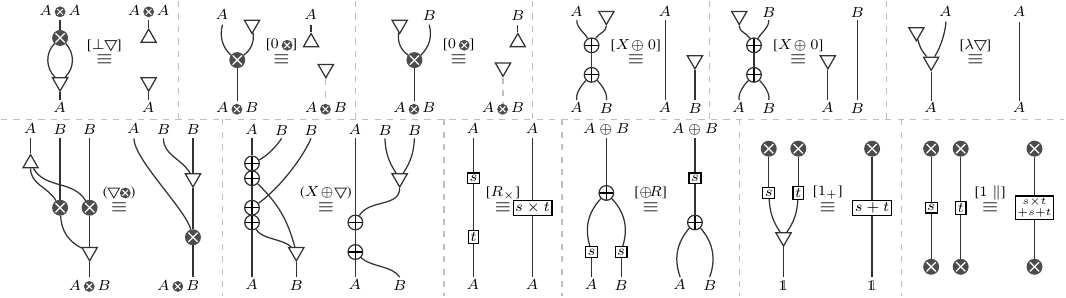}
	\caption{Reminder of the Induced Equations for the \Langage}
	\label{fig:eq_induced_core}
\end{figure*}
\section{Notable Lemmas for the Equational Theory}
\label{app:lemma}

In this appendix, we work on a collection of useful lemmas and equations. In order to distinguish them from the equations of \Cref{fig:eq_tensor_plus,fig:eq_contraction_nat,fig:eq_null_nat,fig:eq_maclane,fig:eq_scalar,fig:eq_unit}, we name the new equations with squared brackets $[-]$ instead of parentheses $(-)$.

\subsection{Induced Equations}

	From the equations of \Cref{fig:eq_tensor_plus,fig:eq_null_nat,fig:eq_contraction_nat,fig:eq_maclane,fig:eq_scalar}, we can obtain a few additional equation, which we list in \Cref{fig:eq_induced_core,fig:eq_induced}, and which can be proved as follows. Except for \eqCCtoC, \eqUsum{} and \eqUpar{}, all the equations are proved in the Functional fragment of the language.\\
	
	For \eqLambdaC{}, we simply use \eqSigmaC{} then \eqRhoC{}.
	
	For \eqTTsimple{} we start from the right-hand-side:
		\begin{center}
			\tikzfig{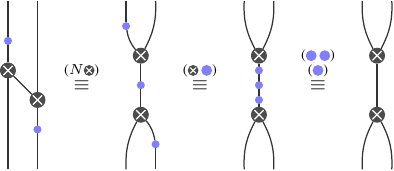}
		\end{center}
	For \eqSigmaT{}:
		\begin{center}
			\tikzfig{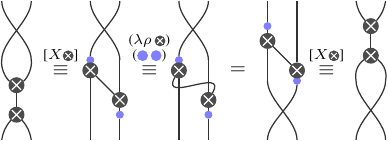}
		\end{center}
	The equation \eqTCright{} here is the mirrored left-right of the equation \eqTCleft{} in \Cref{fig:eq_contraction_nat}. We keep the same name for convenience. However, deducing one from the other is surprisingly non-trivial:
		\begin{center}
			\tikzfig{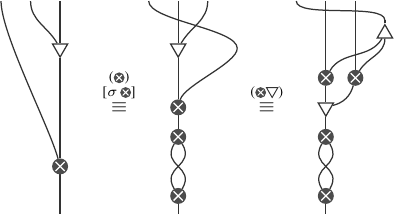}
			\end{center}
			\begin{center}
			\tikzfig{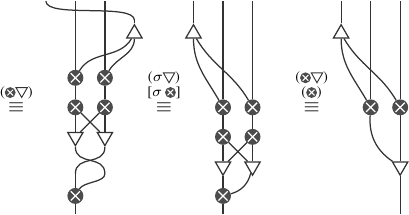}
		\end{center}
	The equations \eqTNleft{} are both proved symmetrically:
	\begin{center}
		\tikzfig{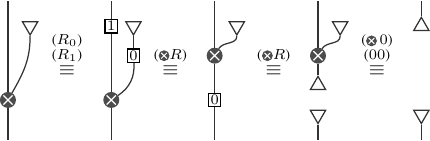}
	\end{center}
	The equations \eqPPNright{} are both proved symmetrically:
\begin{center}
	\tikzfig{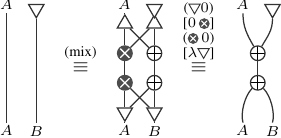}
\end{center}
	For \eqSP{}, we start from the right-hand-side:
\begin{center}
	\tikzfig{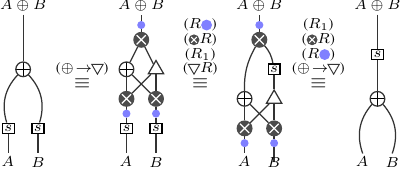}
\end{center}
	For \eqCtoP{}, we start from the right-hand-side:
	\begin{center}
		\tikzfig{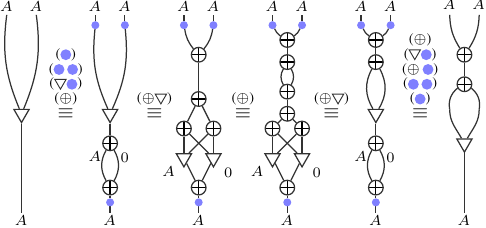}
	\end{center}
For \eqbotC{}:
\begin{center}
	\tikzfig{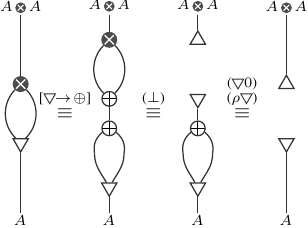}
\end{center}
For \eqCCtoT{}, we start from the right-hand-side:
	\begin{center}
		\tikzfig{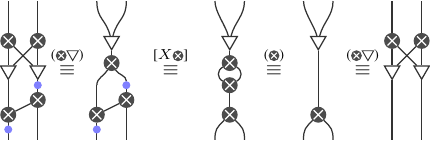}
	\end{center}
For \eqPCtoT{}:
	\begin{center}
		\tikzfig{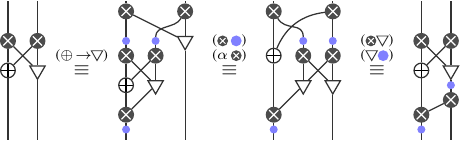}
	\end{center}
For \eqSS, we start from the right-hand-side:
	\begin{center}
		\tikzfig{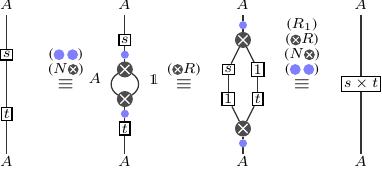}
	\end{center}
The equations \eqPPtoCC{} and \eqPPtoC{} follow from \eqPCtoT{}, but we postpone their proof to \Cref{lem:X_oplus_to_contraction,lem:X_contraction_to_contraction_unit} as it is much more readable using "spider" notations.

Similarly, we postpone the equation \eqCCtoC{} to \Cref{lem:X_contraction_to_contraction_unit}.

For \eqSigmaP{}:
	\begin{center}
		\tikzfig{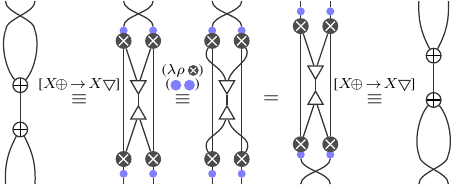}
	\end{center}
The equation \eqPPCright{} here is the mirrored left-right of the equation \eqPPCleft{} in \Cref{fig:eq_contraction_nat}, we keep the same name for convenience. One implies the other through \eqSigmaP{} and \eqSigmaC{}.

Similarly, the equation \eqPPNright{} here is the mirrored left-right of the equation \eqPPNleft{} in \Cref{fig:eq_null_nat}, we keep the same name for convenience, and one implies the other through \eqSigmaP{}.

The equations \eqPPPleft{} are a variation over \eqPPCleft{}. We postpone the their proof to \Cref{lem:X_oplus_oplus} as they are much more readable using  "spider" notations. We postpone similarly the proof of \eqPPSleft{} to \Cref{lem:scalar_through_X_oplus}.

The equation \eqPCreverse{} is a variation over \eqPC{}, and we postpone its proof to \Cref{lem:oplus_contraction_bis} as it is much more readable using "spider" and "disjunction" notations.

The proof of equations \eqUsum{} and \eqUpar{} are also postponed to \Cref{lem:eq_unit_scalar}.

\subsection{The $n$-ary Nodes}

\begin{definition}\label{def:n_ary}
	We define the $n$-ary Tensor for $n \geq 1$ (with $n= 1$ being the identity):
	\begin{center}
		\tikzfig{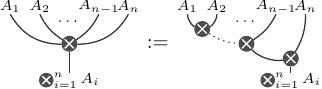}
	\end{center}
Then the $n$-ary Plus for $n \geq 1$ (with $n= 1$ being the identity):
\begin{center}
\tikzfig{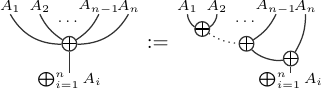}
\end{center}
And lastly the $n$-ary Contraction for $n \geq 1$ (with $n = 1$ being the identity):
\begin{center}
\tikzfig{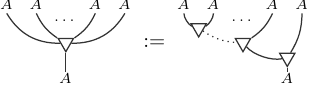}
\end{center}
\end{definition}

Most equations generating $\equiv$ can be generalized naturally to the $n$-ary case by iterating the binary case. We formally prove some of the non-trivial ones.

\begin{lemma}~\label{lem:n_ary_oplus_to_contraction}
	\begin{center}
		\tikzfig{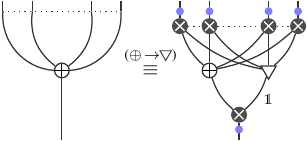}\\ [0.5cm]
		\tikzfig{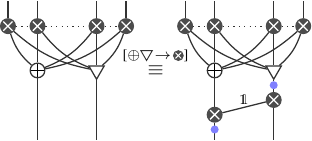}\\ [0.5cm]
		\tikzfig{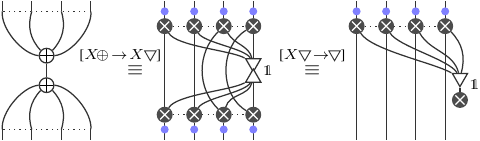}
	\end{center}
\end{lemma}
\begin{proof}
	We proceed inductively, starting with the second one  \eqPCtoT{}:
	\begin{itemize}
		\item $n = 1$ follows from \eqMM{}, \eqT{} and \eqTTsimple{}.
		\item $n = 2$ is the regular \eqPCtoT{}.
		\item For $n > 2$, we simply use the $(n-1)$-ary case, then the binary case, then we undo the case $(n-1)$-ary case.
	\end{itemize}
	We can now prove the first one \eqPtoC{}:
	\begin{itemize}
		\item $n = 1$ follows from \eqMM{} and \eqT{}.
		\item $n = 2$ is the regular \eqPtoC{}.
		\item For $n > 2$, we start by creating a $(n-1)$-Contraction with the $(n-1)$-ary case, then we create the missing binary Contraction using the regular, and lastly we have to eliminate the additional Tensors. For that we can use \eqTT{} to set up a use of \eqPCtoT{}.
	\end{itemize}
	Then, since \eqPPtoCC{} is proved using \eqPCtoT{} and \eqPtoC{} (see \Cref{lem:X_oplus_to_contraction}), we can generalize it to the $n$-ary case without issues. The proof for \eqCCtoC{} given in \Cref{lem:X_contraction_to_contraction_unit} also can also be generalised easily.
\end{proof}

\subsection{The Spider}

\begin{definition}\label{def:spider}
	We define the $(n,m)$-spider for $n \geq 1$, $m \geq 1$ and $\bigtensor_{i=1}^n A_i \isoML \bigtensor_{j=1}^m B_j$:
	\begin{center}
	\tikzfig{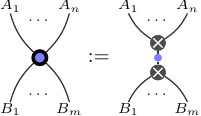}
	\end{center}
\end{definition}
We note that $\tensor$ can be seen as a $(2,1)$-spider by adding a trivial $\maclane$ with the equation \eqM{}, and $\maclane$ is directly a $(1,1)$-spider.

\begin{lemma}\label{lem:spider}
	There exists at most one morphism in $\FCat_\equiv(\wireSetA,\wireSetB)$ such that its diagram is a connected planar graph with for only nodes the $\tensor$ and $\maclane$. In other words, we have the following equation and its up/down mirrored version:
	\begin{center}
		\tikzfig{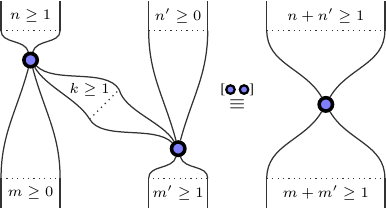}
	\end{center}
\end{lemma}
\begin{proof}~
	\begin{center}
		\tikzfig{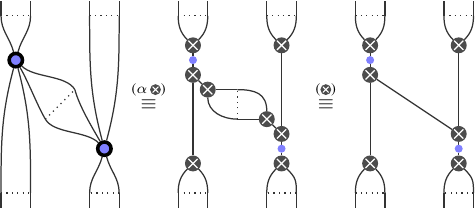}\\ [0.5cm]
		\tikzfig{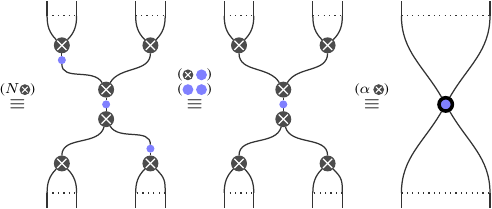}
	\end{center}
\end{proof}

\begin{lemma}\label{lem:spider_equations}
	The equation \eqSpider{} from \Cref{lem:spider} already generalizes the equations \eqTM{}, \eqMM{}, \eqAlphaT{}, \eqT{}, \eqTT{}, and \eqTTsimple{}. We continue this list by generalizing a few additional ones:
	\begin{center}
		\tikzfig{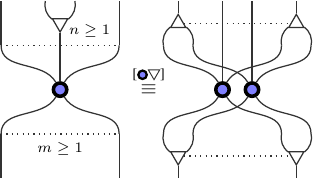} \\ [0.5cm]
		\tikzfig{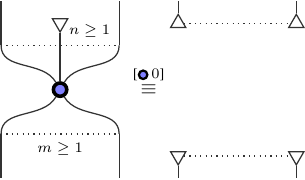}  \\ [0.5cm]
		\tikzfig{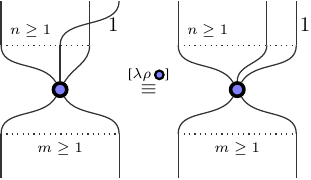}\\ [0.5cm]
		\tikzfig{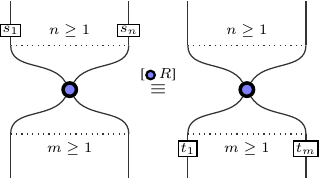}
	\end{center}
	with $s_1 \times \dots \times s_n = t_1 \times \dots \times t_m$.
\end{lemma}
\begin{proof}
	For \eqSpiderC{}, we simply unfold the definition of the spider, and iterate \eqTCleft{}, then use \eqCM{} then iterate \eqTC{}.

	For \eqSpiderN{}, we also unfold the definition of the spider, iterate  \eqTNleft{}, then use \eqMN{} then iterate \eqTN{}.

	For \eqLambdaRhoSpider{}, after unfolding the definition of the tensor, we use \eqAlphaT{} and \eqLambdaRhoT{} to swap the wire of color $\tone$ with the other wires one by one, using \eqTM{} and \eqMM{} to eliminate the $\maclane$ that appeared at each use of \eqLambdaRhoT{}.

	For \eqSpiderS{}, after unfolding the definition of the spider, iterate  \eqTS{}, then use \eqSM{} then iterate \eqTS{}.
\end{proof}

\begin{lemma}\label{lem:contraction_to_X_oplus}
	In order to prove \cref{lem:X_oplus_to_contraction}, we need the following preliminary result:
	\begin{center}
		\tikzfig{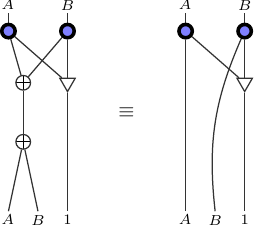}
	\end{center}
\end{lemma}
\begin{proof}~
	\begin{center}
		\tikzfig{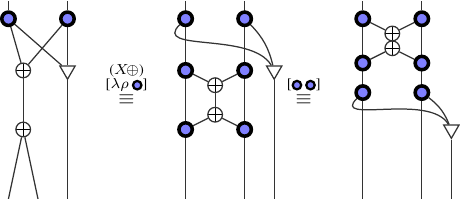} \\ [0.5cm]
		\tikzfig{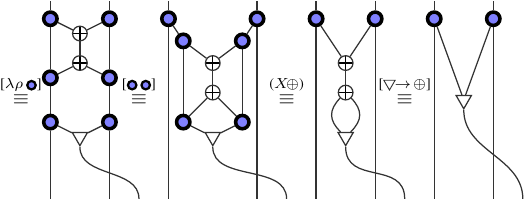}
	\end{center}
	And we conclude by using \eqLambdaRhoSpider.
\end{proof}

We can now prove the first, second and third induced equations that were postponed. Since \eqPPtoC{} is the composition of \eqPPtoCC{} and \eqCCtoC{}, it follows from the next two lemmas that prove both of them.

\begin{lemma}\label{lem:X_oplus_to_contraction}~
	\begin{center}
		\tikzfig{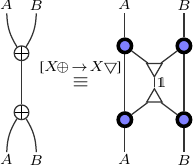}
	\end{center}
\end{lemma}
\begin{proof}~
	\begin{center}
		\tikzfig{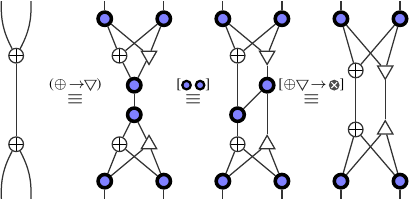}
	\end{center}
	And we conclude by using \Cref{lem:contraction_to_X_oplus} and \eqLambdaRhoSpider.
\end{proof}

\begin{lemma}\label{lem:X_contraction_to_contraction_unit}~
	\begin{center}
		\tikzfig{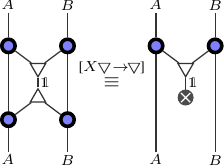}
	\end{center}
\end{lemma}
\begin{proof}~
	\begin{center}
		\tikzfig{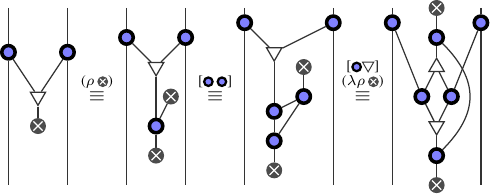}\\ [0.5cm]
		\tikzfig{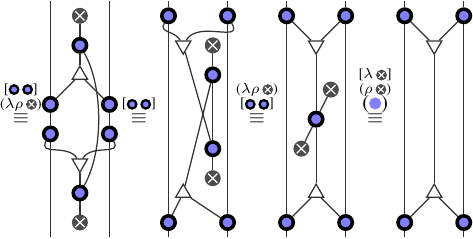}
	\end{center}
\end{proof}

\begin{lemma}\label{lem:scalar_through_X_oplus}~
	\begin{center}
		\tikzfig{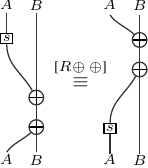}
	\end{center}
\end{lemma}
\begin{proof}~
	\begin{center}
		\tikzfig{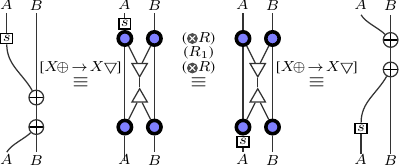}
	\end{center}
\end{proof}

\begin{lemma}\label{lem:X_oplus_oplus}~
	\begin{center}
		\tikzfig{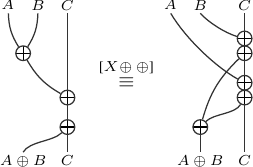}
	\end{center}
	And its mirrored left/right.
\end{lemma}
\begin{proof}~
	\begin{center}
		\tikzfig{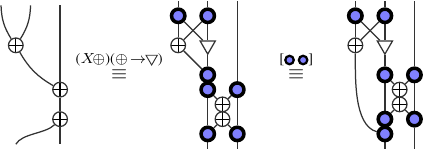} \\ [0.5cm]
		\tikzfig{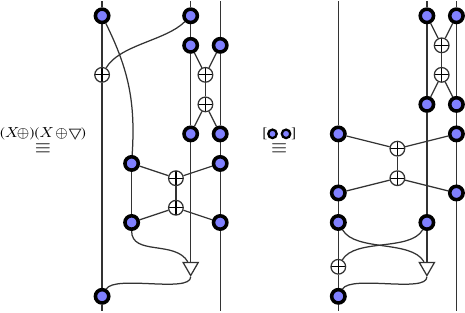} \\ [0.5cm]
		\tikzfig{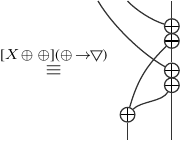}
	\end{center}
	For the mirrored equation, we use the same proof but mirrored.
\end{proof}

\subsection{The Compact Closure}
	While we do not have a Cup and a Cap as generators in order to
	``bend'' wires, we can define the Cap inductively as follows, and the Cup using the up-down mirrored definitions:
	\begin{center}
		\begin{tabular}{cc}
			$\tikzfig{cupcap/tens} $ & $	\tikzfig{cupcap/plus}$ \\&\\
			$\tikzfig{cupcap/one}$ & $\tikzfig{cupcap/zero}$ \\
		\end{tabular}
	\end{center}
	Since they rely on Unit, they are not part of $\FCat$. We will now  the equational theory ensures that they satisfy the snake equations, in other words that $\Cat_{\equiv}$ is a compact closed category. Before that, we need to prove the following lemma:
	
	\begin{lemma}\label{lem:tensor_through_cap}
		\[\tikzfig{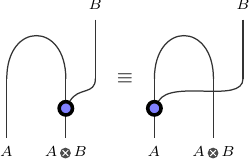}\]
	\end{lemma}
	\begin{proof}
		We proceed inductively over $A$. For $A = \tzero$, it follows from \eqSpiderN. For $A = \tone$, it follows from \eqSpider. For the inductive case $A = B \tensor C$, it follows from \eqSpider{} and the induction hypothesis. For the inductive case $A = B \oplus C$, it follows from \eqPtoC{}, \eqSpider{} and the induction hypothesis.
	\end{proof}
	
	\begin{proposition}\label{prop:compact_close}
		The following equations can be deduced from the equational theory:
		\[\tikzfig{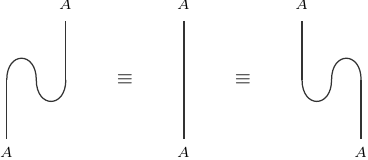}\]
	\end{proposition}
	\begin{proof}
		We proceed inductively over $A$. For $A = \tzero$, it follows from \eqNN. For $A = \tone$, it follows from \eqSpider. For the inductive case $A = B \tensor C$, it follows from \eqSpider{}, \Cref{lem:tensor_through_cap} and the induction hypothesis. For the inductive case $A = B \oplus C$, it follows from \eqPPtoCC, \Cref{lem:tensor_through_cap} and the induction hypothesis.		
	\end{proof}

\subsection{The Disjunctive Collection of Wires}

\begin{definition}\label{def:disjunction}
	We define a disjunction of wires $A_1,\dots,A_n$ for $n \geq 1$:
	\begin{center}
		\tikzfig{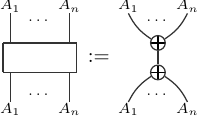}
	\end{center}
\end{definition}
Note the absence of $\maclane$ in this definition.

\begin{lemma}\label{lem:idempotence_disj}
	We have a property of idempotence:
	\begin{center}
		\tikzfig{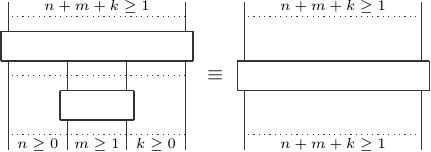}
	\end{center}
\end{lemma}
\begin{proof}
	We simply use \eqAlphaP{}, \eqPM{}, and \eqP{} multiple times each.
\end{proof}

\begin{lemma}\label{lem:disj_reorder}
	As a direct application of the previous lemma, we can reorder two binary disjunctions:
	\begin{center}
		\tikzfig{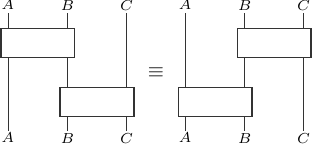}
	\end{center}
\end{lemma}
\begin{proof}
	We use \eqPPtoCC{} to create $\spider$ nodes, then \cref{lem:spider} swap the two and then \eqPPtoCC{} again to make the $\spider$ nodes disappear.
\end{proof}

\begin{lemma}\label{lem:disj_expansion}
	A $n$-ary disjunction of wires is equivalent to the collection of all the binary disjunctions of wires:
	\begin{center}
		\tikzfig{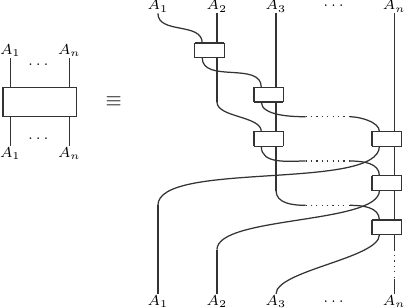}
	\end{center}
	We recall that using \eqSigmaP{}, the left-hand-side and right-hand-side of a binary disjunction can be swapped, and using the \Cref{lem:disj_reorder} the order in which the binary disjunctions are made does not matters.
\end{lemma}
\begin{proof}
	We unfold the definition of a disjunction, and then iterate the use of \eqPPPleft{} to push the "X" at the middle upward, this create the binary disjunctions between $A_n$ and every other wire, and a new "X" at the middle, which allows you to reiterate the same procedure.
\end{proof}

\begin{lemma}\label{lem:contraction_to_disj_nary}
	We can generalize \eqCtoP:
	\begin{center}
		\tikzfig{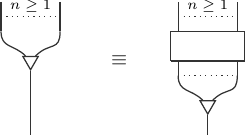}
	\end{center}
\end{lemma}
\begin{proof}
	This is simply using \Cref{lem:disj_expansion}, \eqCtoP{} and \eqAlphaC{}.
\end{proof}

Those last two lemmas give us the tools that lacked to generalize \Cref{lem:contraction_to_X_oplus} to the $n$-ary case as follows:

\begin{lemma}\label{lem:contraction_to_X_oplus_nary}~
	\begin{center}
		\tikzfig{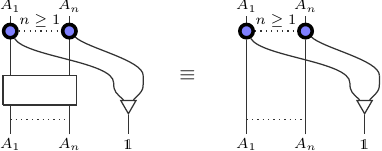}
	\end{center}
It follows that \eqPPtoCC{}, \eqCCtoC{} and \eqPPtoC{} can be generalised to the $n$-ary case.
\end{lemma}
\begin{proof}
	We follow the same proof as for \Cref{lem:contraction_to_X_oplus} and its following lemmas, but using the $n$-ary version of the equations when necessary.
\end{proof}
To tackle the last equation that was postponed, we first need this preliminary lemma:

\begin{lemma}\label{lem:nat_disj}
	The disjunction of wires is transparent with respect to most generators of $\FCat$.
	\begin{center}
		\tikzfig{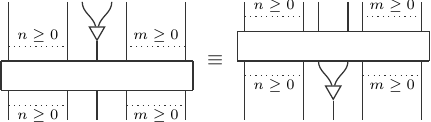}\\ [0,5cm]
		\tikzfig{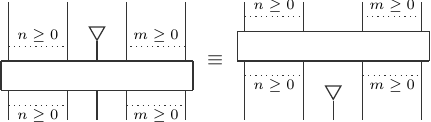}\\ [0,5cm]
		\tikzfig{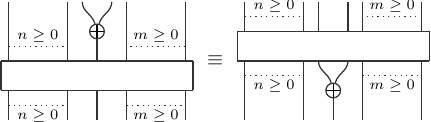}\\ [0,5cm]
		\tikzfig{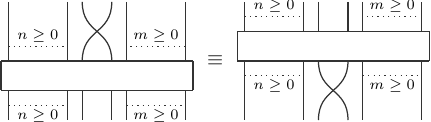}\\ [0,5cm]
		\tikzfig{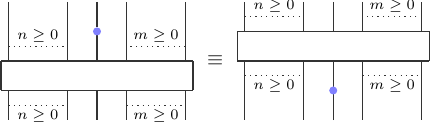}\\ [0,5cm]
		\tikzfig{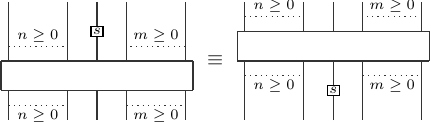}
	\end{center}
	The interaction of a disjunction of wires with Tensors is slightly different, as we have the following equation and its mirrored left/right:
\begin{center}
\tikzfig{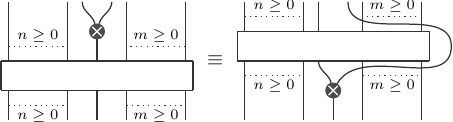}
\end{center}
\end{lemma}
\begin{proof}
	For all of them, we start by using \Cref{lem:disj_expansion}, and then, we just need the various generator to go through the binary disjunction.
	\begin{itemize}
		\item For the Contraction, we use \eqPPCleft{}
		\item For the Null, we use \eqPPNleft{}.
		\item For the Plus, we use \eqPPPleft{}.
		\item For the Swap, we use \eqSigmaP{}.
		\item For the Adapter, we use \eqPM{}.
		\item For the Scalar, we use \eqPPtoCC{} from \Cref{lem:n_ary_oplus_to_contraction}, \eqSpiderS{} and \eqS{}.
		\item For the Tensor, we use \eqPPtoCC{} from \Cref{lem:n_ary_oplus_to_contraction} and \eqSpiderS{}.
	\end{itemize}
\end{proof}

We can now prove the last induced equation of the functional fragment that was postponed.
\begin{lemma}\label{lem:oplus_contraction_bis}~
	\begin{center}
		\tikzfig{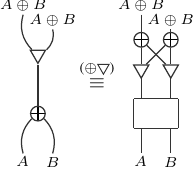}
	\end{center}
\end{lemma}
\begin{proof} We simply use \eqPC{} and then \Cref{lem:nat_disj}.
\end{proof}

We can also deduce the following equation, which we will use in \Cref{app:internal}.
\begin{lemma}\label{lem:PCC}~
	\begin{center}\tikzfig{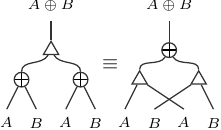}\end{center}
\end{lemma}
\begin{proof}~\\
	\begin{center}
		\tikzfig{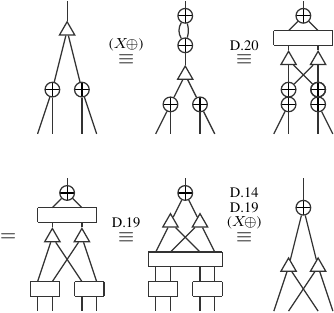}
	\end{center}
\end{proof}

Now, only remain equations relying on the unit. For that, we need one preliminary lemma:
\begin{lemma}\label{lem:unit_fusion}
	We have the following equation:
	\begin{center}
		\tikzfig{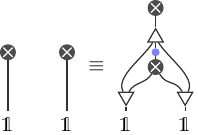}
	\end{center}
\end{lemma}
\begin{proof}~
	\begin{center}
		\tikzfig{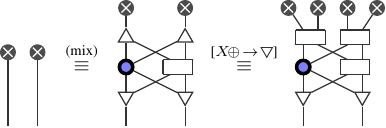}\\ [0.5cm]
		\tikzfig{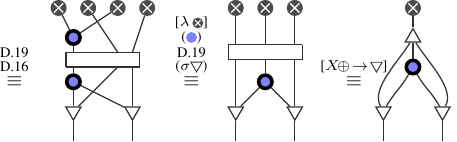}
	\end{center}
	where on the last step, we use the ternary version of \eqPPtoC{} from \Cref{lem:contraction_to_X_oplus_nary}.
\end{proof}

We can now prove \eqUsum{} and \eqUpar{}:

\begin{lemma}\label{lem:eq_unit_scalar}
	We have the following equation:
	\begin{center}
		\tikzfig{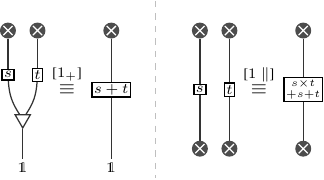}
	\end{center}
\end{lemma}
\begin{proof}~
	\begin{center}
		\tikzfig{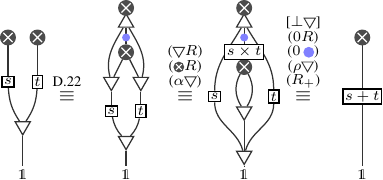}\\ [0.5cm]
		\tikzfig{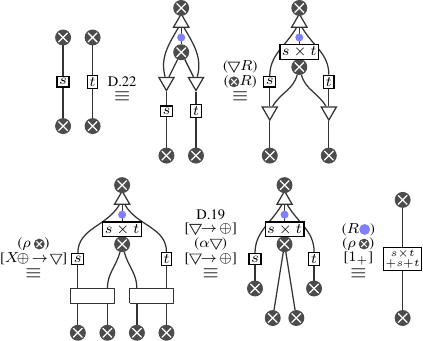}
	\end{center}
\end{proof}

\subsection{Conjunction of Disjunctions of Wires}

\begin{definition}\label{def:conjunction}
	We define the conjunction of disjunctions of wires $(A_1,\dots,A_a),(B_1,\dots,B_b),\dots, (Z_1,\dots,Z_z)$, for $a \geq 1$, $b \geq 1$, etc, $z \geq 1$ and any number $n \geq 2$ of such groups:
	\begin{center}
		\tikzfig{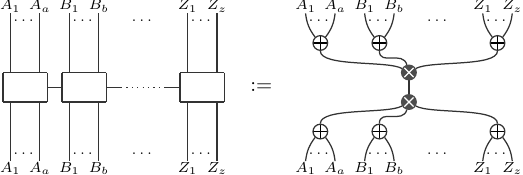}
	\end{center}
\end{definition}
Note the absence of $\maclane$ in this definition.

\begin{lemma}\label{lem:idempotence_conj}
	We have a property of idempotence:
	\begin{center}
		\tikzfig{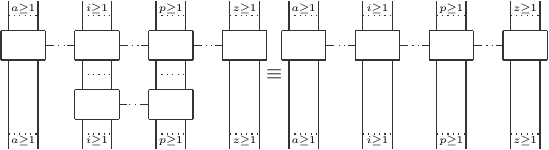}\\ [0.5cm]
		\tikzfig{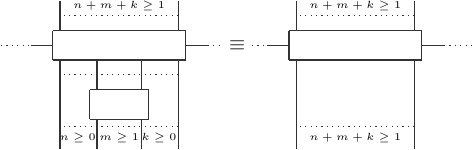}
	\end{center}
\end{lemma}
\begin{proof}
		We simply use \eqAlphaP{}, \eqPM{}, \eqP{} and their equivalent for $\tensor$, multiple times each.
\end{proof}

In order to prove \cref{lem:conj_expansion}, we need the following preliminary result:
\begin{lemma}\label{lem:conj_column}~
	\begin{center}
		\tikzfig{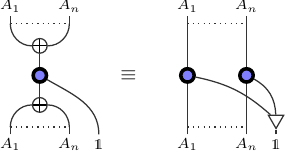}
	\end{center}
\end{lemma}
\begin{proof}
	We use the $n$-ary version of \eqPtoC{}, followed by \eqSpider{} to set up the use of the $n$-ary version of \eqPCtoT{}, and we conclude with \Cref{lem:contraction_to_X_oplus_nary}.
\end{proof}

\begin{lemma}\label{lem:conj_expansion}
	A conjunction of disjunctions of wires can be expanded as follows:
	\begin{center}
		\tikzfig{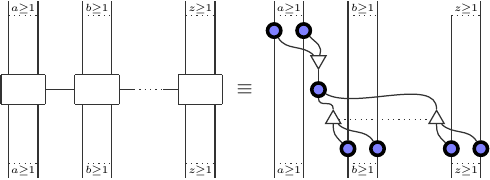}
	\end{center}
\end{lemma}
\begin{proof}
	We simply use \eqSpider{} to set up the use of \Cref{lem:conj_column}, which we use on each group of the conjunction of disjunctions.
\end{proof}

\begin{lemma}\label{lem:conj_expansion_complete}
	A conjunction of disjunctions of wires can be expanded as follows:
	\begin{center}
		\tikzfig{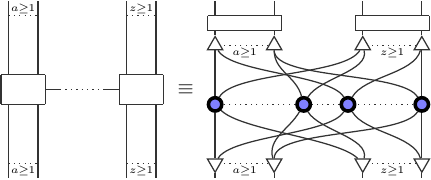}
	\end{center}
	Where on the right-hand side, there are $a \times \dots \times z$ spiders, each corresponding to the selection of exactly one wire from each group.
\end{lemma}
\begin{proof}
	We use \Cref{lem:conj_expansion}, and then resolve interactions of the Contractions with everything, that is using \eqTC{}, \eqTCleft{}, and \eqCC{}. Since \eqCC{} generates a binary disjunction of wires, we also need \eqPPtoCC{} to push those disjunction to the top of the diagram, and \Cref{lem:disj_expansion} to recompose a bigger disjunction of wires.
\end{proof}

We can generalize \Cref{lem:nat_disj} to conjunctions of disjunctions.

\begin{lemma}\label{lem:nat_conj}
	The disjunction of wires is transparent with respect to most generators of $\FCat$:
	\begin{center}
		\tikzfig{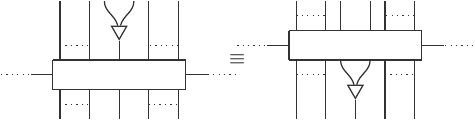}\\ [0,5cm]
		\tikzfig{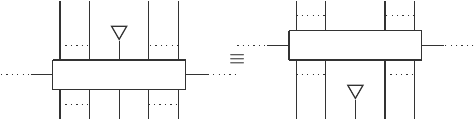}\\ [0,5cm]
		\tikzfig{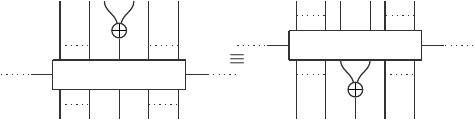}\\ [0,5cm]
		\tikzfig{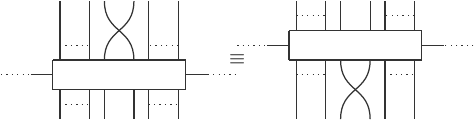}\\ [0,5cm]
		\tikzfig{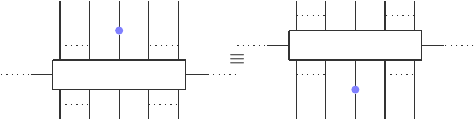}\\ [0,5cm]
		\tikzfig{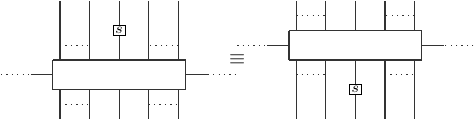}
	\end{center}
The interaction with Tensors is slightly different, as we have the following equation and its mirrored left/right:
\begin{center}
\tikzfig{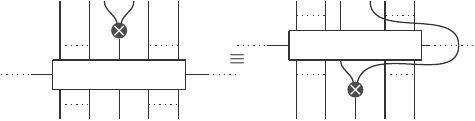}
\end{center}
\end{lemma}
\begin{proof}
	For all of them, we start by using \Cref{lem:conj_expansion}, and then, we just need the various generator to go through the binary disjunction.
	\begin{itemize}
		\item For the Contraction, we use \eqSpiderC{}
		\item For the Null, we use \eqSpiderN{} and \eqCN{}.
		\item For the Plus, we use \eqPtoC{} and \eqSpiderC{}.
		\item For the Swap, we use \eqSigmaC{}.
		\item For the Adapter, we use \eqSpider{}.
		\item For the Scalar, we use \eqSpiderS{} and \eqS{}.
		\item For the Tensor, we use \eqSpider{}.
	\end{itemize}
\end{proof}

\begin{lemma}\label{lem:nat_disj_conj_tensor_sychro}
	In \Cref{lem:nat_disj,lem:nat_conj}, whenever a we had an interaction with Tensors, we had to choose between left and right, but this choice is in fact not always necessary:
	\begin{center}
		\tikzfig{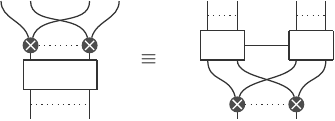}\\ [0,5cm]
		\tikzfig{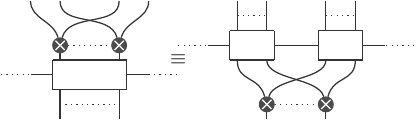}
	\end{center}
\end{lemma}
\begin{proof}
	For the first one, we start from the right-hand-side and use \Cref{lem:conj_expansion}. This puts us in a situation where we can use the $n$-ary version of \eqPPtoCC{} (which is proved in the same way as its binary version), and the result is exactly the left-hand-side of the first equation.

	For the second one, we have a conjunction of disjunctions and the $\tensor$s arrive on the $i$-th group. We use \Cref{lem:conj_expansion} unless $i = 1$ in which case we use a mirrored left/right version of that lemma (which is proved in the same way). We then use the following rewriting sequence:
	\begin{center}
		\tikzfig{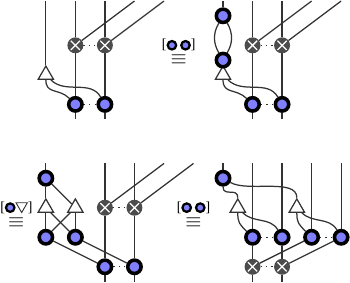}
	\end{center}
\end{proof}

\begin{lemma}\label{lem:disj_and_conj}
	Lastly, we have a property that generalizes the equation \eqmix{} by stating that two disjunctions in parallel is that same as a conjunction of dijunctions "+" a greater disjunction:
	\begin{center}
		\tikzfig{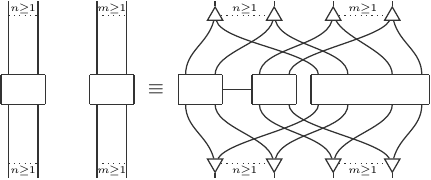}
	\end{center}
\end{lemma}
\begin{proof}
	We simply unfold the definition of the disjunction, use \eqmix{} on the two wires in the middle, then use \eqPC{} to push the contractions toward the inputs/outputs, and then \eqPPPleft{}, \eqAlphaP{}, \eqP{}, and \Cref{lem:disj_expansion} to rearrange the $\oplus$ in order to be able to recreate a disjunction and a conjunction of disjunctions.
\end{proof}

\clearpage
\section{Existence of the Normal Form}

In this appendix, we prove that every morphism can be put in normal form. The result of completeness  (\Cref{thm:complete_fun}) is an immediate consequence of the result of this appendix combined with the result of the next appendix (that two normal forms with the same semantics are necessarily equal).

\subsection{Converting the Objects}

It is a well known result in real vector species that any vector space is isomorphic to $\mathbb{R}^n$ where $n$ is its dimension. We could generalize this result and build an isomorphism between any object $\wireSetA$ and $1 \oplus \dots \oplus 1$ with a number of $1$ equal to the "dimension" of $\wireSetA$, however it will be more practical to build an "almost"-isomorphism between $\wireSetA$ and $1 \parallel \dots \parallel 1$ instead.

\begin{definition}\label{def:iso}
	For any object $\wireSetA$, we define its dimension $\dim(\wireSetA)$ as follows:
	\[ \dim(1) := 1 \qquad  \dim(A \tensor B) := \dim(A) \times \dim(B)\]
	\[ \dim(0) = \dim(\varnothing) := 0 \qquad \dim(A\oplus B) := \dim(A)+\dim(B) \]\[ \dim(\wireSetA \parallel B) := \dim(\wireSetA) \times \dim(B) + \dim(\wireSetA) + \dim(B)\]
	We define the diagrams $\iso_{\wireSetA} : \wireSetA \to 1\parallel \dots \parallel 1$ (with $\dim(\wireSetA)$ copies of $1$) inductively on the object $\wireSetA$ as follows:
	\begin{center}
		\tikzfig{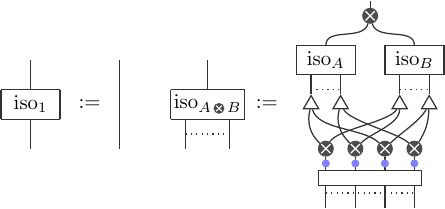}\\ [0.5cm]
		\tikzfig{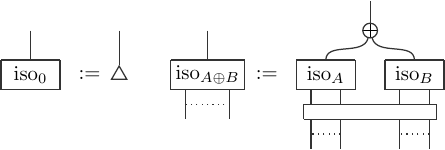}\\ [0.5cm]
		\tikzfig{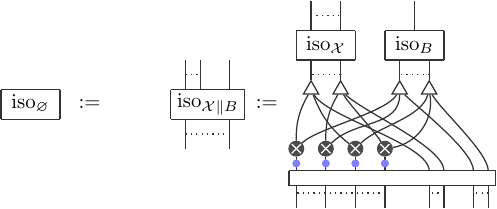}
	\end{center}
	The diagrams $\mirror{\iso_{\wireSetA}} :  1\parallel \dots \parallel 1 \to \wireSetA$ are defined symmetrically.
\end{definition}

Since this "almost"-isomorphism between $\wireSetA$ and $1 \parallel \dots \parallel 1$ comes from an isomorphism between $\wireSetA$ and $1 \oplus \dots \oplus 1$, the wires coming out of the "almost"-isomorphism are actually in disjunction.
The definition of the isomorphism on $\wireSetA\parallel B$ is asymmetric, however this asymmetry disappear when we quotient by $\equiv$, as shown in the following lemma.
\begin{lemma}\label{lem:iso_perm}
	For every object $\wireSetA$ and every color $B$, there exists a permutation of wires $\sigma_0$ such that:
	\begin{center}
		\tikzfig{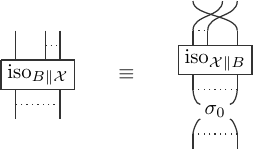}
	\end{center}
\end{lemma}
\begin{proof}
	We proceed by induction over $\wireSetA$. If $\wireSetA = \varnothing$, then the result is trivial. If $\wireSetA = \wireSetC\parallel D$, then we follow this procedure:
	\begin{enumerate}
		\item Since, $B \parallel \wireSetA = (B \parallel \wireSetC) \parallel D$ so we unfold the definition of $\iso_{(B \parallel \wireSetC) \parallel D}$.
		\item We can then use the induction hypothesis to transform $\iso_{B \parallel \wireSetC}$ into $\iso_{\wireSetC\parallel B}$.
		\item Unfolding the definition of $\iso_{\wireSetC\parallel B}$, we obtain:
		\begin{center}
			\tikzfig{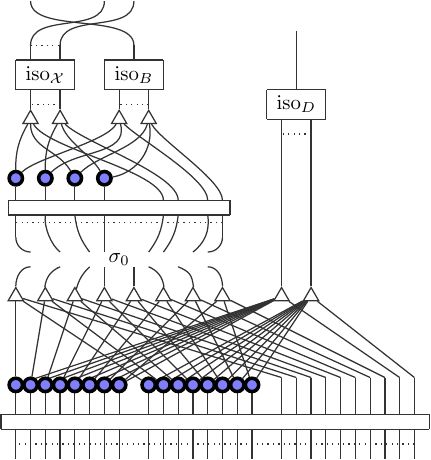}
		\end{center}
	\item
	We use \eqSigmaC{}, \eqLambdaRhoSpider{} and \Cref{lem:nat_disj} to push the permutation $\sigma_0$ to the very bottom. We call the new permutation at the bottom $\sigma_1$
	\item
	We use \eqSigmaP{} and \Cref{lem:nat_disj} to push the disjunction downward and \Cref{lem:idempotence_disj} to eliminate it.
	\item
	We use \eqSpiderC{} to push the Tensors downward.
	\item We can then swap $\iso_D$ and $\iso_B$ and push this swap downward using \eqSigmaC{}, \eqSigmaP{}, \eqLambdaRhoSpider{} (together with \eqM{}) and \eqSpider{}. This swap will merge with $\sigma_1$ to yield a new permutation $\sigma_2$ at the bottom.
	\item We would like to conclude to refold the definition of $\iso_{\wireSetC \parallel D}$, so we need to do the step 6, 5, and 4 in reverse beforehand.
	\item Lastly, we can refold the definition of $\iso_{\wireSetA \parallel B}$.
	\end{enumerate}
\end{proof}

\begin{lemma}\label{lem:iso_nat_disj}
	For every color $A$, we have:
	\begin{center}
		\tikzfig{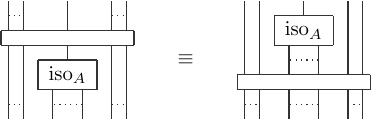}
	\end{center}
\end{lemma}
\begin{proof}
	We proceed by induction over $A$.
	\begin{itemize}
		\item If $A = 1$ or $A = 0$, the result is trivial.
		\item If $A = B \oplus C$, the results follows from \Cref{lem:nat_disj} and the induction hypothesis on $\iso_B$ and $\iso_C$.
		\item If $A = B \tensor C$, the results follows from \Cref{lem:nat_disj} (going to the left of every Tensor) and the induction hypothesis on $\iso_B$ (we do not use need induction hypothesis on $\iso_C$).
	\end{itemize}
\end{proof}

\begin{lemma}\label{lem:iso_nat_conj}
	The previous lemma generalizes to conjunctions of disjunctions, so we have for every color $A$:
	\begin{center}
		\tikzfig{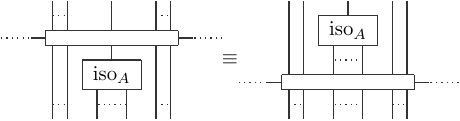}
	\end{center}
\end{lemma}
\begin{proof}
	The proof is the same as for \Cref{lem:iso_nat_disj}, using \Cref{lem:nat_conj} instead of \Cref{lem:nat_disj}.
\end{proof}

\begin{lemma}\label{lem:iso_contraction}
	For any color $A$, $\iso_A$ is duplicated by Contractions and erased by Null:
	\begin{center}
		\tikzfig{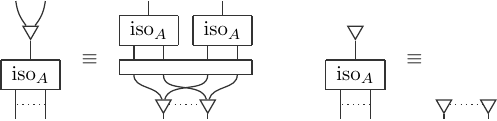}
	\end{center}
\end{lemma}
\begin{proof}
	For both, we proceed inductively on $A$.
	\begin{itemize}
		\item For $A = 0$ we use \eqLambdaC{} for the duplication and \eqNN{} for the erasure.
		\item For $A = 1$ we use \eqCtoP{} for the duplication and the erasure is trivial.
		\item For $A = B \oplus C$, then for the duplication we use \eqPC{}, the induction hypothesis and \Cref{lem:iso_nat_disj}, while for the erasure we use \eqPN{} and the induction hypothesis.
		\item For $A = B \tensor C$, then for the duplication we use \eqTC{}, the induction hypothesis and \eqAlphaC{},  while for the erasure we use \eqTN{}, the induction hypothesis and \eqLambdaC{}.
	\end{itemize}

	For the duplication, we use \eqCC{}, \eqAlphaC{} and \Cref{lem:idempotence_disj,lem:nat_disj}.

	For the erasure we use \eqCN{}, \eqSN{}, \eqLambdaC{}, \eqRhoC{}  and \Cref{lem:idempotence_disj,lem:nat_disj}.
\end{proof}

Before proving \Cref{prop:iso_is_iso}, we need the following preliminary lemma:
\begin{lemma}
	\label{lem:iso_is_iso_tensor}
	We have the following equation:
	\begin{center}
		\tikzfig{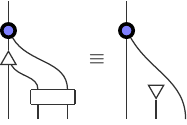}
	\end{center}
	Which implies:
	\begin{center}
		\tikzfig{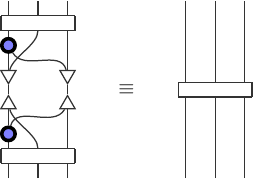}
	\end{center}
	Which can be generalized as follows:
	\begin{center}
		\tikzfig{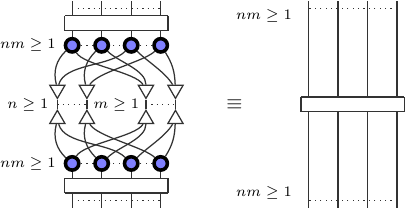}\\ [0.5cm]
		\tikzfig{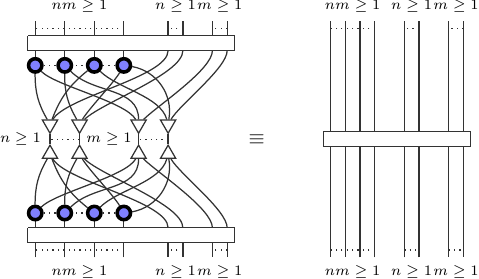}
	\end{center}
\end{lemma}
\begin{proof}
	For the first one:
	\begin{center}
		\tikzfig{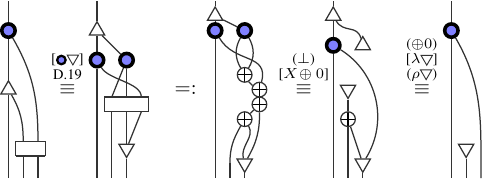}
	\end{center}
	For the second one, we start by using \eqCC{} on both pairs of contraction, and then we focus on eliminating each of the four "crossing" wires represented in red:
	\begin{center}
		\tikzfig{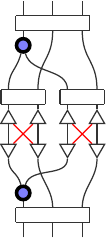}
	\end{center}
	The two binary disjunction can be eliminated by pushing them upward with \Cref{lem:nat_disj} and absorbing them with \Cref{lem:idempotence_disj}. Then, to eliminate the first red wire, we start by using \Cref{lem:nat_disj} and obtain:
	\begin{center}
		\tikzfig{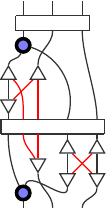}
	\end{center}
	Using \Cref{lem:disj_expansion}, we obtain the binary disjunction necessary to use the first part of this lemma, resulting in:
	\begin{center}
		\tikzfig{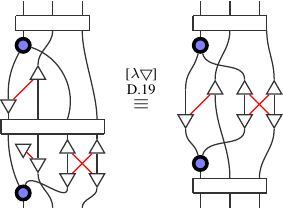}
	\end{center}
	So we successfully eliminated the first red wire. By iterating the same process for every red wire, we end up with:
	\begin{center}
		\tikzfig{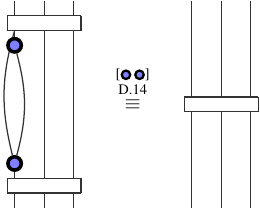}
	\end{center}
	For the third and fourth part, we can use \eqSpider{}, \eqAlphaC{} and \Cref{lem:idempotence_disj,lem:nat_disj} to set up the use the second part of this lemma $nm$ times.
\end{proof}

\begin{proposition}\label{prop:iso_is_iso}
	For any object $\wireSetA$, we have:
	\begin{center}
		\tikzfig{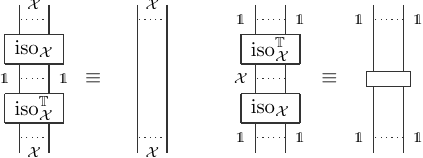}
	\end{center}
\end{proposition}
\begin{proof}
	We proceed by induction on $\wireSetA$, writing (IH) for the induction hypothesis:
	\begin{itemize}
		\item For the base cases $\tone$, $\tzero$ and $\varnothing$, it is trivial.
		\item For the case $\wireSetA = B \oplus C$, see \Cref{fig:proof_iso_is_iso_plus}.
		\item For the case $\wireSetA = B \otimes C$, see \Cref{fig:proof_iso_is_iso_tensor}.
		\item For the case $\wireSetA = \wireSetB \parallel C$, see \Cref{fig:proof_iso_is_iso_parallel}.
	\end{itemize}
\end{proof}

\begin{corollary}\label{cor:iso_is_iso}:
	For any object $\wireSetA$, we have $\oplus_{\tone,\dots,\tone} \circ \iso_{\wireSetA}$ is an isomorphism between $\wireSetA$ and $\bigoplus_{i=1}^{\dim(\wireSetA)} \tone$
\end{corollary}
\begin{proof}
	This follows from \Cref{prop:iso_is_iso,lem:idempotence_disj}.
\end{proof}

\begin{figure*}
	\[\tikzfig{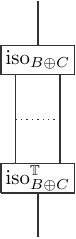}
	:=
	\tikzfig{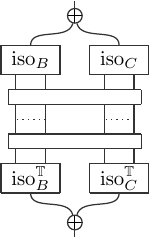}
	\TextEquiv{\parbox{0.5cm}{\ref{lem:idempotence_disj}\\\ref{lem:iso_nat_disj}\\\ref{lem:nat_disj}\\[-0.3cm]}}
	\tikzfig{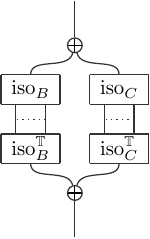}
	\TextEquiv{\parbox{0.5cm}{(IH)\\\eqP{}\\[-0.3cm]}}
	\tikzfig{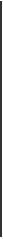}
	\hfill
	\tikzfig{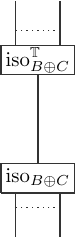}
	:=
	\tikzfig{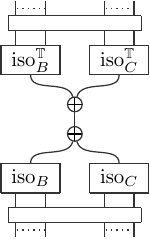}
	\TextEquiv{\parbox{0.5cm}{\ref{lem:iso_nat_disj}\\(IH)\\\ref{lem:idempotence_disj}\\[-0.3cm]}}
	\tikzfig{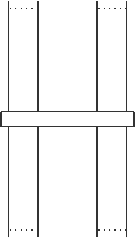}\]
	\caption{Proof of \Cref{prop:iso_is_iso}, $\oplus$ case.}
	\label{fig:proof_iso_is_iso_plus}
\end{figure*}

\begin{figure*}~\vspace{2cm}~
	\[\tikzfig{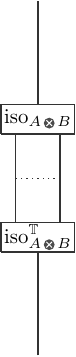}
	:=
	\tikzfig{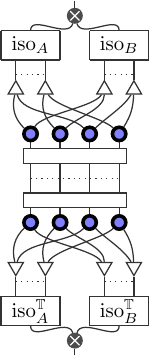}
	\TextEquiv{\ref{lem:iso_nat_disj}}
	\tikzfig{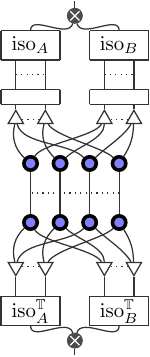}
	\TextEquiv{\ref{lem:conj_expansion_complete}}
	\tikzfig{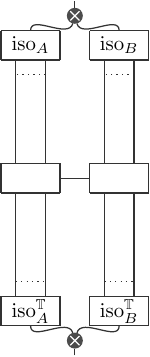}
	\TextEquiv{\parbox{0.5cm}{\ref{lem:iso_nat_conj}\\\ref{lem:nat_disj_conj_tensor_sychro}\\[-0.3cm]}}
	\tikzfig{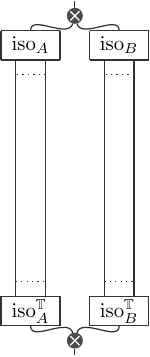}
	\TextEquiv{\parbox{0.5cm}{(IH)\\\eqT{}\\[-0.3cm]}}
	\tikzfig{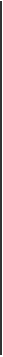}\]\[
	\tikzfig{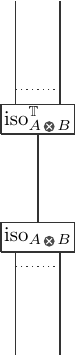}
	:=
	\tikzfig{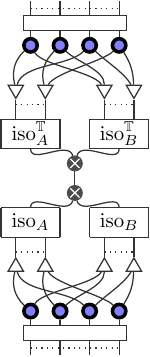}
	\TextEquiv{\parbox{0.5cm}{\ref{lem:iso_nat_conj}\\\ref{lem:nat_conj}\\\ref{lem:nat_disj_conj_tensor_sychro}\\[-0.3cm]}}
	\tikzfig{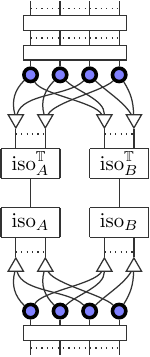}
	\TextEquiv{\parbox{0.5cm}{\ref{lem:idempotence_disj}\\(IH)\\[-0.3cm]}}
	\tikzfig{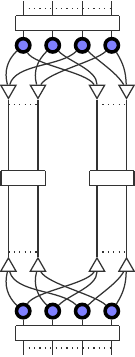}
	\TextEquiv{\parbox{0.5cm}{\ref{lem:nat_disj}\\\ref{lem:idempotence_disj}\\[-0.3cm]}}
	\tikzfig{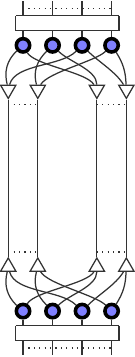}
	\TextEquiv{\ref{lem:iso_is_iso_tensor}}
	\tikzfig{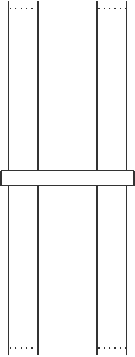}\]
	\caption{Proof of \Cref{prop:iso_is_iso}, $\tensor$ case.}
	\label{fig:proof_iso_is_iso_tensor}
\end{figure*}

\begin{figure*}\centering
	\[\tikzfig{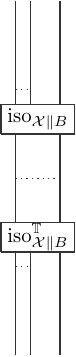}
	:=
	\tikzfig{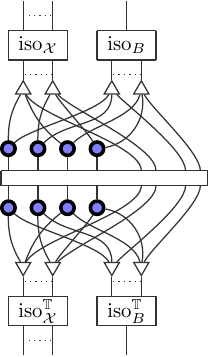}
	\TextEquiv{\parbox{0.5cm}{\ref{lem:disj_expansion}\\\ref{lem:idempotence_disj}\\\ref{lem:iso_nat_disj}\\\eqAlphaC\\[-0.3cm]}}
	\tikzfig{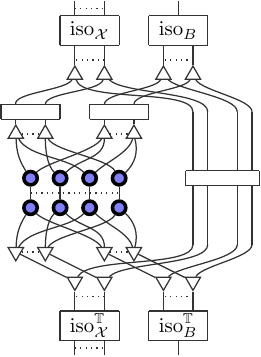}
	\TextEquiv{\parbox{0.5cm}{\ref{lem:conj_expansion_complete}\\\ref{lem:disj_and_conj}\\[-0.3cm]}}
	\tikzfig{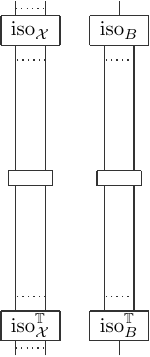}
	\TextEquiv{\parbox{0.5cm}{\ref{lem:idempotence_disj}\\(IH)\\[-0.3cm]}}
	\tikzfig{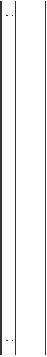}\]\[
	\tikzfig{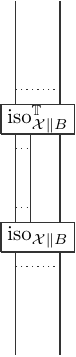}
	:=
	\tikzfig{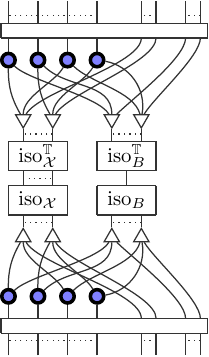}
	\TextEquiv{(IH)}
	\tikzfig{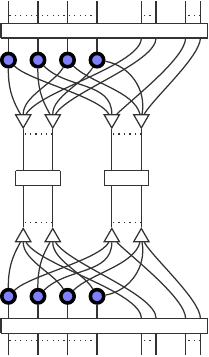}
	\TextEquiv{\parbox{0.5cm}{\ref{lem:nat_disj}\\\ref{lem:idempotence_disj}\\[-0.3cm]}}
	\tikzfig{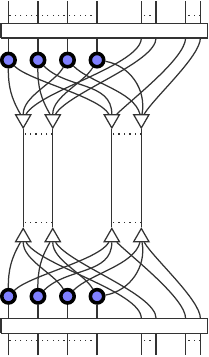}
	\TextEquiv{\ref{lem:iso_is_iso_tensor}}
	\tikzfig{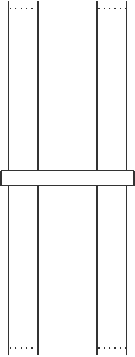}
	\]
	\caption{Proof of \Cref{prop:iso_is_iso}, $\parallel$ case.}
	\label{fig:proof_iso_is_iso_parallel}
\end{figure*}

\begin{figure*}~\vspace{2cm}~
	\[\tikzfig{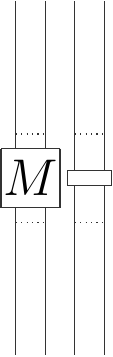}
	\TextEquiv{\parbox{0.5cm}{\ref{lem:idempotence_disj}\\\ref{lem:disj_and_conj}\\[-0.3cm]}}
	\tikzfig{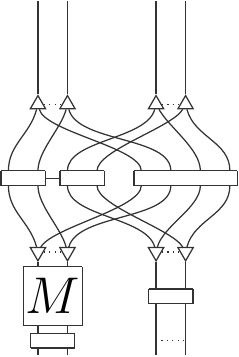}
	\TextEquiv{\ref{lem:conj_expansion_complete}}
	\tikzfig{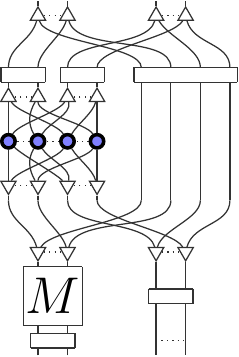}
	\TextEquiv{\parbox{0.5cm}{\ref{lem:mat_contraction}\\\ref{lem:mat_through_tensor}\\\ref{lem:nat_disj}\\[-0.3cm]}}
	\tikzfig{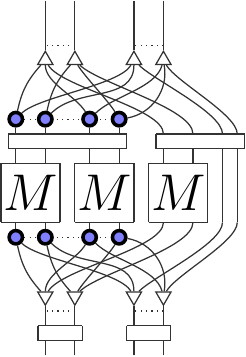}\]
	\caption{Proof of \Cref{lem:mat_parallel_id}}
	\label{fig:proof_mat_parallel_id}
\end{figure*}

\subsection{Representing Matrices}

\begin{definition}\label{def:mat}
	For any $m \times n$ matrix $M$ with coefficients from the commutative semiring $R$, we define the diagram $[M] : \tone \parallel \dots \parallel \tone \to \tone\parallel \dots \parallel \tone$ as follows:
	\begin{center}
		\tikzfig{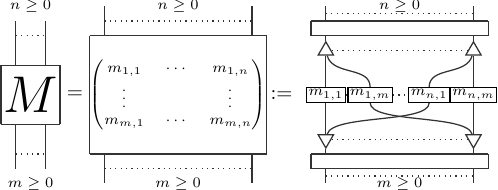}
	\end{center}
\end{definition}

Note that using \eqSzero{}, \eqAlphaC{}, \eqLambdaC{} and \eqRhoC{}, a coefficient equal to zero in this matrix can be instead represented by having a "missing" wire. For example, the identity matrix can be written as follows:

\begin{center}
	\tikzfig{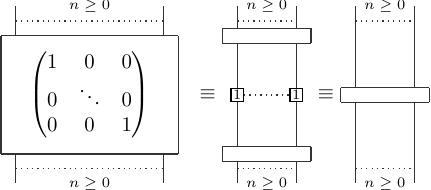}
\end{center}

\begin{proposition}\label{prop:mat_product}
	For any two matrices $M$ and $N$ with dimensions such that the $M \times N$ is defined, we have
	\[ [M] \circ [N] = [M \times N] \]
\end{proposition}
\begin{proof}
	We assume $M$ is a $p \times m$ matrix and $N$ is a $m \times n$ matrix.

	We use \eqCC{}, \eqSC{} and \Cref{lem:idempotence_disj,lem:nat_disj} to push the disjunctions (except the one at the very bottom) upward, the upward Contractions upward, and the downward Contractions downward.

	What remains is a lot of scalars in the middle. Whenever two scalars are in sequence, we use \eqSS{} to multiply them, and we note that we are always multiplying a "$n_{i,j}$" with a "$m_{j,k}$" for some $i,j,k$. For $1 \leq i \leq n$ and $1 \leq k \leq p$, we consider the $i$-th contraction at the top and the $k$-th contraction at the bottom. They are linked with $m$ wires, with for scalars $\{n_{i,j}\times m_{j,k} ~|~ 1 \leq j \leq m\}$. Using \eqSsum{}, we sum all of those scalars into $\sum_j n_{i,j}\times m_{j,k}$.

	What remains is exactly the matrix $M \times N$.
\end{proof}

\begin{lemma}\label{lem:mat_contraction}
	For any matrix $M$, its diagram is duplicated by disjunctions of Contractions and erased by Null:
	\begin{center}
		\tikzfig{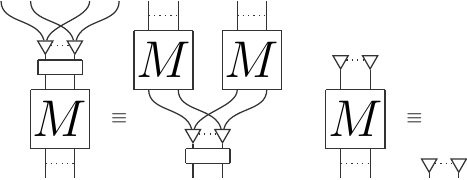}
	\end{center}
\end{lemma}
\begin{proof}
	For the duplication, we use \eqCC{}, \eqAlphaC{} and \Cref{lem:idempotence_disj,lem:nat_disj}.

	For the erasure we use \eqCN{}, \eqSN{}, \eqLambdaC{}, \eqRhoC{}  and \Cref{lem:idempotence_disj,lem:nat_disj}.
\end{proof}

\begin{lemma}\label{lem:mat_through_tensor}
	In order to prove \Cref{lem:mat_parallel_id}, we will need the following equation:
	\begin{center}
		\tikzfig{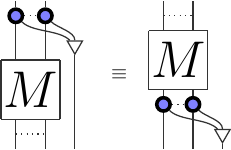}
	\end{center}
\end{lemma}
\begin{proof}
	We start from both sides, and unfold the definition of the matrix and then push the Spiders toward the middle using \Cref{lem:nat_disj} and \eqSpiderC{}. We can then equate both sides using \eqSpiderS{}, \eqS{} \eqAlphaC{} and \eqSigmaC{}.
\end{proof}

\begin{lemma}\label{lem:mat_parallel_id}
	For any matrix $M$ of size $m \times n$, we have:
	\begin{center}
		\tikzfig{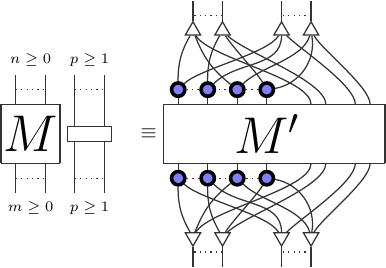}
	\end{center}
where $M'$ defined as the following block matrix of size $(pm+m+p) \times (pn+n+p)$:
\[ M' = \begin{pmatrix}
M &  & 0 & 0 & 0 \\
 & \ddots &  & 0 & 0 \\
0 &  & M & 0 &0 \\
0 & 0 & 0& M & 0 \\
0 & 0& 0& 0& \id
\end{pmatrix} \]
\end{lemma}
\begin{proof}
	We start by doing rewriting as described in \Cref{fig:proof_mat_parallel_id}. The only thing missing for the middle part to be a matrix $M'$ are two big disjunctions covering all the wires. For that we use \Cref{lem:idempotence_disj,lem:nat_disj} to move and duplicate the disjunctions and \Cref{lem:disj_expansion} to recompose a bigger disjunction from smaller ones.
\end{proof}

\subsection{Normal Form for the Functional \Langage}

\begin{lemma}\label{lem:maclane_decomposition}
	Any $\maclane_{A,A'}$ can be decomposed into the sequential and parallel compositions of Tensor, Plus, Null, $\maclane_{\tone\tensor \tone,\tone}$ and their mirrored version.
\end{lemma}
\begin{proof}
	Using \eqMM{}, we can decompose $\maclane_{A,A'}$ into more elementary $\maclane$, following the definition of $\isoML$. For all the case except $A \tensor \tone$ and $\tone \tensor A$, the result is immediate, as shown in the first two lines of \Cref{fig:proof_maclane_dec}. For $A \tensor \tone$ and $\tone \tensor A$, we proceed inductively on $A$ as shown in the third line of that figure.
\end{proof}

\begin{theorem}\label{thm:nf_fun}
	For every diagram $d \in \FCat_\equiv(\wireSetA,\wireSetB)$, there exists a $\dim(\wireSetB) \times \dim(\wireSetA)$ matrix $M$ such that $d$ can be put in the following normal form:
	\begin{center}
		\tikzfig{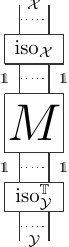}
	\end{center}
\end{theorem}
\begin{proof}
	We start by using \Cref{lem:maclane_decomposition} to rewrite all the $\maclane$ that are not $\maclane_{\tone\tensor \tone,\tone}$ or $\maclane_{\tone,\tone\tensor \tone}$.

	We then proceed inductively on $d$. We start with the generators as described in \Cref{fig:proof_nf_fun}. Then, for the inductive case:
	\begin{itemize}
		\item For $d = d_1 \circ d_2$, we simply use \Cref{prop:iso_is_iso,lem:idempotence_disj,prop:mat_product}.
		\item For $d = d_1\parallel \id_{C}$ with $C$ a color, we use \Cref{lem:mat_parallel_id} and remark that it directly builds the normal form (except for some disjunctions that can be obtained using \Cref{lem:idempotence_disj}).
		\item For $d = \id_{C} \parallel d_2$ with $C$ a color, we use \Cref{lem:iso_perm} and then reuse the previous case.
		\item For $d = d_1 \parallel d_2$, we use the bifunctoriality and associativity of $\parallel$ to reduce this case to the previous cases.
	\end{itemize}
\end{proof}

\begin{figure*}
	\centering
	\tikzfig{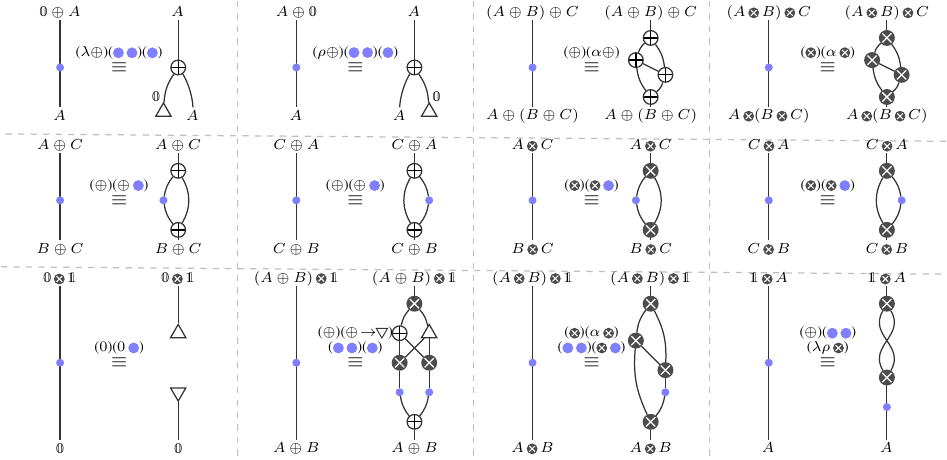}
	\caption{Decomposing $\maclane$.}
	\label{fig:proof_maclane_dec}
\end{figure*}

\begin{figure*}
	\centering
	$
	\tikzfig{lang/id} \quad
	\TextEquiv{\parbox{0.5cm}{\ref{prop:iso_is_iso}\\\ref{lem:idempotence_disj}\\[-0.3cm]}}  \quad
	\tikzfig{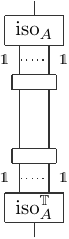}
	\hfill
	\tikzfig{lang/swap} \quad
	\TextEquiv{\parbox{0.5cm}{\ref{prop:iso_is_iso}\\\ref{lem:idempotence_disj}\\\ref{lem:iso_perm}\\[-0.3cm]}} \quad
	\tikzfig{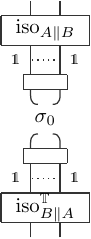}
	\hfill
	\tikzfig{lang/plus} \quad
	\TextEquiv{\parbox{0.5cm}{\ref{prop:iso_is_iso}\\\ref{lem:iso_nat_disj}\\\ref{lem:idempotence_disj}\\[-0.3cm]}} \quad
	\tikzfig{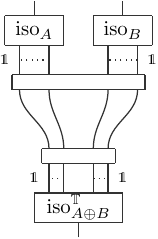} \quad
	\TextEquiv{\parbox{0.5cm}{\eqLambdaC{}\\\eqTN{}\\\ref{lem:nat_disj}\\[-0.3cm]}} \quad
	\tikzfig{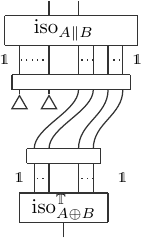}
	$\\ [0.5cm]$
	\tikzfig{lang/tensor-PN} \quad
	\TextEquiv{\parbox{0.5cm}{\ref{prop:iso_is_iso}\\\ref{lem:iso_nat_conj}\\\ref{lem:nat_conj}\\\ref{lem:nat_disj_conj_tensor_sychro}\\\ref{lem:idempotence_disj}\\[-0.3cm]}} \quad
	\tikzfig{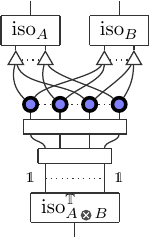}\quad
	\TextEquiv{\parbox{0.5cm}{\eqRhoC{}\\\ref{lem:nat_disj}\\[-0.3cm]}} \quad
	\tikzfig{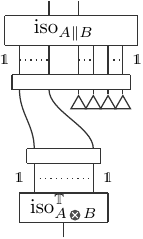}
	\hfill
	\tikzfig{lang/contraction} \quad
	\TextEquiv{\parbox{0.5cm}{\ref{prop:iso_is_iso}\\\ref{lem:iso_contraction}\\\ref{lem:idempotence_disj}\\[-0.3cm]}} \quad
	\tikzfig{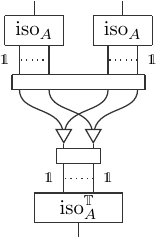} 	\quad
	\TextEquiv{\parbox{0.5cm}{\eqLambdaC{}\\\eqTN{}\\\ref{lem:nat_disj}\\[-0.3cm]}} \quad
	\tikzfig{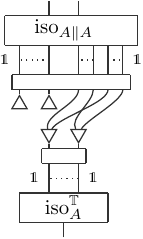}
	$\\ [0.5cm]$
	\tikzfig{lang/null} \quad
	\TextEquiv{\parbox{0.5cm}{\ref{prop:iso_is_iso}\\\ref{lem:idempotence_disj}\\\ref{lem:iso_contraction}\\[-0.3cm]}}  \quad
	\tikzfig{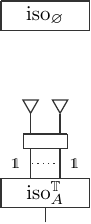}
	\hfill
	\tikzfig{lang/scal} \quad
	\TextEquiv{\parbox{0.5cm}{\ref{prop:iso_is_iso}\\\ref{lem:idempotence_disj}\\\eqPS{}\\\eqTS{}\\\eqCS{}\\[-0.3cm]}}  \quad
	\tikzfig{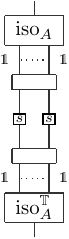}
	\hfill
	\tikzfig{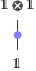} \quad
	\TextEquiv{\eqT} \quad
	\tikzfig{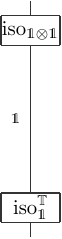}
	$
	\caption{Normal form for the generators.}
	\label{fig:proof_nf_fun}
\end{figure*}

\subsection{Normal Form for the \Langage}

In this section, we generalise the previous result in presence of the Unit generator. 

\begin{proposition}\label{prop:nf}
	For every diagram $d \in \Cat_\equiv(\wireSetA,\wireSetB)$, there exists a $(\dim(\wireSetB)+1) \times (\dim(\wireSetA)+1)$ matrix $M$ such that $d$ can be put in the following pseudo normal form:
	\begin{center}
		\tikzfig{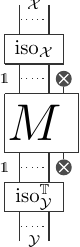}
	\end{center}
\end{proposition}
\begin{proof}
	We start by noting that if $d$ is functional, then we can simply use \Cref{thm:nf_fun} and then note that:
	\begin{center}
		\tikzfig{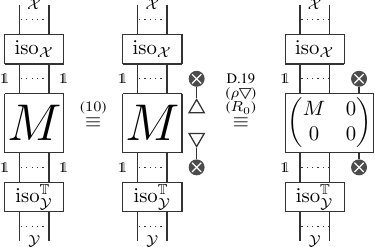}
	\end{center}
	If $d$ is not functional, we can factor all the Unit and obtain $d_1$ functional such that:
	\begin{center}
		\tikzfig{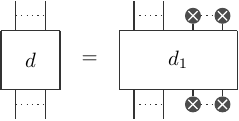}
	\end{center}
	Using \Cref{lem:unit_fusion}, we can obtain $d_2$ functional such that:
	\begin{center}
		\tikzfig{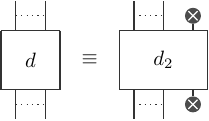}
	\end{center}
	Then, we can use \Cref{thm:nf_fun} on $d_2$:
	\begin{center}
		\tikzfig{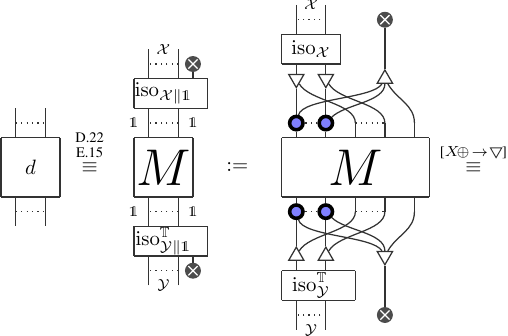}\\ [0.5cm]
		\tikzfig{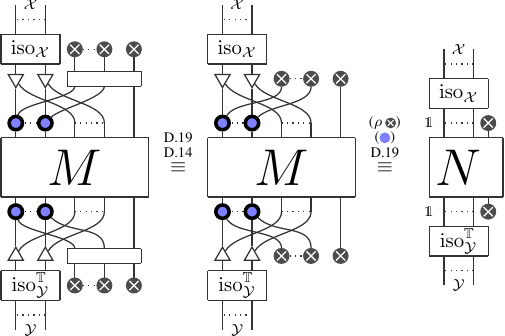}
	\end{center}
	Using the $n$-ary version of \eqPPtoC{}, and where $N$ is the matrix given by:
	\[
	M = \begin{pmatrix}
	A & B & C \\
	D & E & F \\
	G & H & x
	\end{pmatrix} \qquad N := \begin{pmatrix}
	A+B+D+E & C+F \\
	G+H & x
	\end{pmatrix}
	\]
	with $A,B,D,E$ being $\dim(\wireSetB) \times \dim(\wireSetA)$ matrices, $C,F$ being $\dim(\wireSetB) \times 1$ matrices, $G,H$ being $1 \times \dim(\wireSetA)$ matrices, and $x$ being a scalar.
\end{proof}

However, when $R$ is not cancellative, this result is not enough.

\begin{theorem}\label{thm:nf}
	For every $s \in R$, we define $[s]_{+1}$ as the equivalence class of $\{t ~|~ s+1 = t+1\}$. We arbitrarily chose\footnote{The axiom of choice might be required for having a unique normal form, however it is unnecessary for completeness by itself. The method would then be to give up uniqueness of the normal form but show that all the various normal forms can be rewritten into one another.} canonical representatives in each of those equivalences class.
	
	For every diagram $d \in \Cat_\equiv(\wireSetA,\wireSetB)$, there exists a $(\dim(\wireSetB)+1) \times (\dim(\wireSetA)+1)$ matrix $M$, in which the bottom-right coefficient is canonical, such that $f$ can be put in the following normal form:
	\begin{center}
		\tikzfig{normal_form/thm_nf}
	\end{center}
\end{theorem}
\begin{proof}
	We start by using \Cref{prop:nf} to obtain a pseudo normal form, and check the bottom right coefficient, which we write $s$. We then use \eqUS{} to replace $s$ by the canonical representative of $[s]_{+1}$.
\end{proof}

\clearpage
\section{Uniqueness of the Normal Form}

In this appendix, we prove the uniqueness of the normal form, that is that two normal forms with the same semantics are necessarily equal. The result of completeness  (\Cref{thm:complete_fun}) is an immediate consequence of the result of this appendix combined with the result of the previous appendix (that every morphism can be put in normal form).

We recall that as defined in \Cref{def:matrix_R}, we have a linear, full and faithful functor $\Morphism{-}$ which sends any $n \times m$ matrix $M$ with coefficient in $R$ to a morphism  of $\bfup{H}\left(\bigoplus_{i=1}^n \tone,\bigoplus_{j=1}^m \tone\right)$.

\subsection{The Functional Fragment}

\begin{proposition}\label{prop:sem_block_matrix}
	For any $m \times n$ matrix $M$, we have
	\[ \interp{\tikzfig{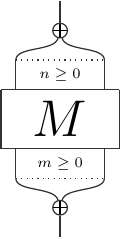}} = \Morphism{M} \in \bfup{H}\left(\bigoplus_{i=1}^n \tone,\bigoplus_{j=1}^m \tone\right)\]
\end{proposition}
\begin{proof}
	For that, we start by noting that the following diagram has for semantics the $i$-th (for $1 \leq i \leq n$) projection from $\bigoplus_{i=1}^n \tone$ to $\tone$:
	\begin{center}
		\tikzfig{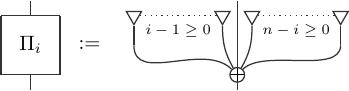}
	\end{center}
	Then, using the equational theory we can prove
	\begin{center}
		\tikzfig{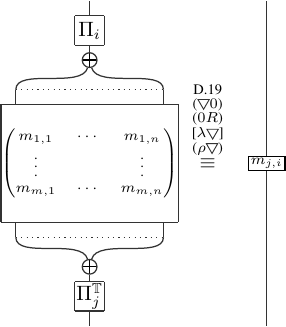}
	\end{center}
	Using the soundness (\Cref{prop:soundness_fun}), this allows us to check the result on every coefficient of the matrix.
\end{proof}

\begin{corollary}[Uniqueness of the Normal Form]\label{cor:nf_unique_fun}
	Given any diagrams $d,e \in \FCat(\wireSetA,\wireSetB)$ with $\interp{d} = \interp{e}$, if there exists two matrices $N$ and $M$ such that
	\begin{center}
		\tikzfig{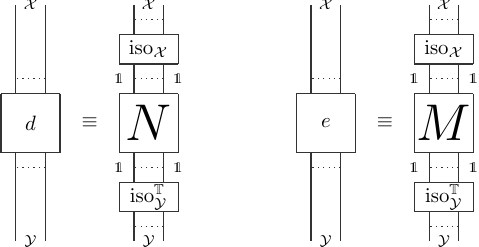}
	\end{center}
	then $N = M$ and $d \equiv e$.
\end{corollary}
\begin{proof}
	Using \Cref{prop:iso_is_iso,lem:idempotence_disj} we obtain:
	\begin{center}
		\tikzfig{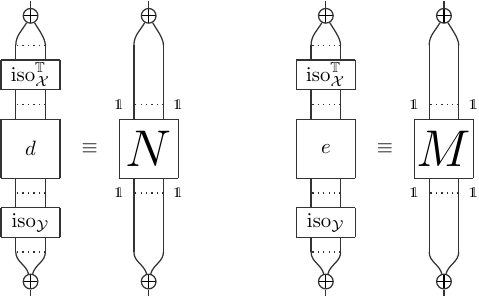}
	\end{center}
	Taking the semantics $\interp{-}$, and using the soundness (\Cref{prop:soundness_fun}) we obtain:
	\[ \interp{\oplus_{\tone,\dots,\tone} \circ \iso_{\wireSetA}}\circ \interp{d} \circ \interp{\mirror{\left(\oplus_{\tone,\dots,\tone} \circ \iso_{\wireSetA}\right)}} = \Morphism{N} \]
	\[ \interp{\oplus_{\tone,\dots,\tone} \circ \iso_{\wireSetA}}\circ \interp{e} \circ \interp{\mirror{\left(\oplus_{\tone,\dots,\tone} \circ \iso_{\wireSetA}\right)}} = \Morphism{M} \]
	Since we have $\interp{d} = \interp{e}$, it follows that $\Morphism{N} = \Morphism{M}$, hence $N=M$ by faithfullness (\Cref{prop:matrix_is_fff}), hence $d \equiv e$.
\end{proof}

\subsection{The Whole Calculus}
\begin{proposition}\label{prop:sem_block_matrix_unit}
	For any $m \times n$ matrix $M$, we have
	\[ \interp{\tikzfig{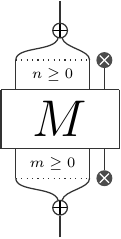}}^{\oplus \oone} = \Morphism{M} \bfup{~+~} \iota_\ell \circ \pi_r \in \bfup{H}^{\oplus \oone}\left(\bigoplus_{i=1}^n \tone,\bigoplus_{j=1}^m \tone\right)\]
	where the "$\bfup{+~}  \iota_\ell \circ \pi_r$" corresponds to adding +1 to the coefficient at the bottom right of the matrix.
\end{proposition}
\begin{proof}
	We start by decomposing the diagram in three different parts:
	\begin{center}
		\tikzfig{matrix_nf/proof_sem_matrix_1}
	\end{center}
	Then, we take the semantics $\interp{-}^{\oplus \oone}$, and for the middle part, since there is no Unit, the semantics is simply $\interp{-}^{\oplus \oone} = \interp{-} \oplus \id_{\tone}$ so we can use \Cref{prop:sem_block_matrix}. We then simply compute the result.
\end{proof}

\begin{corollary}[Uniqueness of the Normal Form]\label{cor:nf_unique}
	Given any diagrams $d,e \in \Cat(\wireSetA,\wireSetB)$ with $\interp{d}^{\oplus \oone} = \interp{e}^{\oplus 1}$, if there exists two matrices $N$ and $M$ such that
	\begin{center}
		\tikzfig{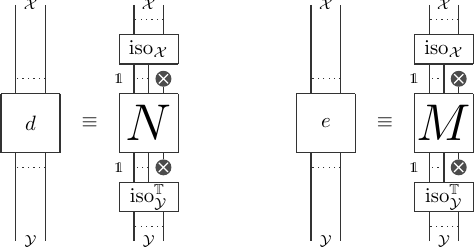}
	\end{center}
	then $N = M$ and $d \equiv e$.
\end{corollary}
\begin{proof}
	The proof starts by the same reasonning as \Cref{cor:nf_unique_fun}, but using \Cref{prop:sem_block_matrix_unit} instead of \Cref{prop:sem_block_matrix}. We obtain $ \Morphism{M} \bfup{~+~} \iota_\ell \circ \pi_\ell =  \Morphism{N} \bfup{~+~} \iota_\ell$. Using faitfullness and linearity (\Cref{prop:matrix_is_fff}), it means that "$M$ with +1 added to the bottom right corner's coefficient" is equal to "$N$ with a +1 added to the bottom right corner's coefficient". Using the fact that the bottom right corners of $M$ and $N$ are canonical, it follows that $M = N$.
\end{proof}
\clearpage
\section{Universality}
\label{app:universality}
In this appendix, we prove the universality result of \Cref{thm:universality}, both for the functional fragment and for the whole calculus.

\subsection{The Functional Fragment}
For $\wireSetA$ an object of $\Cat$, we define
\[ \tikzfig{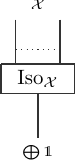} := \tikzfig{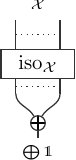} \]
and similarly for $\mirror{\Iso_{\wireSetA}}$.

\begin{lemma}\label{lem:Iso_is_iso}
	The diagram $\Iso_{\wireSetA}$ is an isomorphism (up to $\equiv$) with for inverse $\mirror{\Iso_{\wireSetA}}$.
\end{lemma}
\begin{proof}
	To prove $\Iso_{\wireSetA} \circ \mirror{\Iso_{\wireSetA}} \equiv \id$, we simply use  \Cref{prop:iso_is_iso} then \eqPP.
	To prove $\mirror{\Iso_{\wireSetA}} \circ \Iso_{\wireSetA} \equiv \id$, we simply use start by noting that in \Cref{def:iso}, the diagram $\iso_{\wireSetA}$ always has a disjunction of wires at the end, so using \Cref{lem:idempotence_disj} we can eliminate the Plus and obtain $\mirror{\Iso_{\wireSetA}} \circ \Iso_{\wireSetA} \equiv \mirror{\iso_{\wireSetA}} \circ \iso_{\wireSetA}$, and conclude with \Cref{prop:iso_is_iso}.
\end{proof}

\begin{theorem}
	For any $\wireSetA,\wireSetB$ objects of $\FCat$, for every morphism $f \in \bfup{H}(\interp{\wireSetA},\interp{\wireSetB})$, there exists a diagram $d \in \FCat(\wireSetA,\wireSetB)$ such that $\interp{d} = f$.
\end{theorem}
\begin{proof}
	Using the fullness of $\Morphism{-}$ (\Cref{prop:matrix_is_fff}), there exists a matrix $M$ such that
	\[ \Morphism{M} = \interp{\Iso_{\wireSetB}} \circ f \circ \interp{\mirror{\Iso_{\wireSetA}}} \]
	Using \Cref{prop:sem_block_matrix}, we then obtain that	
	\[ \interp{\tikzfig{matrix_nf/prop_sem_matrix}} = \interp{\Iso_{\wireSetB}}  \circ f \circ \interp{\mirror{\Iso_{\wireSetA}}}\]
	Using \Cref{lem:Iso_is_iso}, it follows that
	\[ \interp{\tikzfig{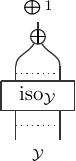}} \circ \interp{\tikzfig{matrix_nf/prop_sem_matrix}} \circ \interp{\tikzfig{normal_form/iso_then_plus}} = f  \]
	We then combine the three diagrams into one using functoriality of $\interp{-}$. Given that the diagram $\iso_{\wireSetA}$ always has a disjunction of wires at the end, we can use \Cref{lem:idempotence_disj} to remove the unnecessary $\oplus$, so we then have
	\[ \interp{\tikzfig{normal_form/thm_nf_fun}} = f \]
\end{proof}

\subsection{The Whole Calculus}

\begin{theorem}
	For any $\wireSetA,\wireSetB$ objects of $\Cat$, for every morphism $f \in \bfup{H}^{\oplus \oone}(\interp{\wireSetA},\interp{\wireSetB})$, there exists a diagram $d \in \Cat(\wireSetA,\wireSetB)$ such that $\interp{d}^{\oplus \oone} = f$.
\end{theorem}
\begin{proof}
	By definition of $\bfup{H}^{\oplus \oone}$, there exists $g$ such that $f = g \bfup{~+~} \iota_r \circ \pi_r$.
	
	Using the fullness of $\Morphism{-}$ (\Cref{prop:matrix_is_fff}), there exists a matrix $M$ such that
	\[ \Morphism{M} = \interp{\Iso_{\wireSetB \oplus \tone}} \circ g \circ \interp{\mirror{\Iso_{\wireSetA \oplus \tone}}}  \]
	Unfolding a little bit the definition of $\Iso_{\wireSetB \oplus \tone}$ and computing its semantics, we end up with $\Iso_{\wireSetB \oplus \tone} = \interp{\Iso_{\wireSetB}} \oplus \id$, hence
	\[ \Morphism{M} = (\interp{\Iso_{\wireSetB}} \oplus \id) \circ g \circ (\interp{\mirror{\Iso_{\wireSetA}}} \oplus \id) \]
	Using \Cref{prop:sem_block_matrix_unit}, we then obtain that	
	\[ \interp{\tikzfig{matrix_nf/prop_sem_matrix_unit}} = (\interp{\Iso_{\wireSetB}} \oplus \id) \circ g \circ (\interp{\mirror{\Iso_{\wireSetA}}} \oplus \id) \bfup{~+~} \iota_r \circ \pi_r\]
	Since $\iota_r = (n_{\interp{\wireSetB}} \oplus \id)  \circ \iota_r$ and $\pi_r = \pi_r \circ (\interp{\mirror{\Iso_{\wireSetA}}}\oplus \id)$, and using linearity of the composition, we have
	\[ \interp{\tikzfig{matrix_nf/prop_sem_matrix_unit}} = (\interp{\Iso_{\wireSetB}}\oplus \id) \circ (g\bfup{~+~} \iota_r \circ \pi_r) \circ (\interp{\mirror{\Iso_{\wireSetA}}}\oplus \id) \]
	Hence 
	\[ \interp{\tikzfig{matrix_nf/prop_sem_matrix_unit}} = (\interp{\Iso_{\wireSetB}} \oplus \id) \circ f \circ (\interp{\mirror{\Iso_{\wireSetA}}} \oplus \id) \]
	Using \Cref{lem:Iso_is_iso}, it follows that
	\[ \interp{\tikzfig{normal_form/iso_then_plus_inv}} \circ \interp{\tikzfig{matrix_nf/prop_sem_matrix_unit}} \circ \interp{\tikzfig{normal_form/iso_then_plus}} = f  \]
	We then combine the three diagrams into one using functoriality of $\interp{-}$. Given that the diagram $\iso_{\wireSetA}$ always has a disjunction of wires at the end, we can use \Cref{lem:idempotence_disj} to remove the unnecessary $\oplus$, so we then have:
	\[ \interp{\tikzfig{normal_form/thm_nf}} = f \]
\end{proof}

\clearpage
\section{\Langage{} as an Internal Language}
In this appendix, we prove the claims made in \Cref{sec:internal} that our language is an internal language for semiadditive categories, with a symmetric monoidal structure distributive over it, and such that the homset of automorphisms over the latter's unit are isomorphic to $R$.

We start by defining the subcategory with a single input and a single output, then prove our property for the functional fragment, and then a weaker property for the full language.

\subsection{Restricting to Single Input and Output}
\label{app:singleIO}

\begin{figure*}
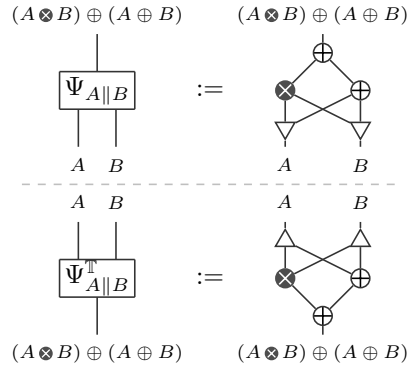

	\centering
	\tikzfig{internal/bar_to_parallel}
	\caption{Isomorphism between $A \parallel B$ and $(A|B) := (A \tensor B) \oplus (A \oplus B)$.}
	\label{fig:bar_to_parallel_app}
\end{figure*}

We define $\sFCat$ as the full subcategory of $\FCat$ where objects are colors, hence morphisms are all the morphisms of $\FCat$ with a single input and a single output. We define $\sCat$ similarly.

We want to show that $\sFCat_\equiv$ and $\FCat_\equiv$ are equivalent, and so are $\sCat_\equiv$ and $\Cat_\equiv$. Since the inclusion functors is full and faithful, we simply need to show that they are essentially surjective, in other words that:
\begin{itemize}
	\item for every object $\wireSetA$, we can define a color $\bfup{color}(\wireSetA)$
	\item such that we have a natural isomorphism $\Psi_{\wireSetA} : \bfup{color}(\wireSetA) \to \wireSetA$ for $\FCat_\equiv$
	\item that is also a natural ismorphism for $\Cat_\equiv$
\end{itemize}

We remark that $\Psi_{A \parallel B}$ as defined in \Cref{fig:bar_to_parallel_app} is an isomorphism of $\FCat_\equiv$ between the two-colors object $A \parallel B$ one-color object $(A \tensor B) \oplus (A \oplus B)$. This can be proven either by showing that its semantics is the identity and then using \Cref{thm:complete_fun}, or by hand in a very similar way to the proof that $\iso_{A\parallel B}$ is an isomorphism in the appendices. Iterating the use of this isomorphism, we can build an isomorphism
\[  \Psi_{A_1 \parallel \dots \parallel A_n} : \bfup{color}(A_1 \parallel \dots \parallel A_n) \to A_1 \parallel \dots \parallel A_n \]
\[ \text{where }\begin{cases} \bfup{color}(\varnothing) := \tzero \\ \bfup{color}(A) := A \\ \bfup{color}(A\parallel B) := (A \tensor B) \oplus (A \oplus B) \\ \bfup{color}(\wireSetA \parallel B) := \bfup{color}(\bfup{color}(\wireSetA) \parallel B)\end{cases}\]
The operation $\bfup{color}$ can be extended as a functor for $d \in \Cat_\equiv(\wireSetA,\wireSetB)$ with $\bfup{color}(d) = \Psi^{-1}_{\wireSetB} \circ d \circ \Psi_{\wireSetA}$. It follows that $\Psi$ is a natural isomorphism, both for $\Cat_\equiv$ and for $\FCat_\equiv$.

\subsection{The Functional Fragment}
\label{app:internal}

\begin{figure*}
	\centering
	\tikzfig{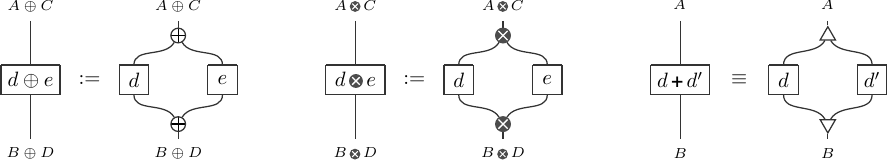}
	\caption{The bifunctors $\oplus$ and $\tensor$, and induced enrichement $\bfup{+}$, for $d,d',e$ in $\sFCat[R]$.}
	\label{fig:bifunctor_def}
\end{figure*}

\begin{figure*}
	\centering
	\tikzfig{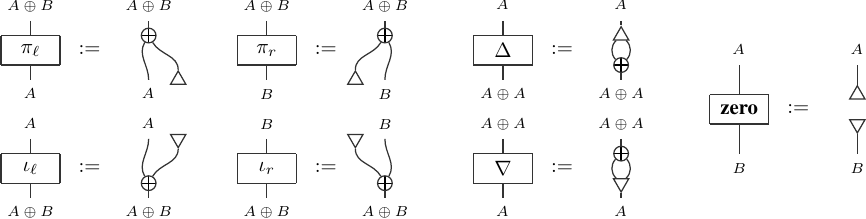}
	\caption{The projections $\pi$, injections $\iota$, null $\zero$, diagonal $\Delta$, co-diagonal $\nabla$, and the morphism $\zero$.}	\label{fig:pairing_def}
\end{figure*}

\begin{figure*}
	\centering
	\tikzfig{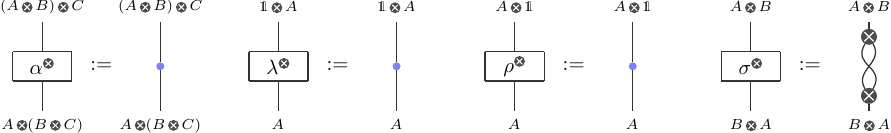}
	\caption{The associator $\alpha^{\tensor}$, left-unitor $\lambda^{\tensor}$, right-unitor $\rho^{\tensor}$ and swap $\sigma^{\tensor}$.}	\label{fig:smc_def}
\end{figure*}

In \Cref{thm:internal_language_fun}, we claim that $\sFCat[R]_\equiv$ is an \emph{internal language} for "semiadditive categories with an symmetric monoidal structure which is distributive, together with a semiring isomorphism from $R$ to the homset of the multiplicative unit" or $R$-distributive. We provide here a proof of this claim, in two steps:
\begin{enumerate}
	\item We prove that $\sFCat[R]_\equiv$ is a $R$-distributive category.
	\item Since in \Cref{sec:cat_sem}, we provided a categorical semantics $\interp{-}$ from $\sFCat[R]_\equiv$ to any such category, so we can see $\interp{-}$ as going from $\FCat[R]_\equiv$ to itself. We prove that $\interp{-} : \sFCat[R]_\equiv \to \sFCat[R]_\equiv$ is the identity.
\end{enumerate}

In order to avoid confusion, we keep $\interp{-}_{\bfup{H}}$ for the semantics toward an arbitrary category $\bfup{H}$ that we know to be a $R$-distributive category, for example matrices with coefficients in $R$, and we write $\interp{-}_{\circlearrowleft}$ for the semantics toward $\sFCat$.

	And for any two diagrams $d \in \sFCat(A,B)$ and $e \in \sFCat(C,D)$ we define $d \oplus e$ and $d \tensor e$ as in \Cref{fig:bifunctor_def}. We note that $\interp{d \oplus e}_{\bfup{H}} \equiv \interp{d}_{\bfup{H}} \oplus \interp{e}_{\bfup{H}}$ and $\interp{d \tensor e}_{\bfup{H}} \equiv \interp{d}_{\bfup{H}} \tensor \interp{e}_{\bfup{H}}$, so using the completeness, we obtain that they are bifunctors in $\sFCat_\equiv$. Together with the definitions of \Cref{fig:pairing_def}, and using the completeness to prove the universal diagrams, we have that $(\sFCat_\equiv,\oplus,\tzero)$ is a semiadditive category. Then, adding the definitions of \Cref{fig:smc_def} and using again the completeness to prove the coherence diagrams, we have that $(\sFCat_\equiv,\tensor,\tone)$ is a symmetric monoidal category and satisfies distributivity. Lastly, we have a semiring isomorphism $s \in R \mapsto [s]_1 : \tone \to \tone$, the inversibility coming from the uniqueness of the normal form (\Cref{cor:nf_unique_fun}).
	This shows that $\sFCat[R]$ is indeed a $R$-distributive category.

	We now prove that $\interp{-}_{\circlearrowleft}$ is the identity. It is immediate on objects, so we need to prove that for all $d \in \sFCat[R](A,B)$, $\interp{d}_{\circlearrowleft} = d$.
	
	In order to proceed inductively, we generalize the statement as follows: for all $f \in \FCat[R]_\equiv(\wireSetA,\wireSetB)$, $\interp{d}_{\circlearrowleft} = \Psi^{-1}_{\wireSetB} \circ d \circ \Psi_{\wireSetA}$, with $\Psi$ the natural isomorphism defined in \Cref{app:singleIO}. The proof is an inductive proof that relies on completeness, that is \Cref{thm:complete_fun}. For every generator of $G \in \FCat(\wireSetA,\wireSetB)$, we prove
	\[ \interp{\interp{G}_{\circlearrowleft}}_{\bfup{H}} = \interp{\Psi^{-1}_{\wireSetB} \circ G \circ \Psi_{\wireSetA}}_{\bfup{H}}\]
	using the matrix notation as in \Cref{fig:mat_sem_fun}, this is some simple matrix computation, for example with $G$ being the Plus $\oplus_{A,B} : A \parallel B \to (A \oplus B)$, on the left hand side we have:
	\[ \interp{\interp{G}_{\circlearrowleft}}_{\bfup{H}} = \interp{\pi_r}_{\bfup{H}} = \interp{\tikzfig{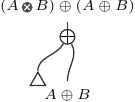}}_{\bfup{H}} \]
	\[\begin{array}{rl} =&  
	\begin{blockarray}{cccccc}
	& \text{\scalebox{0.7}{$(\interp{A}\tensor \interp{B}) \tensor \interp{A}$}} & \text{\scalebox{0.7}{$(\interp{A}\tensor \interp{B}) \tensor \interp{B}$}} & \text{\scalebox{0.7}{$\interp{A}\tensor \interp{B}$}} & \text{\scalebox{0.7}{$\interp{A}$}} & \text{\scalebox{0.7}{$\interp{B}$}} \\
	\begin{block}{c(ccccc)}
	\text{\scalebox{0.7}{$\interp{A}$}}  &&&&\id&\\
	\text{\scalebox{0.7}{$\interp{B}$}}  &&&&&\id\\
	\end{block}
	\end{blockarray} \\
	\circ&
	\begin{blockarray}{cccc}
	& \text{\scalebox{0.7}{$\interp{A}\tensor \interp{B}$}} & \text{\scalebox{0.7}{$\interp{A}$}} & \text{\scalebox{0.7}{$\interp{B}$}} \\
	\begin{block}{c(ccc)}
	\text{\scalebox{0.7}{$(\interp{A}\tensor \interp{B}) \tensor \interp{A}$}}&&&\\ 
	\text{\scalebox{0.7}{$(\interp{A}\tensor \interp{B}) \tensor \interp{B}$}}&&&\\
	\text{\scalebox{0.7}{$\interp{A}\tensor \interp{B}$}} &\id&&\\
	\text{\scalebox{0.7}{$\interp{A}$}}  &&\id&\\
	\text{\scalebox{0.7}{$\interp{B}$}}  &&&\id\\
	\end{block}
	\end{blockarray} \\ =& 
	\begin{blockarray}{cccc}
	& \text{\scalebox{0.7}{$\interp{A}\tensor \interp{B}$}} & \text{\scalebox{0.7}{$\interp{A}$}} & \text{\scalebox{0.7}{$\interp{B}$}} \\
	\begin{block}{c(ccc)}
	\text{\scalebox{0.7}{$\interp{A}$}}  &&\id&\\
	\text{\scalebox{0.7}{$\interp{B}$}}  &&&\id\\
	\end{block}
	\end{blockarray} \end{array}\]
	while on the right hand side we have:
	\[ \interp{\Psi_{A \parallel B}}_{\bfup{H}} = 
	\begin{blockarray}{cccc}
	& \text{\scalebox{0.7}{$\interp{A}\tensor \interp{B}$}} & \text{\scalebox{0.7}{$\interp{A}$}} & \text{\scalebox{0.7}{$\interp{B}$}} \\
	\begin{block}{c(ccc)}
	 \text{\scalebox{0.7}{$\interp{A}\tensor \interp{B}$}} &\id&&\\
	\text{\scalebox{0.7}{$\interp{A}$}}  &&\id&\\
	\text{\scalebox{0.7}{$\interp{B}$}}  &&&\id\\
	\end{block}
	\end{blockarray}\]
	\[ \interp{\Psi^{-1}_{A \oplus B}}_{\bfup{H}} = \interp{\id}_{\bfup{H}} = \id\]
	\[ \interp{G}_{\bfup{H}} =  
	\begin{blockarray}{cccc}
	& \text{\scalebox{0.7}{$\interp{A}\tensor \interp{B}$}} & \text{\scalebox{0.7}{$\interp{A}$}} & \text{\scalebox{0.7}{$\interp{B}$}} \\
	\begin{block}{c(ccc)}
	\text{\scalebox{0.7}{$\interp{A}$}}  &&\id&\\
	\text{\scalebox{0.7}{$\interp{B}$}}  &&&\id\\
	\end{block}
	\end{blockarray}  \]
	hence:
	\[ \interp{\Psi^{-1}_{A \oplus B} \circ G \circ \Psi_{A \parallel B}}_{\bfup{H}} = \begin{blockarray}{cccc}
	& \text{\scalebox{0.7}{$\interp{A}\tensor \interp{B}$}} & \text{\scalebox{0.7}{$\interp{A}$}} & \text{\scalebox{0.7}{$\interp{B}$}} \\
	\begin{block}{c(ccc)}
	\text{\scalebox{0.7}{$\interp{A}$}}  &&\id&\\
	\text{\scalebox{0.7}{$\interp{B}$}}  &&&\id\\
	\end{block}
	\end{blockarray} = \interp{\interp{G}_{\circlearrowleft}}_{\bfup{H}} \]
	Then using completeness for $\interp{-}_{\bfup{H}}$ we have $\interp{G}_{\circlearrowleft} = \Psi^{-1}_{A \oplus B} \circ G \circ \Psi_{A \parallel B}$.
	
	It remains the two inductive cases. For the sequential composition, we have 
	\[ \begin{array}{rcl} \interp{e \circ d}_{\circlearrowleft} &=& \interp{e}_{\circlearrowleft} \circ \interp{d}_{\circlearrowleft} \\ &=&  \Psi^{-1}_{\wireSetC} \circ e\circ \Psi_{\wireSetB} \circ \Psi^{-1}_{\wireSetB} \circ d \circ \Psi_{\wireSetA} \\&=& \Psi^{-1}_{\wireSetC} \circ e \circ d \circ \Psi_{\wireSetA} \end{array} \]
	and for the parallel composition we have
	\[ \begin{array}{rcl} \interp{d \parallel e}_{\circlearrowleft} &=& m^| \circ (\interp{d}_{\circlearrowleft} | \interp{e}_{\circlearrowleft}) \circ m^| \\ &=&   m^| \circ (( \Psi^{-1}_{\wireSetB} \circ d \circ \Psi_{\wireSetA}) | ( \Psi^{-1}_{\wireSetD} \circ e \circ \Psi_{\wireSetC})) \circ m^| \\&=&  \Psi^{-1}_{\wireSetB||\wireSetD} \circ (d \mid e) \circ \Psi_{\wireSetA||\wireSetC} \end{array} \]
	where $\Psi_{\wireSetA||\wireSetC} = (\Psi_{\wireSetA}|\Psi_{\wireSetC})\circ m^|$ is obtained by computing their semantics by $\interp{-}_{\bfup{H}}$, remarking that they are equal, and using  the completeness result of \Cref{thm:complete_fun}.

\subsection{The Whole Calculus}

In \Cref{thm:internal_language}, we extended the result to $\Cat$. To prove that, we match every diagram $d \in \sCat(A,B)$ to a diagram $\IsoPlusUn{d} \in \sFCat(A\oplus \tone, B \oplus \tone)$. 

For that, we remind that $\Iso_A : A \to \bigoplus \tone$ is defined as
 \[\tikzfig{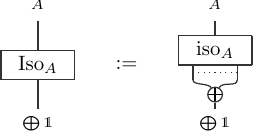}\]
and is an isomorphism as proven in \Cref{lem:Iso_is_iso}.

Rewriting the normal form of \Cref{thm:nf} with it, there exists a matrix $M$ such that $d$ can be written as below. Let $M_{+1}$ be "$M$ with +1 added to the bottom right corner's coefficient". We define $\IsoPlusUn{d}$ from that $M$ as:

\[\tikzfig{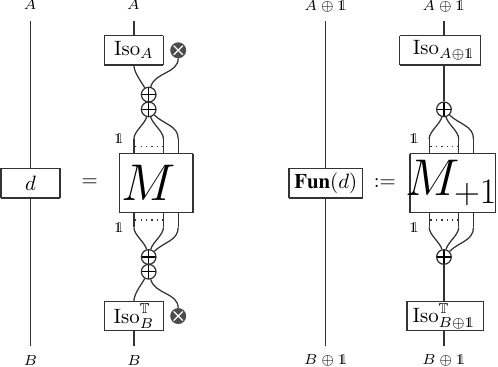}\]

We note that $\IsoPlusUn{d}$ is already in the normal form given by \Cref{thm:nf_fun}. To prove an equivalence of category between $\sCat_\equiv$ and $(\sFCat_\equiv)^{\oplus \tone}$, we check that $\IsoPlusUn{-}$ is a well defined full and faithful functor that is essentially surjective.

\paragraph{Well Defined} We need to check the soundness of $\IsoPlusUn{-}$ with respect to the equational theory. This follows from the uniqueness of the normal form, that is \Cref{cor:nf_unique}.
\paragraph{Functor} We consider $d \in \sCat(A,B)$ and $e \in \sCat(B,C)$. We write $M$ and $N$ their respective matrices when put in normal form, and $P$ the matrix of the normal form of $e \circ d$. From \Cref{prop:sem_block_matrix_unit} it follows that $P_{+1}$ is the result of the matrix product of $M_{+1}$ by $N_{+1}$. So looking at $\IsoPlusUn{e} \circ \IsoPlusUn{d}$, we can use \Cref{lem:Iso_is_iso} to eliminate the $\Iso$ at the middle, then \Cref{prop:mat_product} to merge $M_{+1}$ and $N_{+1}$ into their product $P_{+1}$, then it follows that $\IsoPlusUn{e} \circ \IsoPlusUn{d} = \IsoPlusUn{e \circ d}$.
\paragraph{Full} We consider $f \in(\sFCat_\equiv)^{\oplus \tone}(A,B)$. In particular, $f = g \bfup{~+~} \iota_r \circ \pi_r$ for some $g \in \sFCat_\equiv(A \oplus \tone, B \oplus \tone)$. Putting $g$ in normal form, there exists a matrix $M$ such that we can rewrite $f$ as:
\[ f = \tikzfig{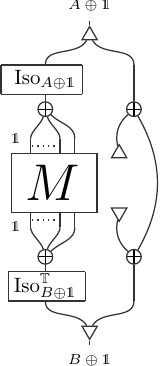} \]
Then, using the following rewriting:
\[ \tikzfig{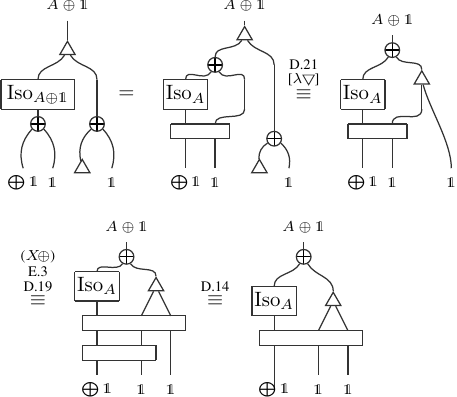}\]
We obtain 
\[ f = \tikzfig{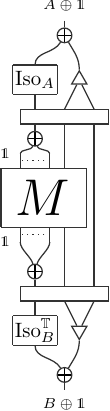} \]
Unfolding the definition of $M$, and relying on \Cref{lem:nat_disj,lem:idempotence_disj} to move around the disjunctions of wire, and \eqS{} and \eqSsum{} to "add" the rightmost wire to the leftmost coefficient of $M$, we obtain:
\[ f = \tikzfig{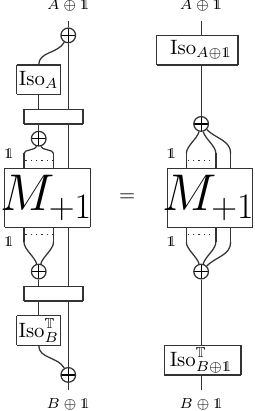} \]
Hence $f$ is of the form $\IsoPlusUn{d}$ for some $d$.
\paragraph{Faithful} In \Cref{thm:nf}, the bottom-right coefficient of $M$ is expected to be canonical. As such, the operation $M \mapsto M_{+1}$ is injective. So two diagrams $d \not\equiv e$ leads to two distinct matrices $M,N$, hence two distinct matrices $M_{+1},N_{+1}$, then using uniqueness of the normal form \Cref{cor:nf_unique_fun} it leads to two distinct diagrams $\IsoPlusUn{d}\not\equiv\IsoPlusUn{e}$.
\paragraph{Essentially Surjective} The functor $\IsoPlusUn{-}$ is surjective on the objects, hence essentially surjective.

\end{document}